\documentclass{article}
\usepackage[margin=1in]{geometry}
\usepackage{amsthm,amsmath,amssymb}
\usepackage{graphicx}
\usepackage{natbib}
\usepackage[hidelinks,breaklinks]{hyperref}
\newtheorem{theorem}{Theorem}[]
\newtheorem{lemma}[theorem]{Lemma}

\newtheorem{remark1}[theorem]{Remark}
\newenvironment{remark}{\begin{remark1} \rm}{\end{remark1}}

\DeclareMathOperator{\nint}{nint}
\DeclareMathOperator*{\argmin}{argmin}
\DeclareMathOperator{\bigoh}{\mathcal{O}}
\def\Id{\rm Id}

\title{An efficient algorithm for integer lattice reduction}
\author{Fran\c{c}ois Charton, Kristin Lauter, Cathy Li, and Mark Tygert}

\begin{document}

\maketitle

\begin{abstract}
A lattice of integers is the collection of all linear combinations
of a set of vectors for which all entries of the vectors are integers
and all coefficients in the linear combinations are also integers.
Lattice reduction refers to the problem of finding a set of vectors
in a given lattice such that the collection of all integer linear combinations
of this subset is still the entire original lattice and so that
the Euclidean norms of the subset are reduced.
The present paper proposes simple, efficient iterations for lattice reduction
which are guaranteed to reduce the Euclidean norms of the basis vectors
(the vectors in the subset) monotonically during every iteration.
Each iteration selects the basis vector for which projecting off
(with integer coefficients) the components of the other basis vectors
along the selected vector minimizes the Euclidean norms
of the reduced basis vectors.
Each iteration projects off the components along the selected basis vector
and efficiently updates all information required for the next iteration
to select its best basis vector and perform the associated projections.
\end{abstract}

\section{Introduction}
\label{intro}

Lattices of integers are common tools in cryptography and number theory,
among other areas, as reviewed by~\cite{cassels}, \cite{peikart}, and others.
A lattice of integers in $m$ dimensions consists
of all linear combinations of a set of $n$ vectors,
with all coefficients in the linear combinations being integers
and all $m$ entries of each of the $n$ vectors being integers.
Thus, the lattice is a collection of infinitely many vectors of integers.

The set of $n$ vectors whose linear combinations form the lattice
is known as a ``basis'' for the lattice.
Traditionally the vectors in the basis are required to be linearly independent,
but throughout the present paper, the term ``basis'' will abuse terminology
slightly in omitting any requirement to be linearly independent.
The term ``reduced basis'' refers
to any unimodular integer linear transformation of a basis.
(A unimodular integer linear transformation is multiplication
with an invertible $n \times n$ matrix such that all entries of the matrix
and its inverse are integers; the entries of the inverse of a matrix
of integers are integers if and only if the absolute value
of the determinant of the matrix is 1.)
The goal of lattice reduction (the topic of the present paper)
is to minimize the Euclidean norms of the vectors in the reduced basis.

Finding a nonzero vector in the lattice whose Euclidean norm is least
(or nearly least) is a common use of lattice reduction. To see the connection,
fix any positive real number $p$.
Any reduced basis which minimizes the sum of the $p$-th powers
of the Euclidean norms of the vectors in the reduced basis
must include a shortest vector as one of the basis vectors.
Indeed, if no basis vector is the shortest vector,
then the sum of the $p$-th powers of the Euclidean norms
can be made smaller by replacing with the shortest vector
one of the basis vectors whose projection on the shortest vector is nonzero.
``Projection'' is defined as follows.

The projection of a column vector $v$ onto a column vector $w$ is $c \cdot w$,
with the scalar coefficient $c$ defined to be the real number
\begin{equation}
c = \frac{v^\top w}{w^\top w},
\end{equation}
where $v^\top$ denotes the transpose of $v$
(and $w^\top$ denotes the transpose of $w$).
Needless to say, $v^\top w$ is the inner product between $v$ and $w$,
while $w^\top w$ is the inner product of $w$ with itself, which is of course
the square of the Euclidean norm of $w$.
Unfortunately, even if all entries of both $v$ and $w$ are integers,
subtracting off this projection $c \cdot w$ from $v$ directly
would result in a vector, $v - c \cdot w$, whose entries may not be integers.
Fortunately, all entries of the vector $v - \nint(c) \cdot w$
would remain integers, and Subsection~\ref{theory} below proves that
the Euclidean norm of $v - \nint(c) \cdot w$ is less than or equal to
the Euclidean norm of $v$. (Here and throughout the present paper,
$\nint(c)$ denotes the result of rounding $c$ to the nearest integer.)
In fact, if both $v$ and $w$ are in a lattice of integers,
then $v - \nint(c) \cdot w$ is in the same lattice of integers
and is shorter (or at least no longer) than $v$.
Theorem~\ref{theorem} below elaborates.

Let us denote by $a^0_1$, $a^0_2$, \dots, $a^0_n$ the initial basis vectors,
each being a column vector of $m$ integers.
The present paper proposes an iterative algorithm that reduces the basis
from iteration $i$ to iteration $i+1$ from $a^i_1$, $a^i_2$, \dots, $a^i_n$
to $a^{i+1}_1$, $a^{i+1}_2$, \dots, $a^{i+1}_n$, starting with $i=0$.
Each iteration $i$ of the algorithm selects an index $k$
that minimizes the Euclidean norms $\| a^i_j - c^i_{j,k} \cdot a^i_k \|$
for all $j \ne k$, where the scalar projection coefficient $c^i_{j,k}$ is
\begin{equation}
\label{rounded}
c^i_{j,k} = \nint\left( \frac{(a^i_j)^\top a^i_k}{(a^i_k)^\top (a^i_k)} \right)
\end{equation}
for $j \ne k$ and
\begin{equation}
c^i_{j,k} = 0
\end{equation}
for $j = k$.
Specifically, the index $k$ is chosen to minimize the sum of the $p$-th powers
of the Euclidean norms, that is, $k$ minimizes
$\sum_{j = 1}^n \| a^i_j - c^i_{j,k} \cdot a^i_k \|^p$.
(If $p = 2$, then this sum is the square of the Frobenius norm of the matrix.)
Iteration $i$ constructs the matrix for the next iteration as
$a^{i+1}_j = a^i_j - c^i_{j,k} \cdot a^i_k$.

{\it Prima facie}, considering all possible indices in order to minimize
the sum of the $p$-th powers appears to be computationally expensive.
However, the sum being minimized can be evaluated efficiently
via the Gram matrix whose entries are $g^i_{j,k} = (a^i_j)^\top a^i_k$,
as can the projection coefficients $c^i_{j,k}$ defined in~(\ref{rounded}).
Moreover, updating the entries of the Gram matrix from one iteration
to the next is also computationally efficient.

The algorithm of the present paper can complement others for lattice reduction,
such as the ``LLL'' algorithm introduced by~\cite{lenstra-lenstra-lovasz}.
Two entire books about the LLL algorithm are those
of~\cite{nguyen-vallee} and~\cite{bremner}.
The numerical examples presented below report the results
of running the algorithm of~\cite{lenstra-lenstra-lovasz}
in conjunction with the algorithm of the present paper.
The implementation of the LLL algorithm used here is
significantly more efficient (albeit less robust to worst-case,
adversarial examples devised for applications outside cryptography)
than the others' reviewed by~\cite{schnorr-euchner} and~\cite{stehle};
see Subsection~\ref{implementation} below
for details.\footnote{Permissively licensed open-source codes implementing
both the algorithm of the present paper and the classical LLL algorithm
of~\cite{lenstra-lenstra-lovasz} are available
at \url{https://github.com/facebookresearch/latticer}}

The primary purpose of the algorithm of the present paper is to polish
the outputs of other algorithms for lattice reduction;
on its own, the algorithm of the present paper tends to get stuck
in local minima, with the iterations attaining an equilibrium
that is far away from optimally minimizing the sum of the $p$-th powers.
The convergence is monotonic, but need not arrive at the optimal minimum
when fully converged.

The remainder of the present paper has the following structure:
Section~\ref{methods} describes and analyzes the algorithm in detail. 
Section~\ref{results} illustrates the algorithm via several numerical examples,
comparing and combining the proposed algorithm with the classic method
of~\cite{lenstra-lenstra-lovasz}.
Section~\ref{conclusion} draws some conclusions.
Appendix~\ref{poor} sketches several seemingly natural alternatives
that performed poorly in numerical experiments;
this appendix also points to other authors' variations of the LLL algorithm.
Appendix~\ref{further} supplements the figures of Section~\ref{results}
with further figures.

\section{Methods}
\label{methods}

This section elaborates the algorithm of the present paper.
First, Subsection~\ref{algorithm} details the algorithm and estimates
its computational costs.
Then, Subsection~\ref{theory} proves that the algorithm
converges monotonically, strictly reducing the sum of the $p$-th powers
of the Euclidean norms of the basis vectors during every iteration
(and hence halting after a finite number of iterations).

\subsection{Algorithm}
\label{algorithm}

This subsection describes the algorithm of the present paper in detail.

Consider the $n$ column vectors $a^0_1$, $a^0_2$, \dots, $a^0_n$,
each of size $m \times 1$.
The proposed scheme makes no assumption on the relative sizes of $m$ and $n$
--- $m$ can be less than, equal to, or greater than $n$.
Calculate the entries of the symmetric square Gram matrix 
\begin{equation}
\label{Gram0}
g^0_{j,k} = (a^0_j)^\top (a^0_k)
\end{equation}
for $j = 1$, $2$, \dots, $n$, and $k = 1$, $2$, \dots, $n$,
where $(a^0_j)^\top$ denotes the transpose of $a^0_j$.
(Of course, $(a^0_j)^\top (a^0_k)$ is simply the inner product
between $a^0_j$ and $a^0_k$.)
This costs $\bigoh(mn^2)$ operations.

Iterations, $i = 0$, $1$, $2$, \dots, will maintain the relation
\begin{equation}
\label{Gram}
g^i_{j,k} = (a^i_j)^\top (a^i_k)
\end{equation}
for $j = 1$, $2$, \dots, $n$, and $k = 1$, $2$, \dots, $n$.

Now repeat all of the following steps, again and again,
moving from $i = 0$ to $i = 1$ to $i = 2$ and so on:

Calculate the projection coefficients
\begin{equation}
c^i_{k,k} = 0
\end{equation}
for $k = 1$, $2$, \dots, $n$, and
\begin{equation}
\label{coeffs}
c^i_{j,k} = \nint\left( \frac{(a^i_j)^\top a^i_k}
                             {(a^i_k)^\top (a^i_k)} \right)
          = \nint\left( \frac{g^i_{j,k}}{g^i_{k,k}} \right)
\end{equation}
for $j = 1$, $2$, \dots, $n$, and $k = 1$, $2$, \dots, $n$ with $k \ne j$,
where $\nint$ denotes the nearest integer.
This costs $\bigoh(n^2)$ operations.

The sum of the $p$-th powers of the Euclidean norms
when projecting off the $k$-th vector is
\begin{equation}
\label{sos}
s^i_k = \sum_{j=1}^n \| a^i_j - c^i_{j,k} \, a^i_k \|^p
= \sum_{j=1}^n \left( g^i_{j,j} + (c^i_{j,k})^2 \, g^i_{k,k}
- 2 c^i_{j,k} \, g^i_{j,k} \right)^{p/2}
\end{equation}
for $k = 1$, $2$, \dots, $n$.
The iterations are likely to work better by starting with $p = 2$
and only later running with, say, $p = 1$.
This costs $\bigoh(n^2)$ to compute.

Define $\tilde{\imath}$ to be the index such that $s^i_{\tilde{\imath}}$
is minimal, that is,
\begin{equation}
\tilde{\imath} = \argmin_{1 \le k \le n} s^i_k.
\end{equation}
This costs $\bigoh(n)$ to calculate.

Project off the $\tilde{\imath}$-th vector to update every vector
\begin{equation}
\label{ortho}
(a^{i+1}_k)_j
= (a^i_k)_j - c^i_{k,\tilde{\imath}} \, (a^i_{\tilde{\imath}})_j
\end{equation}
for $j = 1$, $2$, \dots, $m$, and $k = 1$, $2$, \dots, $n$,
where $(a^i_k)_j$ denotes the $j$-th entry of $a^i_k$.
This costs $\bigoh(mn)$.

Update the Gram matrix via the relation
\begin{equation}
\label{Gram_update}
g^{i+1}_{j,k}
= (a^{i+1}_j)^\top (a^{i+1}_k)
= (a^i_j - c^i_{j,\tilde{\imath}} \, a^i_{\tilde{\imath}})^\top
  (a^i_k - c^i_{k,\tilde{\imath}} \, a^i_{\tilde{\imath}})
= g^i_{j,k}
+ c^i_{j,\tilde{\imath}} \, c^i_{k,\tilde{\imath}}
\, g^i_{\tilde{\imath},\tilde{\imath}}
- c^i_{j,\tilde{\imath}} \, g^i_{\tilde{\imath},k}
- c^i_{k,\tilde{\imath}} \, g^i_{j,\tilde{\imath}}
\end{equation}
for $j = 1$, $2$, \dots, $n$, and $k = 1$, $2$, \dots, $n$.
This costs $\bigoh(n^2)$.

The total cost per iteration is $\bigoh(mn+n^2)$ operations.
Notice that only the precomputation of the Gram matrix in~(\ref{Gram0})
and~(\ref{ortho}) explicitly involve the individual entries of the vectors;
all other steps in the iterations involve only the entries of the Gram matrix.

\begin{remark}
Another possibility is to replace the sum in~(\ref{sos})
with a maximum, replacing~(\ref{sos}) with
\begin{equation}
s^i_k = \max_{1 \le j \le n} \| a^i_j - c^i_{j,k} \, a^i_k \|^2
= \max_{1 \le j \le n} \left( g^i_{j,j} + (c^i_{j,k})^2 \, g^i_{k,k}
- 2 c^i_{j,k} \, g^i_{j,k} \right)
\end{equation}
for $k = 1$, $2$, \dots, $n$.
However, using the maximum may fail to force any but the longest vectors
in the reduced basis to become shorter.
\end{remark}

\begin{remark}
The Gram matrix is symmetric, that is, $g^i_{j,k} = g^i_{k,j}$
for all $i = 0$, $1$, $2$, \dots,
for $j = 1$, $2$, \dots, $n$, and $k = 1$, $2$, \dots, $n$.
Calculating $g^i_{j,k}$ for only $j \le k$ suffices to fill the entire matrix
for all $j = 1$, $2$, \dots, $n$, and $k = 1$, $2$, \dots, $n$.
This can save computational costs in~(\ref{Gram0}) and~(\ref{Gram_update}).
\end{remark}

\begin{remark}
Cryptography can benefit from reductions to collections of basis vectors
that include all the unit basis vectors, each multiplied by the prime order
of a finite field, in addition to the other basis vectors.
This essentially formalizes the concept of lattice reduction
over the finite field. 
Instead of working directly on a collection of basis vectors
$\left( \begin{array}{c|c} A & q \cdot \Id \end{array} \right)$
in this way, where the columns of $A$ form the initial collection
of basis vectors, $q$ is the order of the finite field,
and $\Id$ is the identity matrix, the classical methods for lattice reduction
--- such as the LLL algorithm of~\cite{lenstra-lenstra-lovasz} ---
must operate instead on 
\begin{equation}
\left( \begin{array}{c|c} A   & q \cdot \Id \\\hline
                          \Id & 0           \end{array} \right),
\end{equation}
effectively doubling the dimension of the basis vectors.
\end{remark}

\subsection{Theory}
\label{theory}

The purpose of this subsection is to state and prove Theorem~\ref{theorem},
elaborating the facts stated in the fourth paragraph of the introduction,
Section~\ref{intro}.

The following lemma is helpful in the proof of the subsequent lemma.
\begin{lemma}
Suppose that $r$ is a real number. Then
\begin{equation}
\label{simplified}
(\nint(r))^2 - 2 \nint(r) \, r \le 0.
\end{equation}
\end{lemma}

\begin{proof}
Since the sign of $\nint(r)$ is the same as the sign of $r$,
the sign of the left-hand side of~(\ref{simplified}) is the sign of $r$
times the sign of
\begin{equation}
\label{final}
\nint(r) - 2 r.
\end{equation}
Now, the sign of~(\ref{final}) is opposite to the sign of $r$:
if $r > 1/2$, then $\nint(r) - 2r \le r+1/2-2r = 1/2-r < 0$;
if $r < -1/2$, then $\nint(r) - 2r \ge r-1/2-2r = -1/2-r > 0$;
if $|r| < 1/2$, then $\nint(r) - 2r = -2r$, whose sign is opposite to $r$'s.
Combining these observations yields~(\ref{simplified}).
\end{proof}

The following lemma is helpful in the proof of the subsequent theorem.
\begin{lemma}
Suppose that $c^i_{j,k}$ is the coefficient defined in~(\ref{coeffs})
and $g^i_{j,k}$ is the entry of the Gram matrix defined in~(\ref{Gram})
for $i = 0$, $1$, $2$, \dots, $j = 1$, $2$, \dots, $n$,
and $k = 1$, $2$, \dots, $n$. Then
\begin{equation}
\label{update}
(c^i_{j,k})^2 \, g^i_{k,k} - 2 c^i_{j,k} \, g^i_{j,k} \le 0
\end{equation}
for $i = 0$, $1$, $2$, \dots, $j = 1$, $2$, \dots, $n$,
and $k = 1$, $2$, \dots, $n$.
\end{lemma}

\begin{proof}
Combining~(\ref{coeffs}) and the fact that
$g^i_{k,k} = (a^i_k)^\top (a^i_k) \ge 0$
yields that the sign of the left-hand side of~(\ref{update})
is the same as the sign of
\begin{equation}
(c^i_{j,k})^2 - 2 c^i_{j,k} \, \frac{g^i_{j,k}}{g^i_{k,k}}
= (\nint(f^i_{j,k}))^2 - 2 \nint(f^i_{j,k}) \, f^i_{j,k},
\end{equation}
where
\begin{equation}
f^i_{j,k} = \frac{g^i_{j,k}}{g^i_{k,k}}.
\end{equation}
The lemma follows from~(\ref{simplified}) with $r = f^i_{j,k}$.
\end{proof}

The following theorem states that projection using the coefficients
defined in~(\ref{coeffs}) never increases the Euclidean norm.
This implies that the above algorithm never increases the Euclidean norm
of any of the basis vectors at any time; the sum of the $p$-th powers
of the Euclidean norms therefore converges monotonically to a local minimum,
for every possible (positive) value of $p$ simultaneously.
\begin{theorem}
\label{theorem}
Suppose that $c^i_{j,k}$ is the coefficient defined in~(\ref{coeffs})
for $i = 0$, $1$, $2$, \dots, $j = 1$, $2$, \dots, $n$,
and $k = 1$, $2$, \dots, $n$.
Then
\begin{equation}
\label{mono}
\| a^i_j - c^i_{j,k} \, a^i_k \| \le \| a^i_j \|
\end{equation}
for $i = 0$, $1$, $2$, \dots, $j = 1$, $2$, \dots, $n$,
and $k = 1$, $2$, \dots, $n$.
\end{theorem}

\begin{proof}
The square of the right-hand side of~(\ref{mono}) is
\begin{equation}
\label{rhs}
\| a^i_j \|^2 = g^i_{j,j}.
\end{equation}
The square of the left-hand side of~(\ref{mono}) is
\begin{equation}
\label{lhs}
\| a^i_j - c^i_{j,k} \, a^i_k \|^2
= g^i_{j,j} + (c^i_{j,k})^2 \, g^i_{k,k} - 2 c^i_{j,k} \, g^i_{j,k}.
\end{equation}
Formula~(\ref{update}) shows that~(\ref{lhs})
is less than or equal to~(\ref{rhs}).
\end{proof}

\section{Results and Discussion}
\label{results}

This section presents the results
of several numerical experiments.\footnote{Permissively licensed open-source
software that automatically reproduces all results reported
in the present paper is available at
\url{https://github.com/facebookresearch/latticer}}
Subsection~\ref{figures} describes the figures.
Subsection~\ref{examples} constructs the examples whose numerical results
the figures report.
Subsection~\ref{implementation} provides details of the implementation
and computer system used.
Subsection~\ref{discussion} discusses the empirical results.

\subsection{Description of the figures}
\label{figures}

This subsection describes empirical results, illustrating the algorithms
via Figures~\ref{p2time1-1e-15}--\ref{p2err1-1e-15-31}.
Supplementary figures are available in Appendix~\ref{further}.

The figures refer to the algorithm of~\cite{lenstra-lenstra-lovasz} as ``LLL.''
Subsection~\ref{implementation} details the especially efficient implementation
used here. The figures refer to the algorithm of the present paper as ``ours.''

The figures report on two kinds of experiments.
The first kind runs the LLL algorithm followed by ours,
once in each of ten trials. Each of the ten trials permutes the basis vectors
at random, with a different random permutation. The points plotted
in the figures are the means of the ten trials. The plotted bars range
from the associated minimum over the ten trials to the associated maximum
(not displaying the standard deviation that error bars sometimes report).
The figures refer to this single randomized permutation followed by LLL
followed by ours as ``run only once.''

The second kind of experiment runs the following sequence ten times
in succession: (1) random permutation of the basis vectors
followed by (2) LLL followed by (3) the algorithm of the present paper,
with each of the first nine times feeding its output into the input
of the next time. The figures refer to this sequence of ten
as ``run repeatedly.''

The figures report the reduction either in the Frobenius norm
of the matrix of basis vectors or in the minimum of the Euclidean norms
of the basis vectors. (The Frobenius norm is the square root of the sum
of the squares of the entries of the matrix.)
The figures report the reduction due to LLL run on its own,
as well as the further reduction due to post-processing
the output from LLL via the algorithm of the present paper.
Therefore, the full reduction in norm is the product
of the reported fractions remaining after reduction,
as all runs polish the outputs of the LLL algorithm via ours.
The fraction reduced by ours that is reported for the sequence run repeatedly
pertains only to the very last run of ours,
whereas the fraction reduced by LLL reported
for the same sequence run repeatedly pertains to the entire series of ten runs
of LLL followed by ours, aside from separating out the very last run of ours.

The horizontal axes report the dimension $n$ of the matrix
whose columns are the initial basis vectors being reduced
(the matrix is $n \times n$).
The figure captions' $\delta$ is that used
in the so-called ``Lov\'asz criterion'' of the LLL algorithm.
The figure captions' $p$ is the power used in~(\ref{sos})
for the algorithm of the present paper.
The figure captions' $q$ is an integer characterizing the size of the entries
of the initial basis vectors prior to reduction;
the following subsection provides further details about $q$.

Appendix~\ref{further} presents further figures,
reporting results for different values of $\delta$ and $p$.
The conclusions that can be drawn from the further figures
appear to be broadly consistent with those presented in the present section.

\subsection{Description of the examples}
\label{examples}

This subsection details the particular examples whose numerical results
the figures of the previous subsection report.

To construct the initial basis vectors being reduced
as related to a finite field whose order is a large prime integer $q$,
the examples consider $q = 2^{13} - 1$ and $q = 2^{31} - 1$
(two well-known Mersenne primes).
The matrix whose columns are the initial basis vectors for the examples
takes the form
\begin{equation}
\label{matexample}
A = \left( \begin{array}{c|c} R & q \cdot \Id \\\hline
                              \Id & 0 \end{array} \right),
\end{equation}
where ``$\Id$'' denotes the identity matrix,
``$0$'' denotes the matrix whose entries are all zeros,
and $R$ denotes the pseudorandom matrix whose entries are drawn
independent and identically distributed from the uniform distribution
over the integers $-(q-1)/2$, $-(q-3)/2$, \dots, $(q-3)/2$, $(q-1)/2$.
This special form of matrix effectively views the entries of $R$ modulo $q$,
as unimodular integer linear transformations acting on $A$ from the right
can add any integer multiple of $q$ to any entry of $R$.
The captions to the figures give the values of $q$ considered.

The dimension of the matrix $R$ in~(\ref{matexample}) is
$(2 \ell) \times \ell$, where $\ell = 2$, $4$, $8$, \dots, $128$.
Hence the dimension of the matrix $A$ in~(\ref{matexample}) whose columns
are the initial basis vectors is $n \times n$, with $n = 3 \ell$,
so that $n = 6$, $12$, $24$, \dots, $384$.
The horizontal axes of the figures give the values of $n$
associated with the corresponding points in the plots.

\subsection{Implementation details}
\label{implementation}

This subsection details the implementations used in the reported results.

The implementation of the LLL algorithm of~\cite{lenstra-lenstra-lovasz}
used here is entirely in IEEE standard double-precision arithmetic,
with acceleration via basic linear algebra subroutines (BLAS)
and enhanced accuracy via re-orthogonalization.
BLAS originates from~\cite{lawson-hanson-kincaid-krogh},
\cite{blas}, and others, with implementation in Apple's Accelerate Framework
of Xcode for MacOS Ventura 13.4 in the experiments reported here
(and compiled using Clang LLVM with the highest optimization flag,
{\tt -O3}, set). All experiments ran on a 2020 MacBook Pro
with a 2.3~GHz quad-core Intel Core i7 processor
and 3.733~GHz LPDDR4X SDRAM.
For the greatest portability, the implementation calls {\tt random()}
to generate random numbers; many compiler distributions
implement {\tt random()} via the BSD linear-congruential generator
(such compilers include our implementation's defaults
of MacOS Clang referenced by {\tt gcc} under Xcode
and GNU {\tt gcc} under Linux). The results reported here pertain
to this class of pseudorandom numbers.

Re-orthogonalization combats round-off errors,
as reviewed by~\cite{leon-bjorck-gander}.
Namely, all orthogonalization gets repeated until only at most the last digit
of any result changes, and then one additional round of orthogonalization
occurs (just to be sure).

Furthermore, upon any swapping of basis vectors in the LLL algorithm,
the implementation recomputes the swapped orthogonal basis vectors from scratch
for numerical stability. Recomputing from scratch incurs
extra computational expense, but no more than the computational cost incurred
by the size reduction which accompanies every swap.

The combination of re-orthogonalization and recomputing from scratch
swapped orthogonal basis vectors may not be sufficient to handle
the hardest possible examples constructed purposefully to test resistance
to round-off errors, but is sufficient to handle the random examples
of interest in cryptography (since cryptosystems necessarily
must generate bases at random in order to be secure).
Methods that take advantage of floating-point arithmetic
while also being able to tackle worst-case adversarial examples
are overviewed by~\cite{schnorr-euchner} and~\cite{stehle}.
For the more limited scope of cryptographic applications, however,
the enhanced accuracy of re-orthogonalization and recomputing from scratch
swapped orthogonal basis vectors enables all computations to leverage BLAS
for dramatic accelerations. The implementation of the algorithm
of the present paper also takes advantage of this acceleration,
primarily for calculating the initial Gram matrix
(other steps of the algorithm benefit less from BLAS than LLL does,
at least in the implementation reported here).

Figure~\ref{p2time1-1e-15} indicates that both LLL and the algorithm
of the present paper have running-times proportional to roughly $n^3$,
with LLL being over an order of magnitude slower.
The cost per iteration is proportional to $n^2$, as discussed
in Section~\ref{methods}; the numbers of iterations required appear
to be proportional to $n$ for the examples reported here.

\subsection{Discussion}
\label{discussion}

This subsection discusses the numerical results reported in the figures.

As mentioned in the last paragraph of the preceding subsection,
the cost per iteration of either the LLL algorithm or ours is proportional
to $n^2$. Figure~\ref{p2time1-1e-15} indicates that both LLL and ours
run in time proportional to roughly $n^3$
(with LLL being over an order of magnitude slower),
in accord with the numbers of iterations taken to reach equilibrium
being proportional to the dimension $n$ in the experiments
of the present section.
(Figure~\ref{p2time1-1e-1} shows the same, but for $\delta = 1 - 10^{-1}$
rather than $\delta = 1 - 10^{-15}$, where $\delta$ is the parameter
in the so-called ``Lov\'asz criterion'' of the LLL algorithm.)

Figures~\ref{p2err1-1e-15-13once}--\ref{p2err1-1e-15-31}
indicate that both the LLL algorithm and ours reduce norms
as a function of the dimension $n$ similarly for different sizes
of the entries of the basis vectors; specifically, changing $q = 2^{13} - 1$
to $q = 2^{31} - 1$ has little effect on how the results vary
as a function of $n$. Repeated runs that could in principle help jump out
of local minima appear little more effective than simply taking the best
of several single runs (with each run randomizing the order
of the basis vectors differently).

In all cases, LLL reduces norms far more than ours.
The algorithm of the present paper is suitable only for polishing the results
of another algorithm (such as LLL); the algorithm of the present paper
is very fast and often reduces norms beyond what other algorithms achieve,
but is ineffective on its own for anything other than demonstrating
that a given basis can be reduced further.
Ours appears to be more useful for reducing the Frobenius norm
than the minimum of the Euclidean norms of the basis vectors,
though the algorithm does reduce both (or at least never increases either,
as guaranteed by the theory of Subsection~\ref{theory}
and illustrated in all the figures).

Appendix~\ref{further} presents several more figures, with different settings
of parameters; all further corroborate the results discussed
in the present subsection.
And, naturally, the algorithm and its analysis generalize to lattices formed
from linear combinations of vectors whose entries are real numbers
(not necessarily integers), with the coefficients in the linear combinations
being integers. However, the case involving real numbers may pose challenges
due to round-off errors.

The dual of a lattice of integers is in general such a lattice
of real vectors. The least-possible maximum of the Euclidean norms
of the vectors in a basis for the dual lattice yields bounds
on the Euclidean norm of the shortest nonzero vector in the original lattice,
courtesy of the transference of~\cite{banaszczyk},
as highlighted by~\cite{regev} and others.
The algorithm of the present paper can reduce the maximum
of the Euclidean norms toward the least possible.

\begin{figure}
\begin{centering}
{\includegraphics[width=0.495\textwidth]{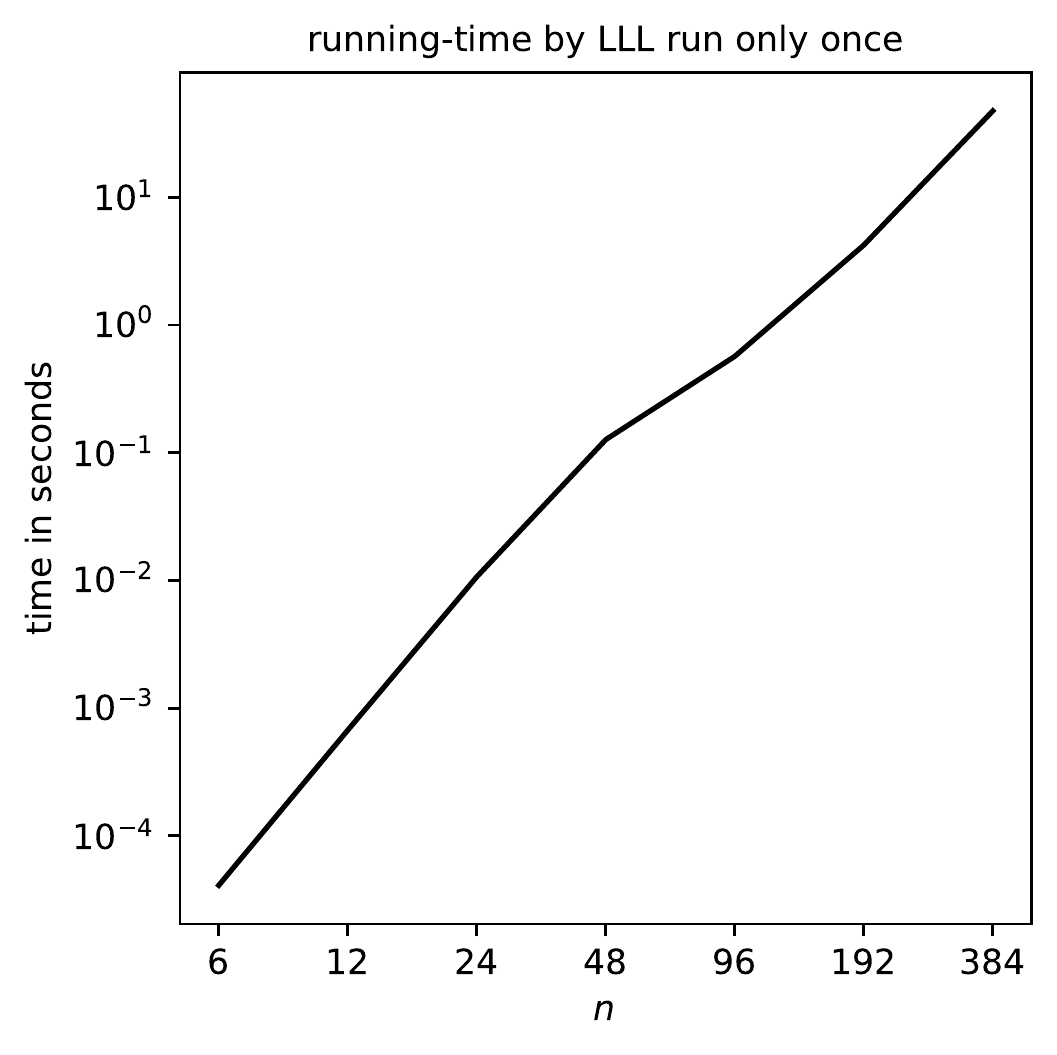}}
{\includegraphics[width=0.495\textwidth]{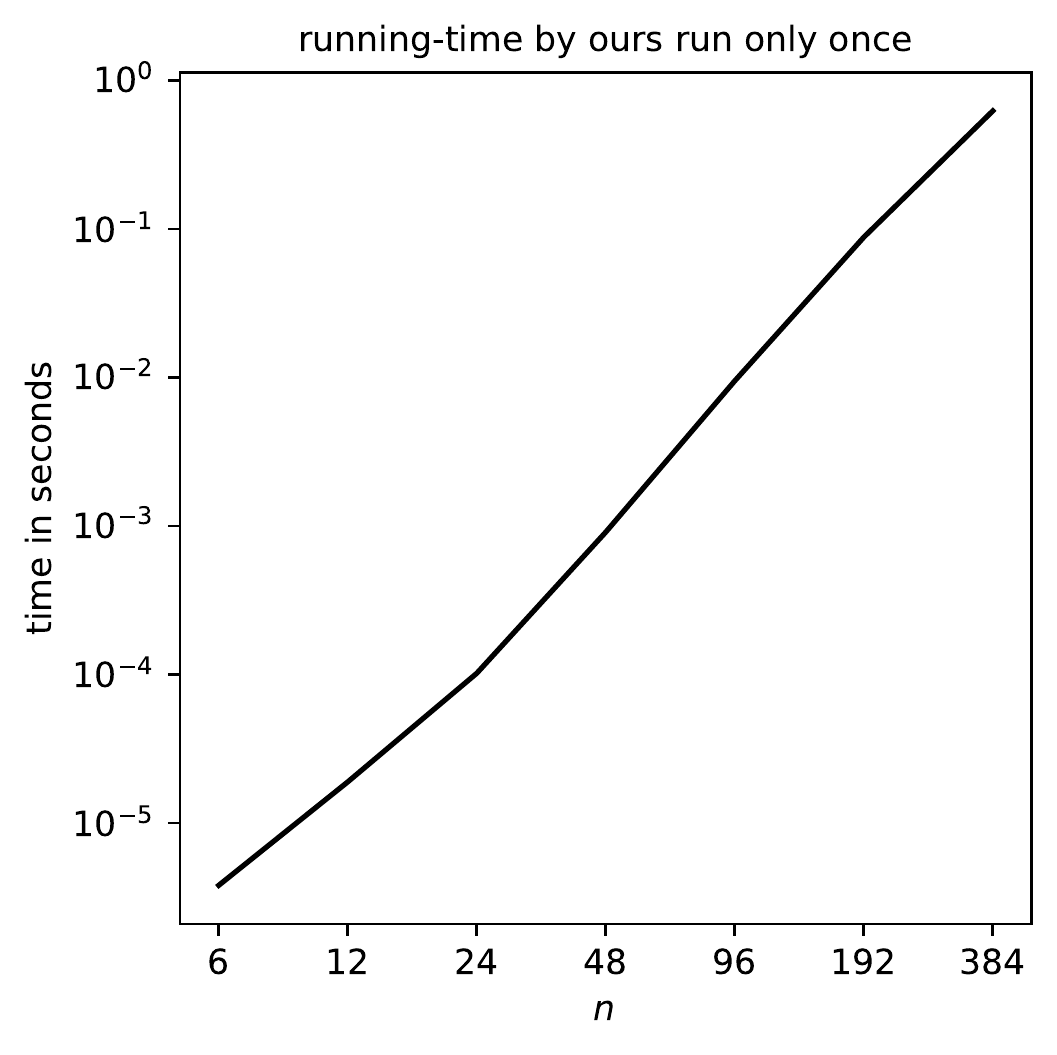}}

{\includegraphics[width=0.495\textwidth]{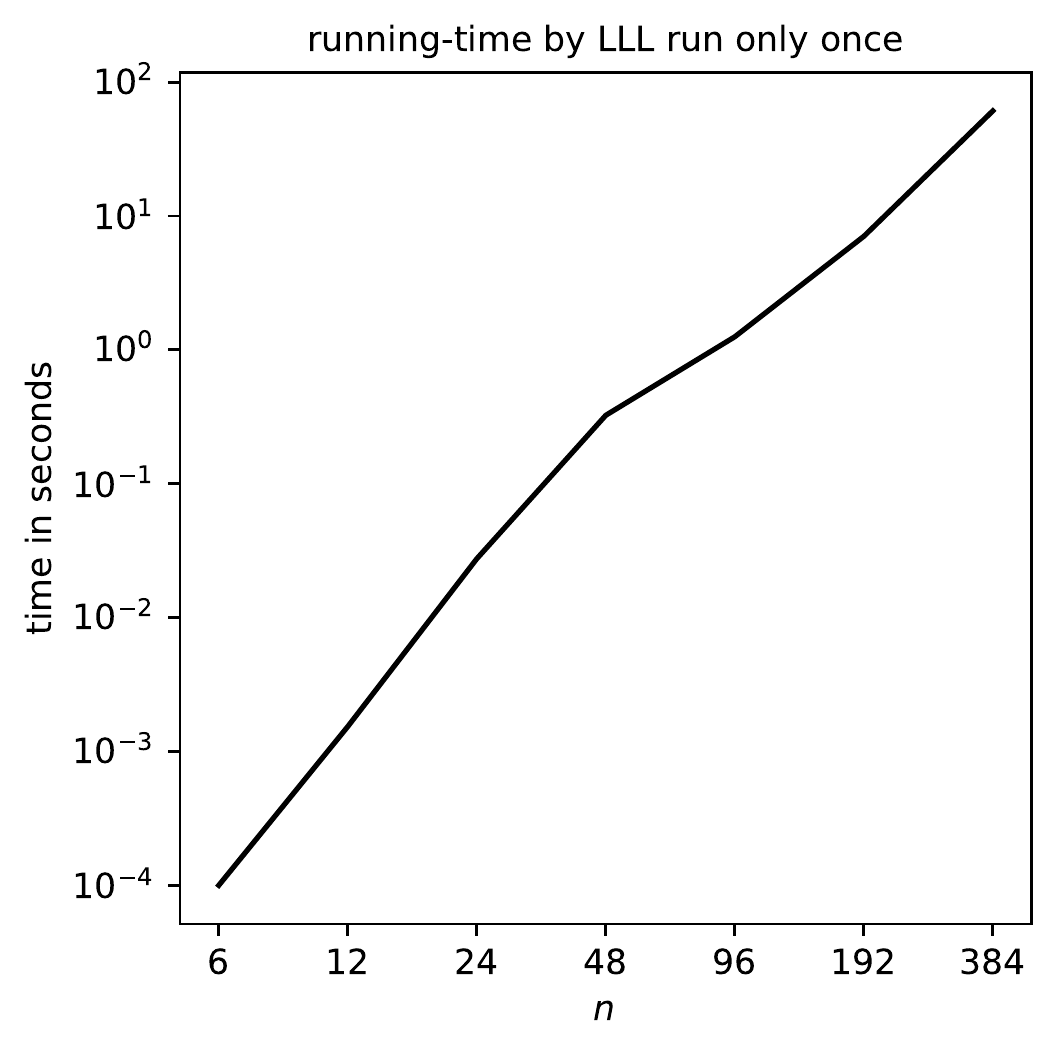}}
{\includegraphics[width=0.495\textwidth]{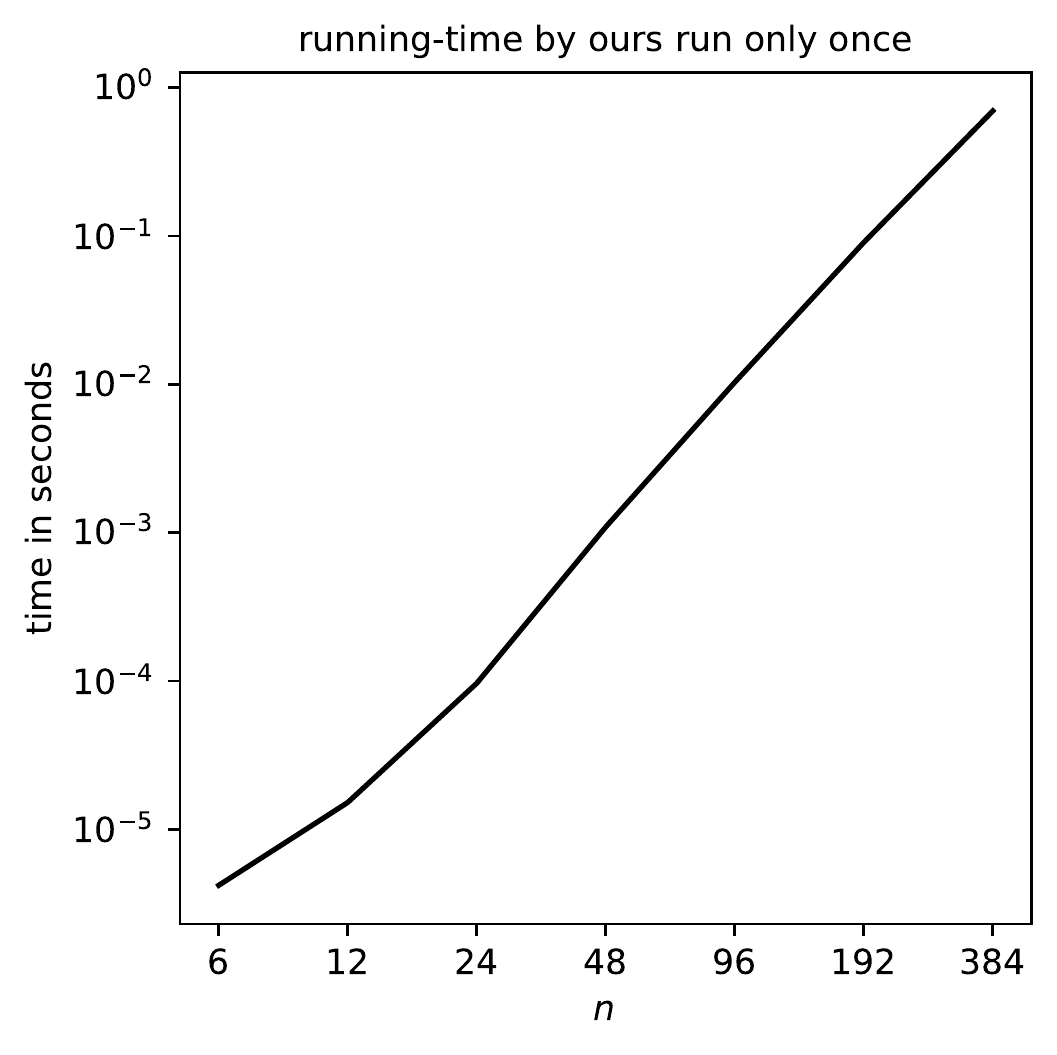}}

\end{centering}
\caption{$\delta = 1-10^{-15}$, $p = 2$;
         the upper plots are for $q = 2^{13} - 1$,
         the lower plots are for $q = 2^{31} - 1$}
\label{p2time1-1e-15}
\end{figure}

\begin{figure}
\begin{centering}
{\includegraphics[width=0.495\textwidth]{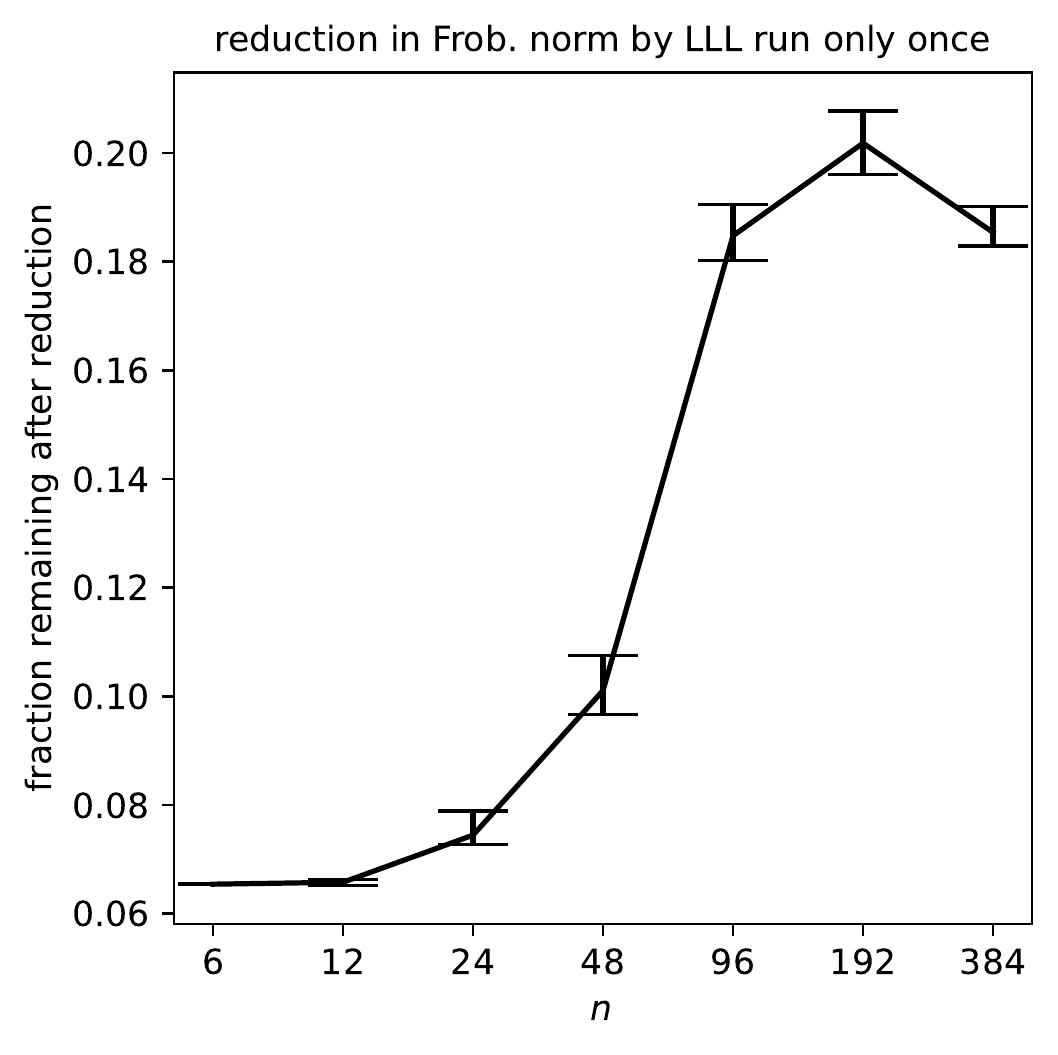}}
{\includegraphics[width=0.495\textwidth]{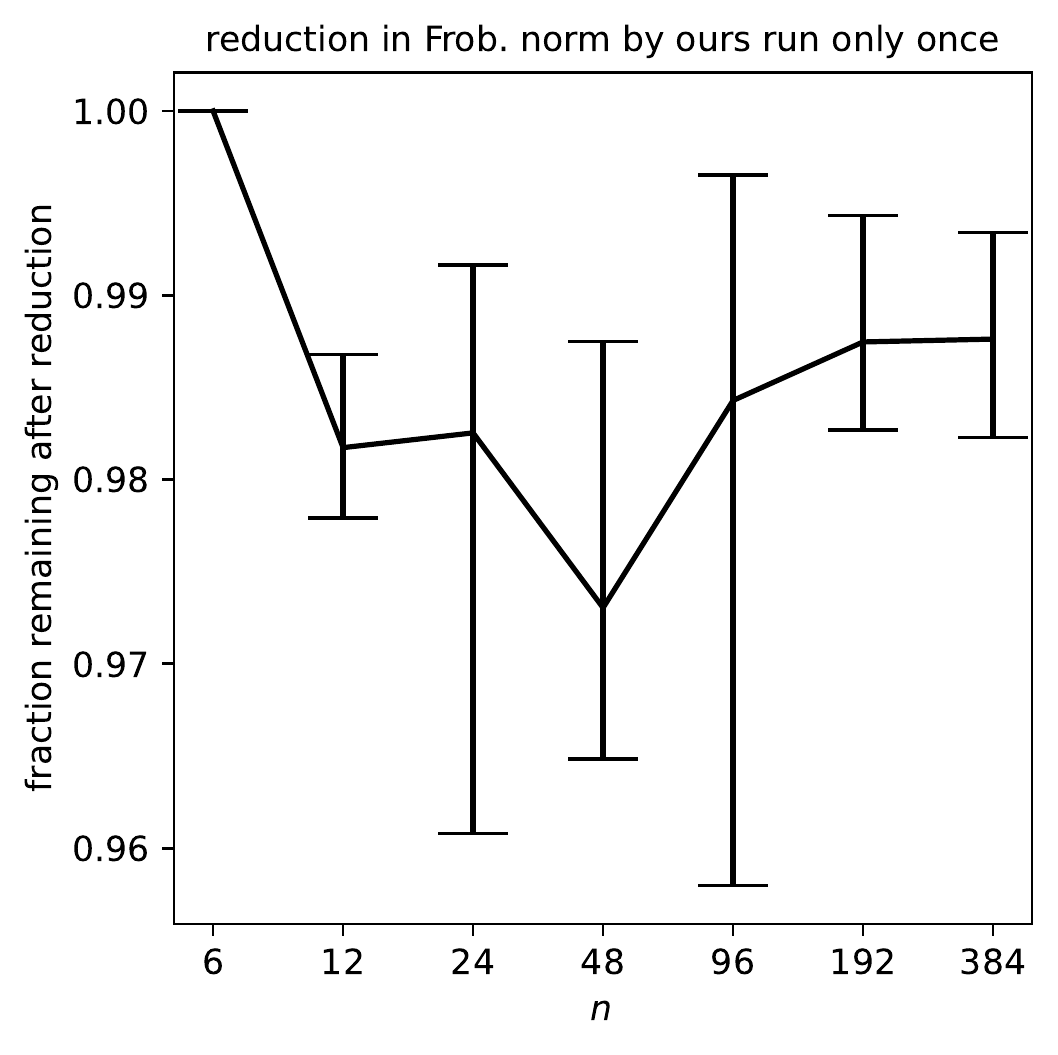}}

{\includegraphics[width=0.495\textwidth]{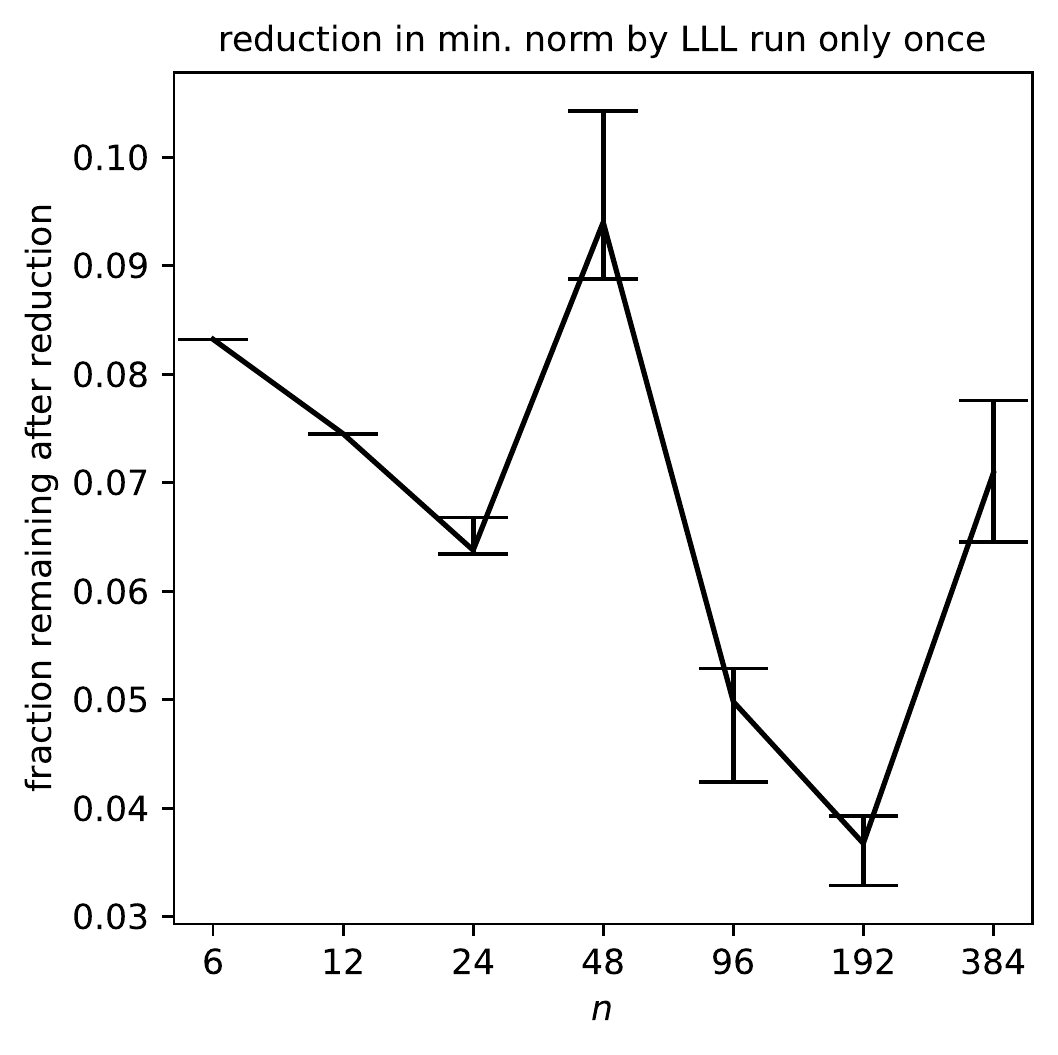}}
{\includegraphics[width=0.495\textwidth]{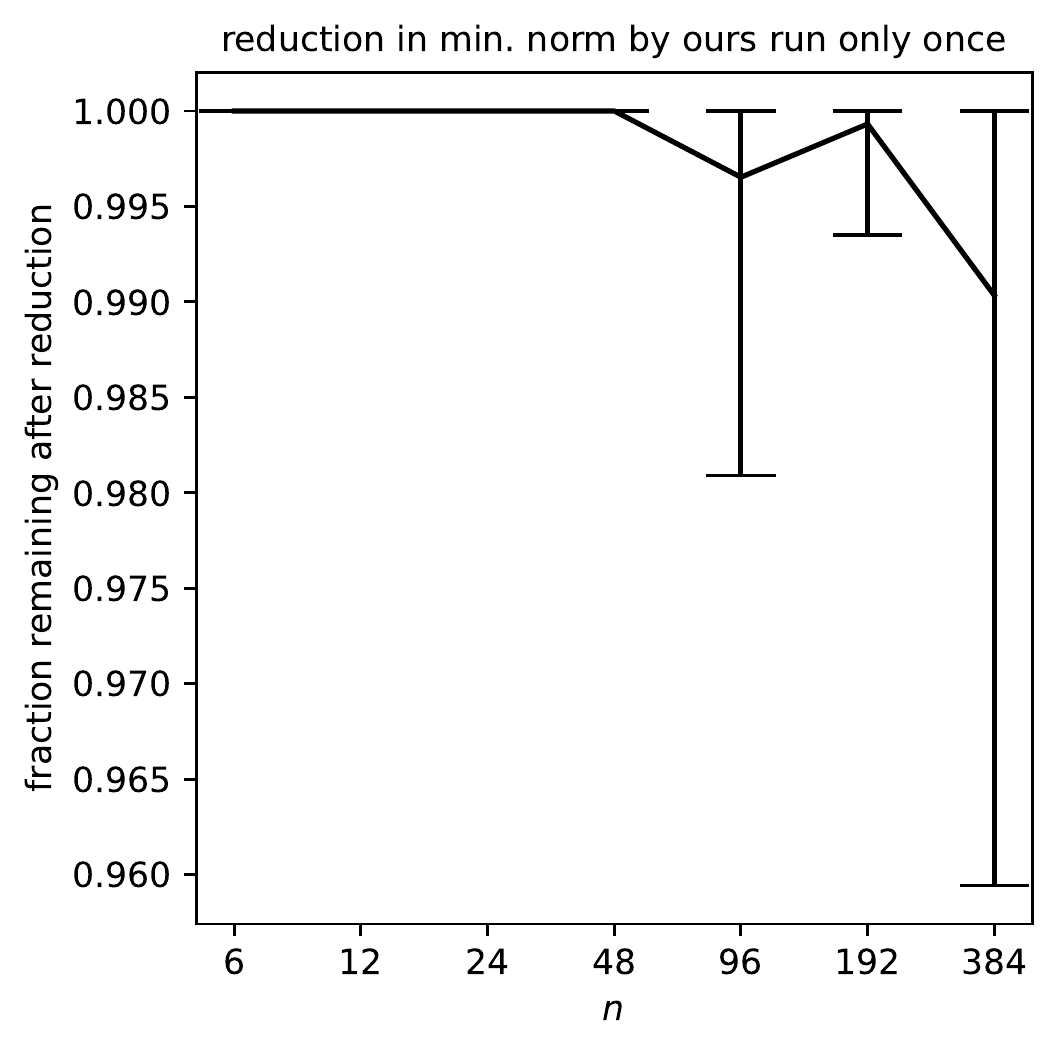}}

\end{centering}
\caption{$\delta = 1-10^{-15}$, $p = 2$, $q = 2^{13} - 1$}
\label{p2err1-1e-15-13once}
\end{figure}

\begin{figure}
\begin{centering}
{\includegraphics[width=0.495\textwidth]{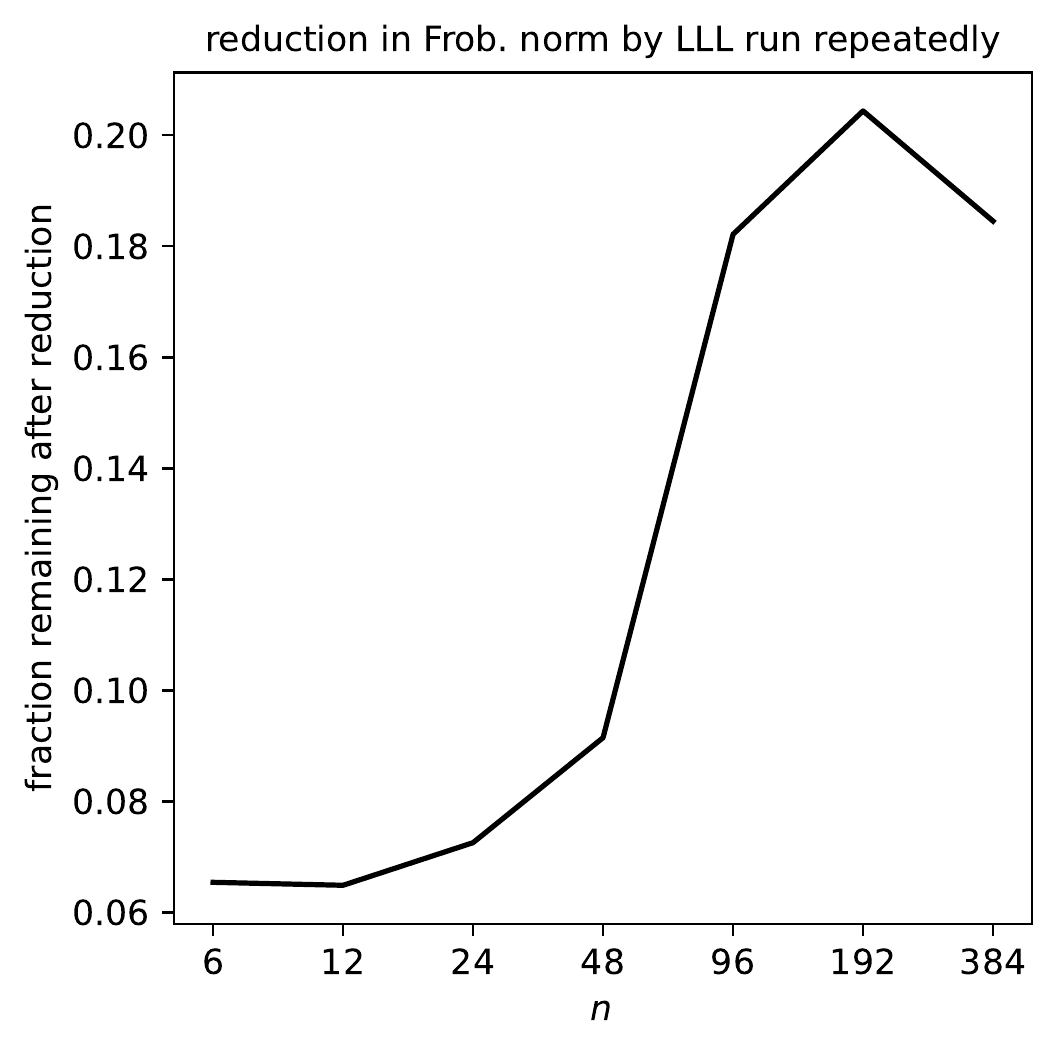}}
{\includegraphics[width=0.495\textwidth]{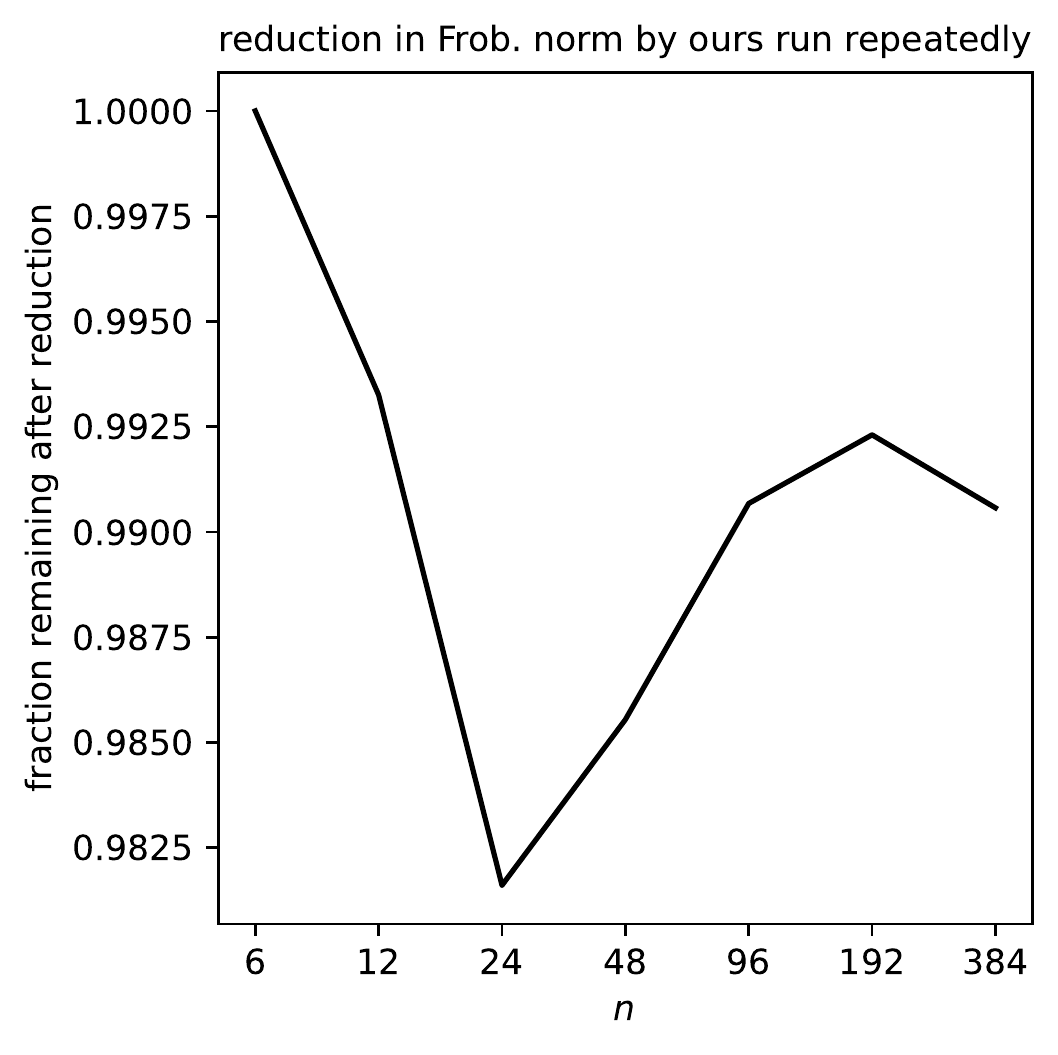}}

{\includegraphics[width=0.495\textwidth]{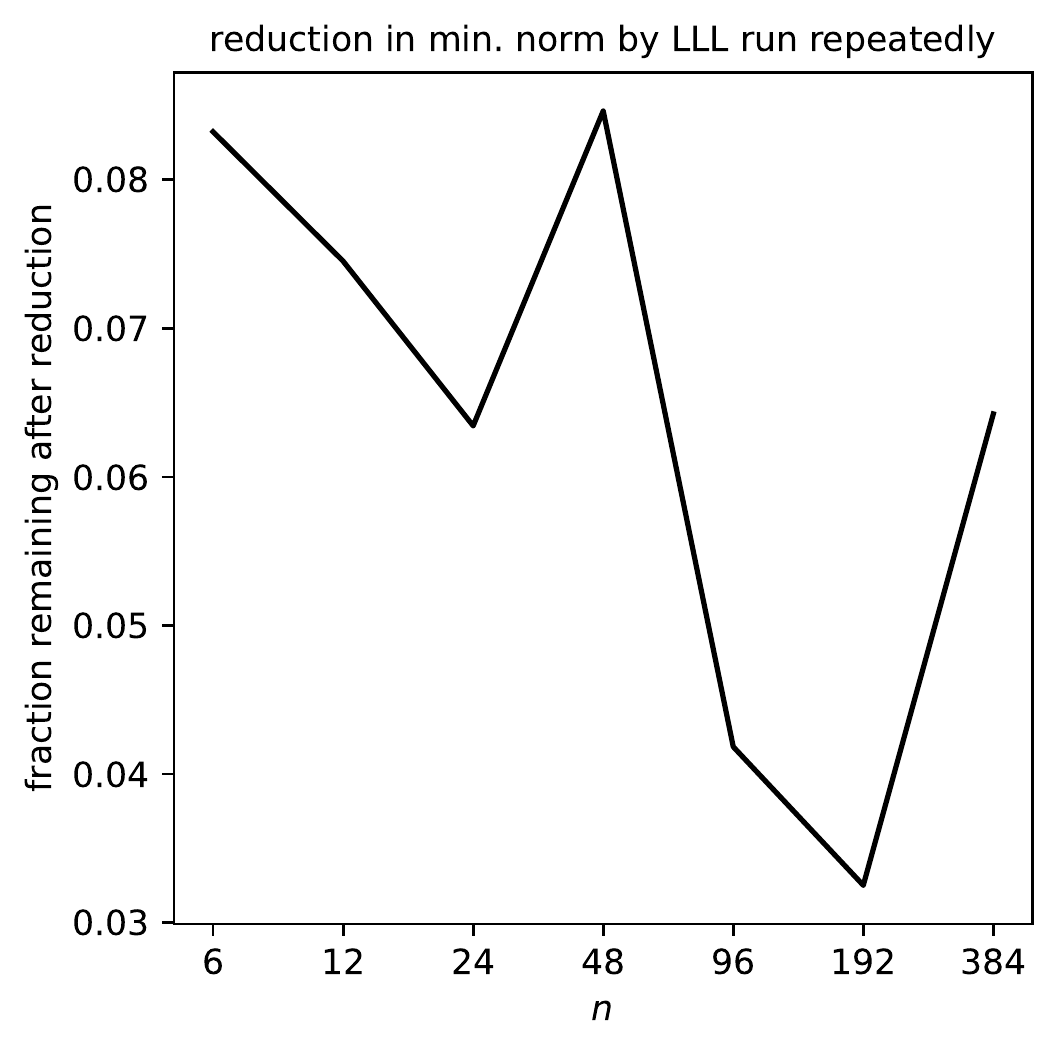}}
{\includegraphics[width=0.495\textwidth]{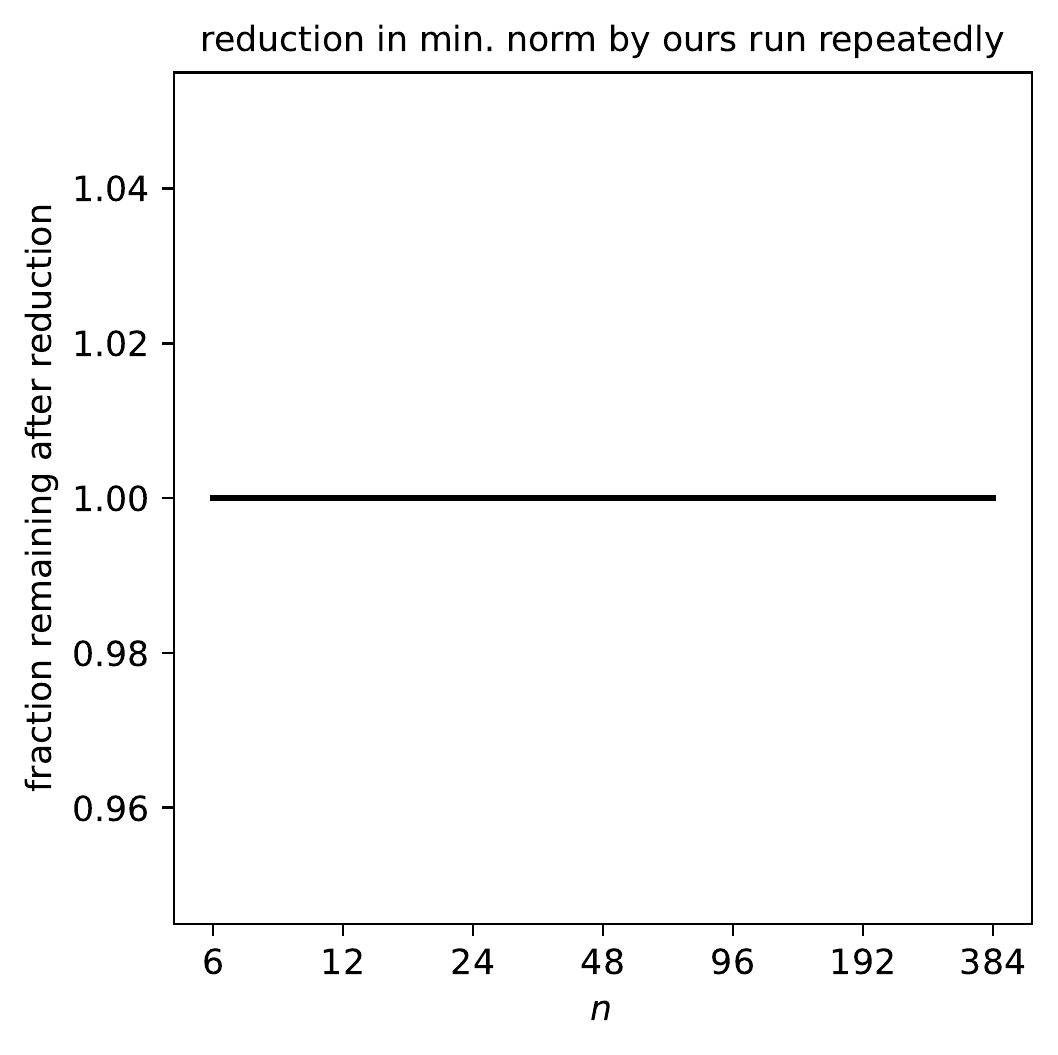}}

\end{centering}
\caption{$\delta = 1-10^{-15}$, $p = 2$, $q = 2^{13} - 1$}
\end{figure}

\begin{figure}
\begin{centering}
{\includegraphics[width=0.495\textwidth]{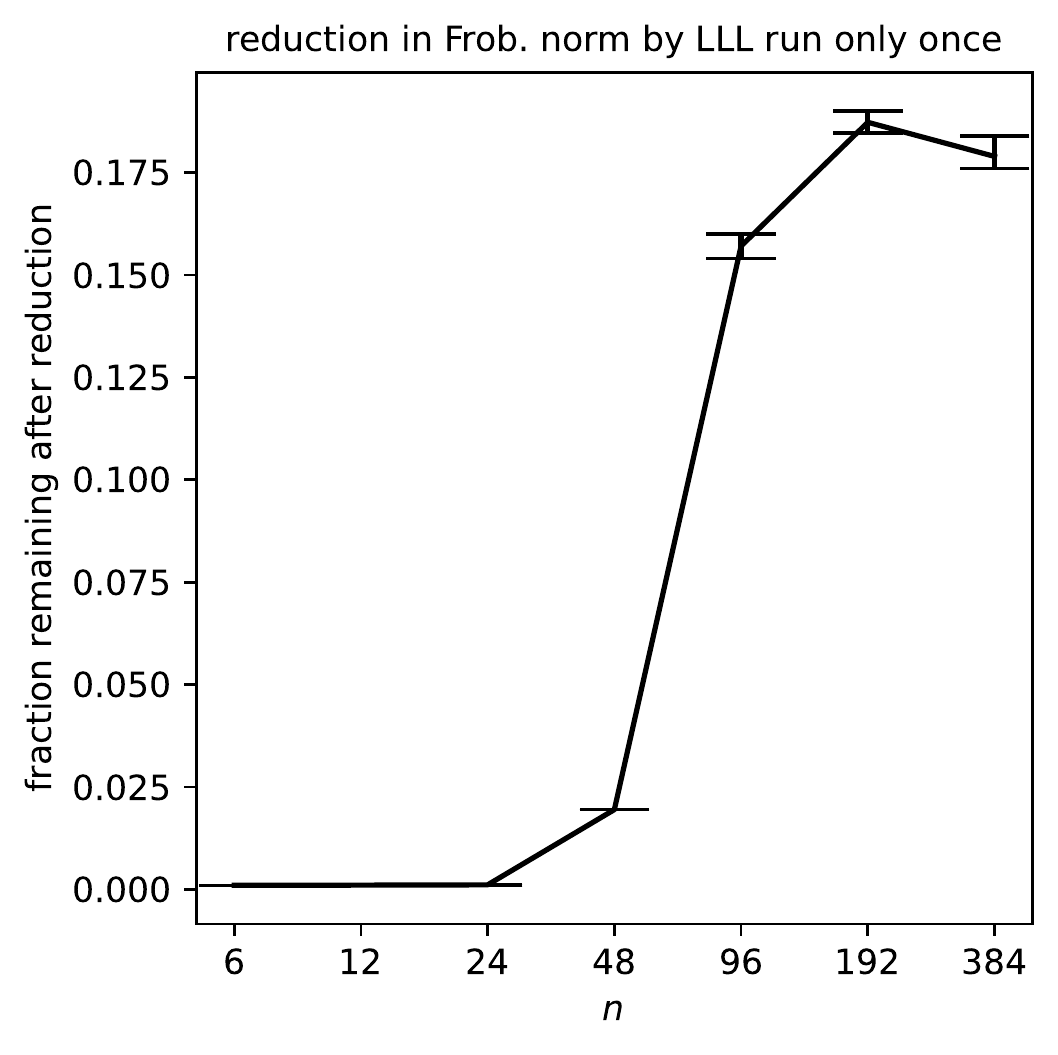}}
{\includegraphics[width=0.495\textwidth]{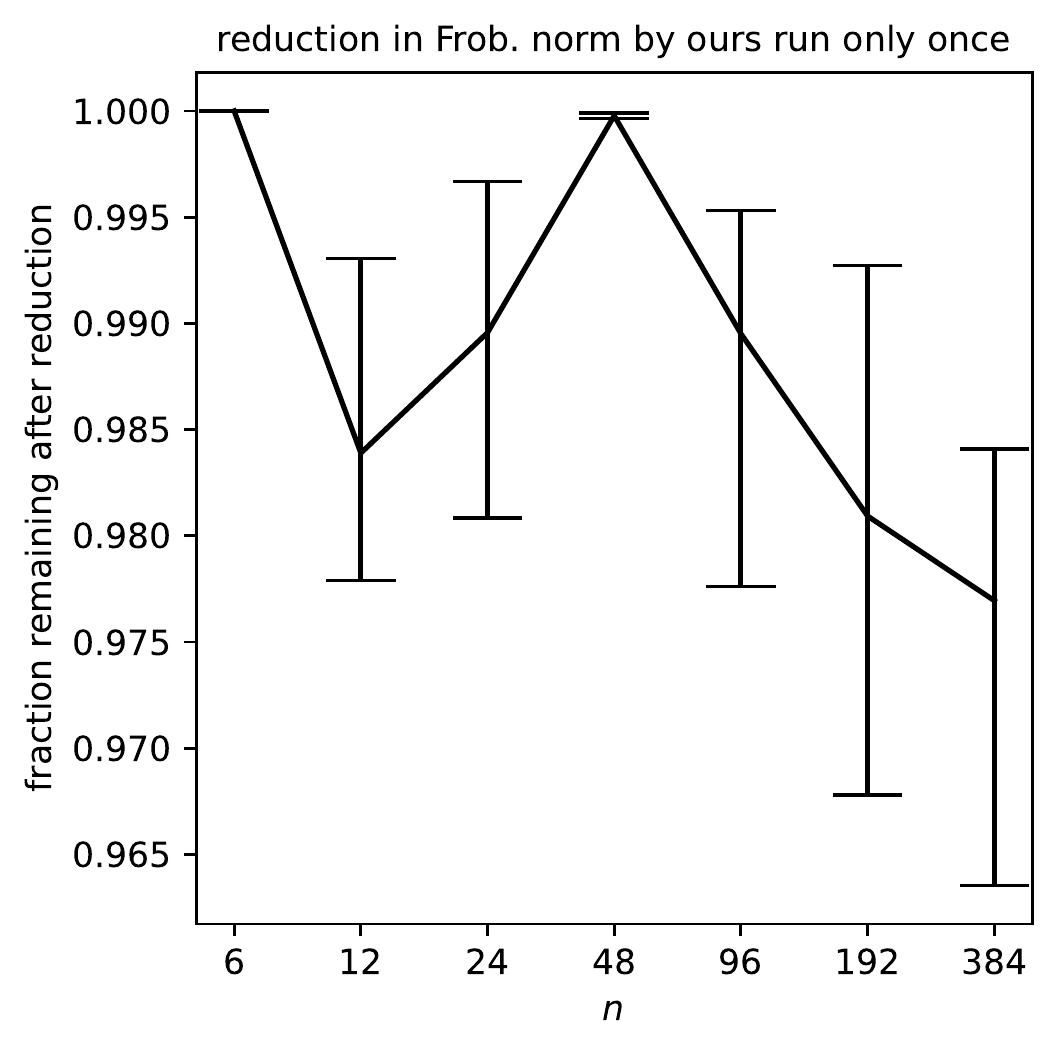}}

{\includegraphics[width=0.495\textwidth]{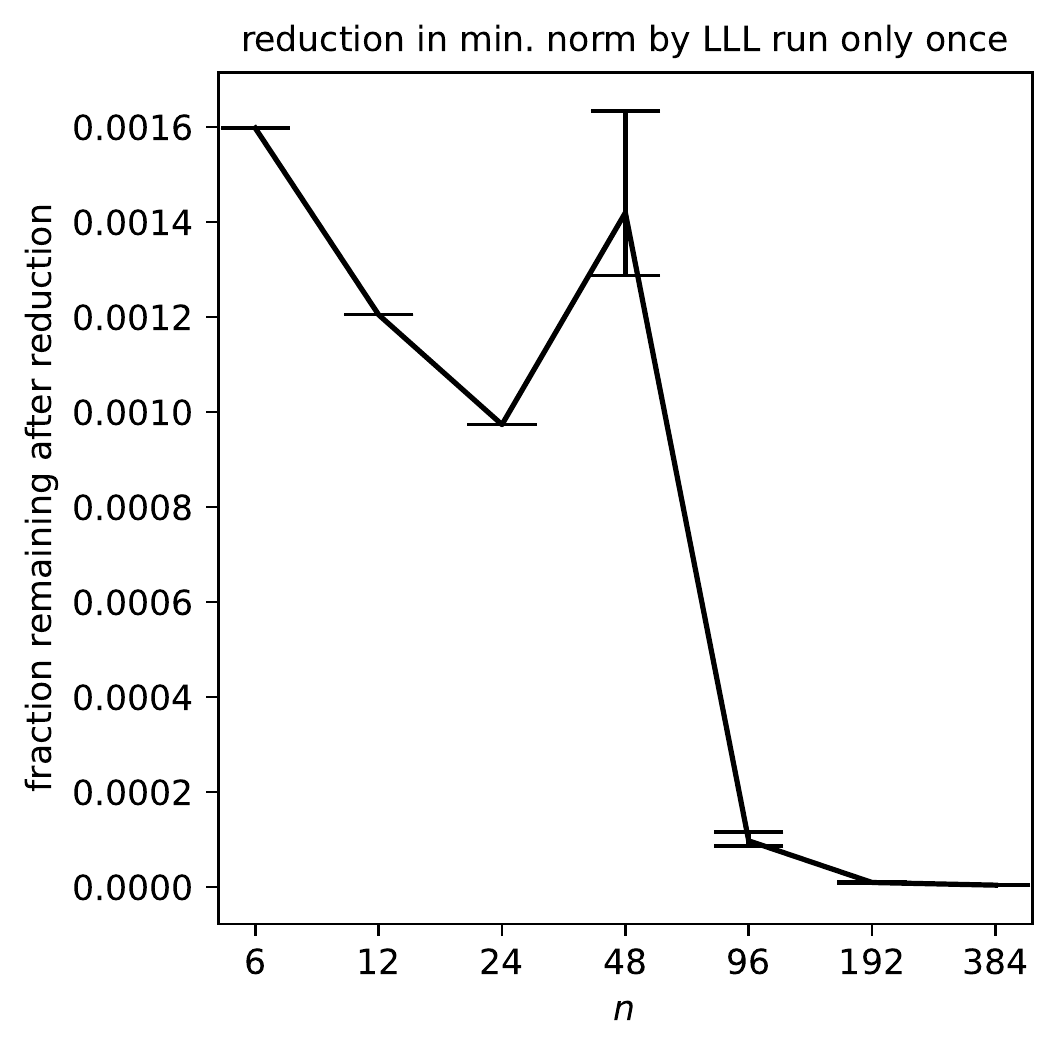}}
{\includegraphics[width=0.495\textwidth]{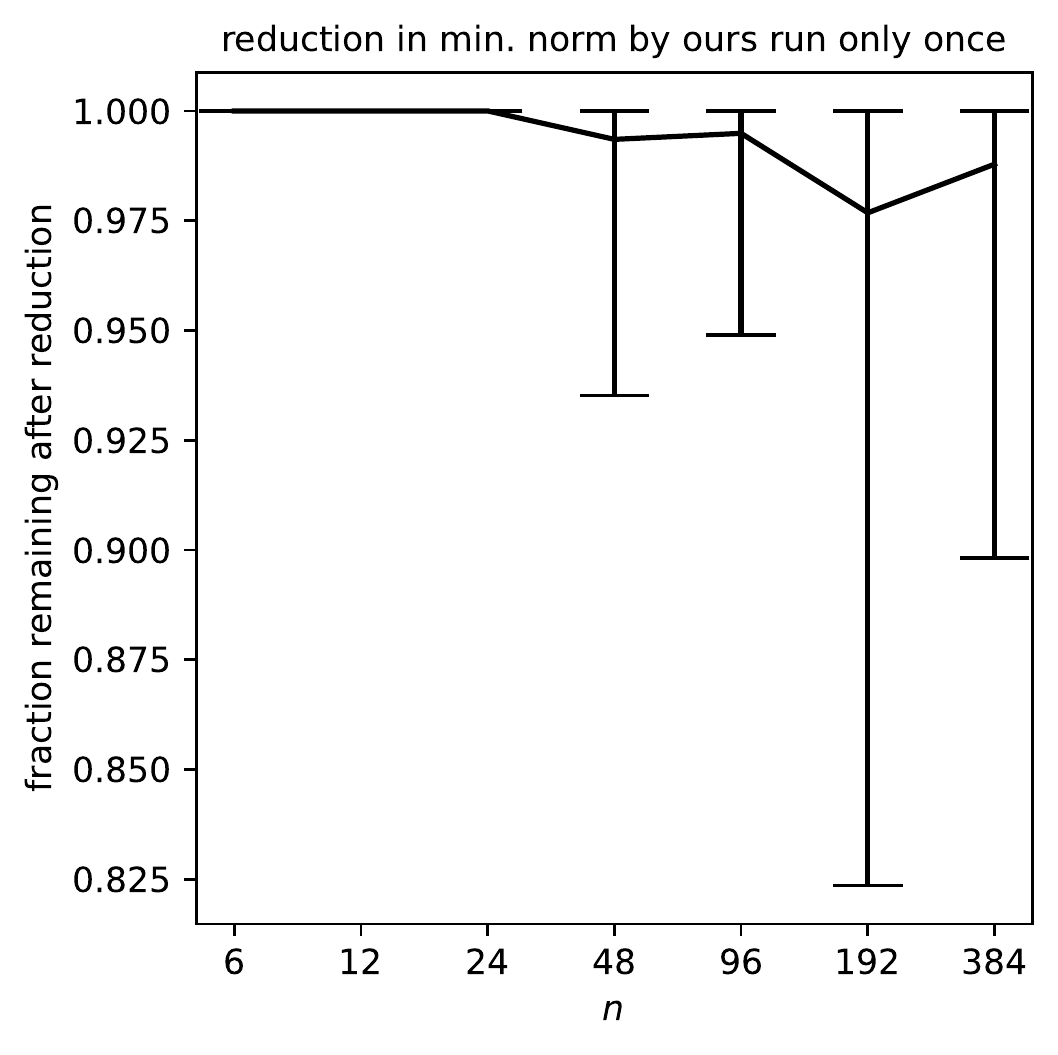}}

\end{centering}
\caption{$\delta = 1-10^{-15}$, $p = 2$, $q = 2^{31} - 1$}
\end{figure}

\begin{figure}
\begin{centering}
{\includegraphics[width=0.495\textwidth]{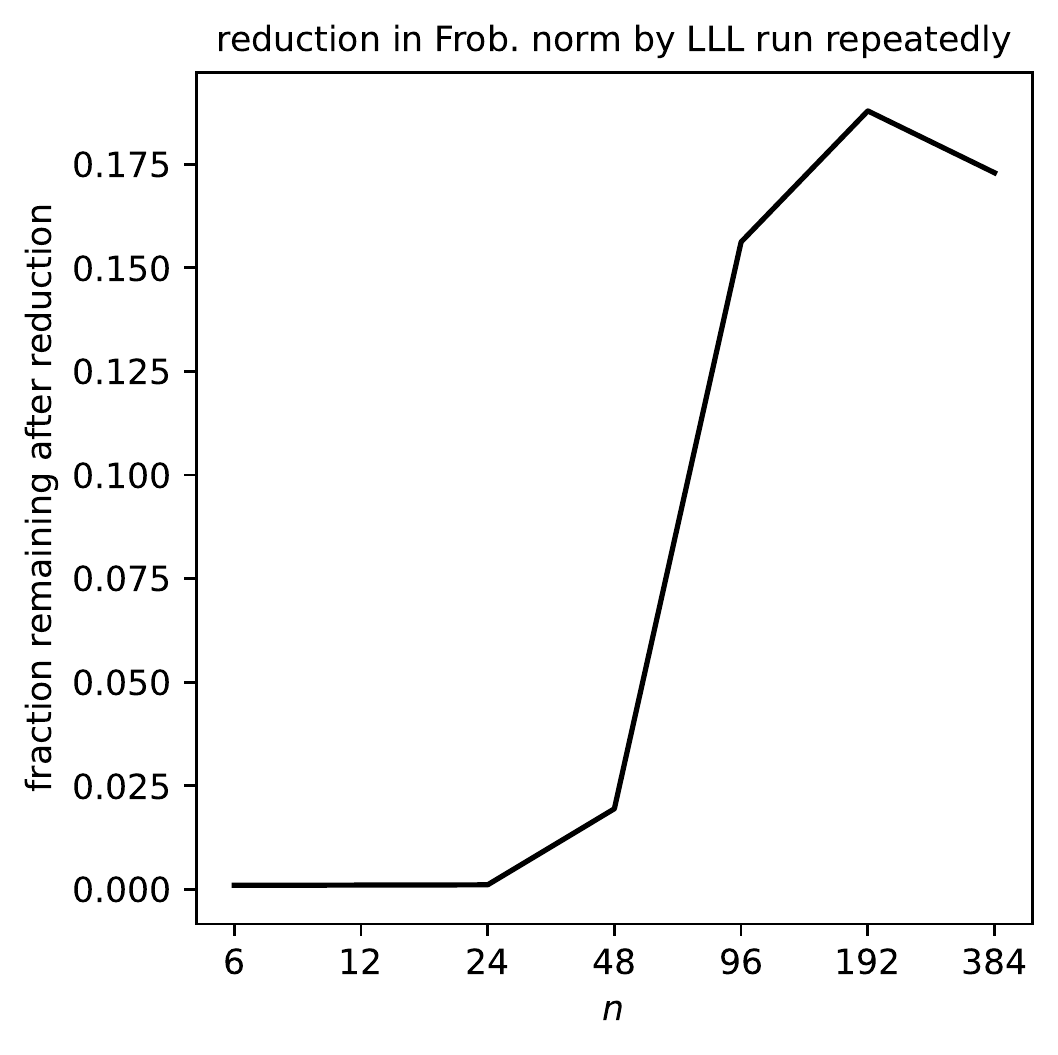}}
{\includegraphics[width=0.495\textwidth]{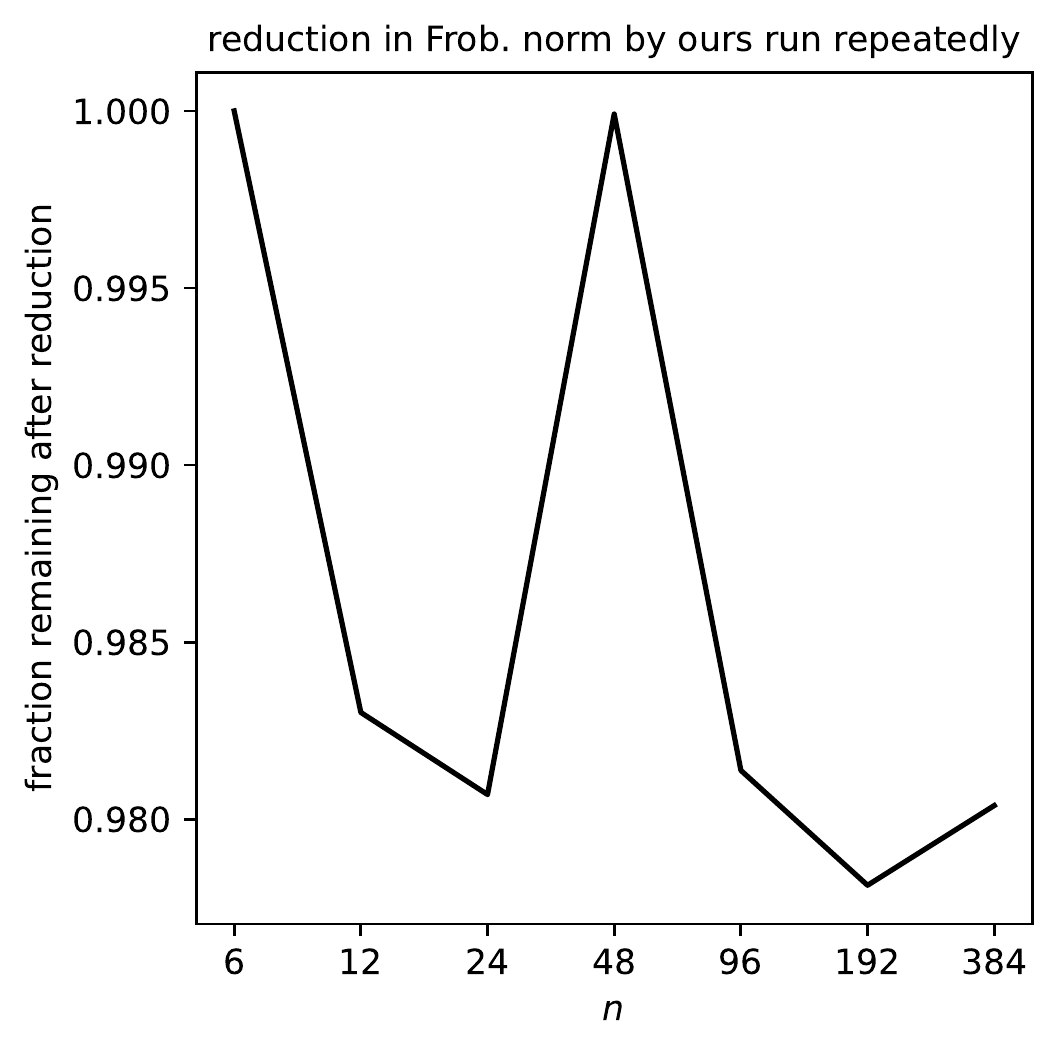}}

{\includegraphics[width=0.495\textwidth]{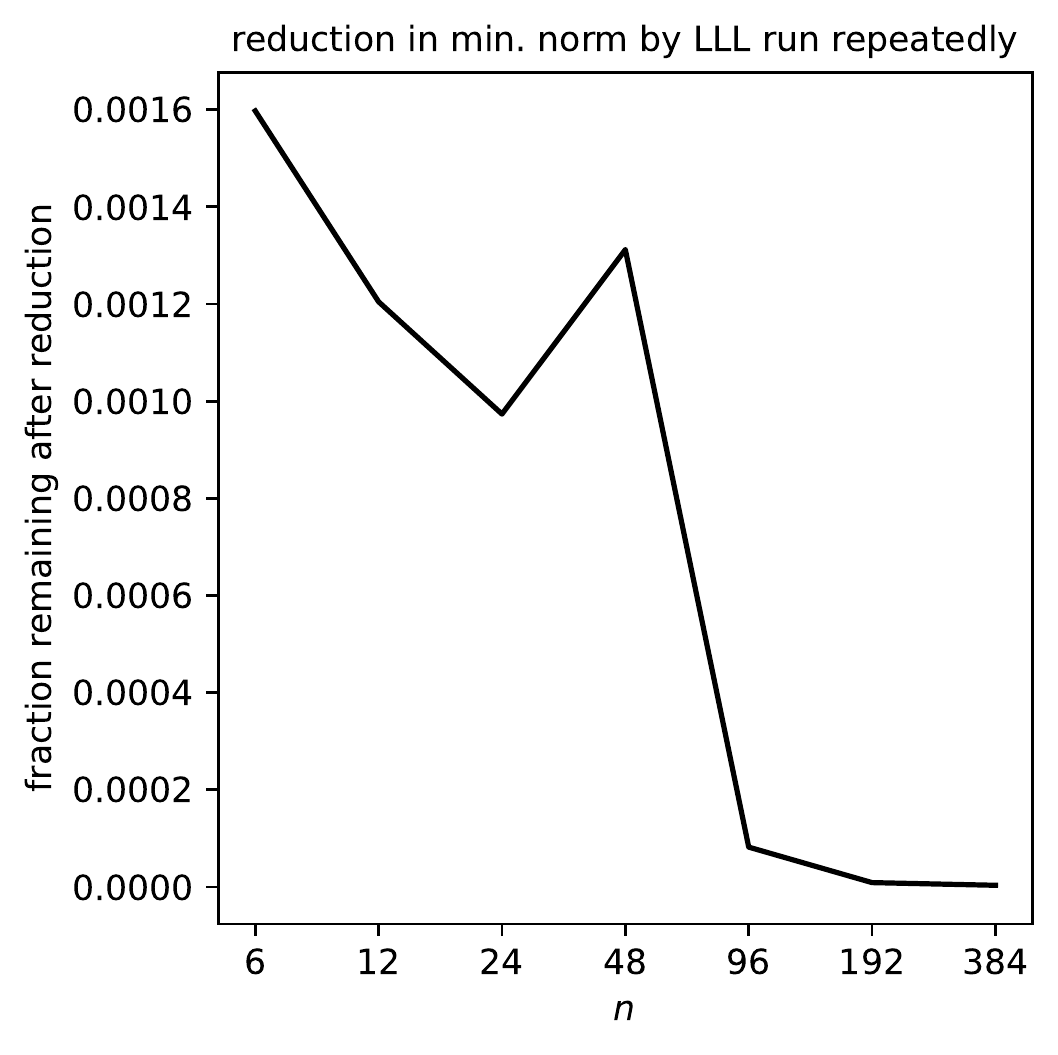}}
{\includegraphics[width=0.495\textwidth]{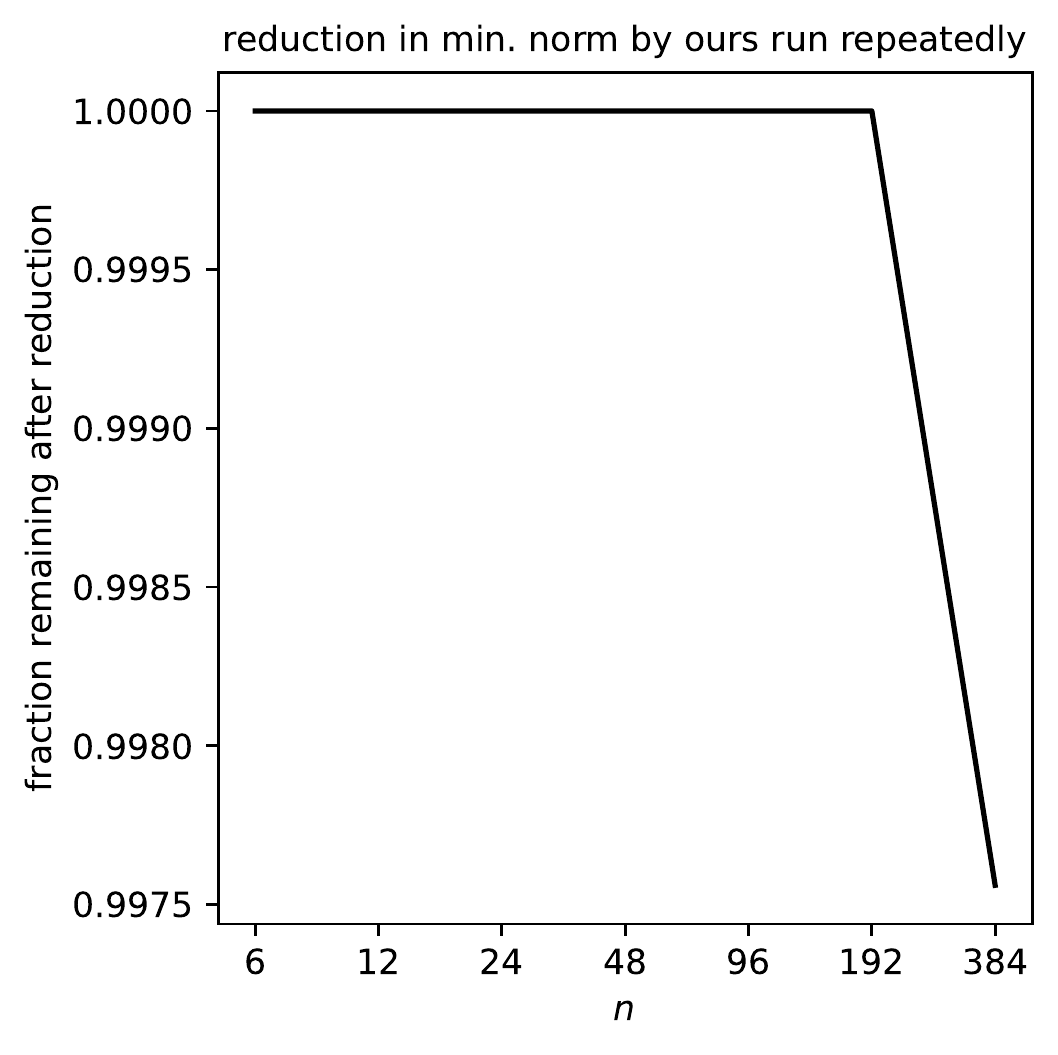}}

\end{centering}
\caption{$\delta = 1-10^{-15}$, $p = 2$, $q = 2^{31} - 1$}
\label{p2err1-1e-15-31}
\end{figure}

\section{Conclusion}
\label{conclusion}

The present paper proposes a very simple and efficient scheme
for lattice reduction. The scheme reduces the Euclidean norms
of the basis vectors monotonically, as guaranteed by rigorous proofs.
Fortunately, the algorithm of the present paper runs much faster
than even a highly optimized implementation of the most classical baseline,
the LLL algorithm of~\cite{lenstra-lenstra-lovasz}.
Unfortunately, the proposed algorithm reduces norms far less than LLL
if not used in conjunction with a method such as LLL.
The main use of the proposed scheme should therefore be to polish the outputs
of another algorithm (such as LLL).
On its own, the algorithm of the present paper tends to get stuck
in rather shallow local minima, with the iterations reaching an equilibrium
that is far away from optimally minimizing the Euclidean norms
of the basis vectors. Convergence is monotonic and hence guaranteed,
but equilibrium tends to attain without reaching the global optimum.
The algorithm of the present paper is extremely efficient computationally,
however, so can be suitable for rapidly burnishing the outputs
of other algorithms for lattice reduction.

\section*{Acknowledgements}

We would like to thank Zeyuan Allen-Zhu, Kamalika Chaudhuri, Mingjie Chen,
Evrard Garcelon, Matteo Pirotta, Jana Sotakova, and Emily Wenger.

\clearpage

\appendix
\section{Poorly performing alternatives}
\label{poor}

This appendix mentions two possible modifications that lack
the firm theoretical grounding of the algorithm presented above
and performed rather poorly in numerical experiments.
These modifications may be natural, yet seem not to work well.
Subsection~\ref{multiple} considers adding to a basis vector
multiple other basis vectors simultaneously,
such that the full linear combination would minimize the Euclidean norm
if the coefficients in the linear combination did not have to be rounded
to the nearest integers.
Subsection~\ref{modified} considers a modified Gram-Schmidt procedure.

\subsection{Combining multiple vectors simultaneously}
\label{multiple}

One possibility is to choose a basis vector at random, say $a^i_j$,
and add to that vector the linear combination of all other basis vectors
which minimizes the Euclidean norm of the result, with the coefficients
in the linear combination rounded to the nearest integers.
That is, choose an index $j$ uniformly at random,
and calculate real-valued coefficients $c^i_{j,k}$ such that the Euclidean norm
$\| a^i_j - \sum_{k=1}^n c^i_{j,k} \cdot a^i_k \|$ is minimal,
where $c^i_{j,j} = 0$. Then, construct
$a^{i+1}_j = a^i_j - \sum_{k=1}^n \nint(c^i_{j,k}) \cdot a^i_k$.

Repeating the process for multiple iterations, $i = 0$, $1$, $2$, \dots,
would appear reasonable. However, this scheme worked well empirically only
when the number $n$ of basis vectors was very small, at least
when the number $m$ of entries in each of the basis vectors was equal to $n$.
Rounding the coefficients $c^i_{j,k}$ to the nearest integers is too harsh
for this process to work well.

\subsection{Modified Gram-Schmidt process}
\label{modified}

Another possibility is to run the classical Gram-Schmidt procedure
on the basis vectors, while subtracting off from all vectors
not yet added to the orthogonal basis the projections
onto the current pivot vector. In this modified Gram-Schmidt scheme,
each iteration chooses as the next pivot vector the residual basis vector
for which adding that basis vector to the orthogonal basis would
minimize the sum of the $p$-th powers of the Euclidean norms
of the reduced basis vectors. The iteration orthogonalizes the pivot vector
against all previously chosen pivot vectors and then subtracts off
(from all residual basis vectors not yet chosen as pivots)
the projection onto the orthogonalized pivot vector,
with the coefficients in the projections rounded to the nearest integers.

This scheme strongly resembles the LLL algorithm
of~\cite{lenstra-lenstra-lovasz}, but with a different pivoting strategy
(using modified Gram-Schmidt). Numerical experiments indicate that
the modified Gram-Schmidt performs somewhat similarly to
yet significantly worse than the classical LLL algorithm.
Omitting the bubble-sorting of the LLL algorithm
via the so-called ``Lov\'asz criterion'' spoils the scheme.

This scheme is also reminiscent of the variants of the LLL algorithm
with so-called ``deep insertions,'' as developed
by~\cite{schnorr-euchner}, \cite{fontein-schneider-wagner},
\cite{yasuda-yamaguchi}, and others.
LLL with deep insertions performs much better, however,
both theoretically and practically.
Other modifications to the LLL algorithm,
notably the BKZ and BKW methods reviewed by~\cite{nguyen-vallee} and others,
also perform much better than the modified Gram-Schmidt.

\section{Further figures}
\label{further}

This appendix presents figures analogous to those of Section~\ref{results},
but using different values of the parameters $\delta$ and $p$
detailed in Subsection~\ref{figures}.
Figures~\ref{p2time1-1e-1}--\ref{p2err1-1e-1-31} are the same
as Figures~\ref{p2time1-1e-15}--\ref{p2err1-1e-15-31},
but with $\delta = 1 - 10^{-1}$ instead of $\delta = 1 - 10^{-15}$.
Figures~\ref{pstime1-1e-15}--\ref{pserr1-1e-1-31} are the same
as Figures~\ref{p2time1-1e-15}--\ref{p2err1-1e-1-31} for $n = 192$,
but with varying values of $p$ rather than just $p = 2$.

\begin{figure}
\begin{centering}
{\includegraphics[width=0.495\textwidth]{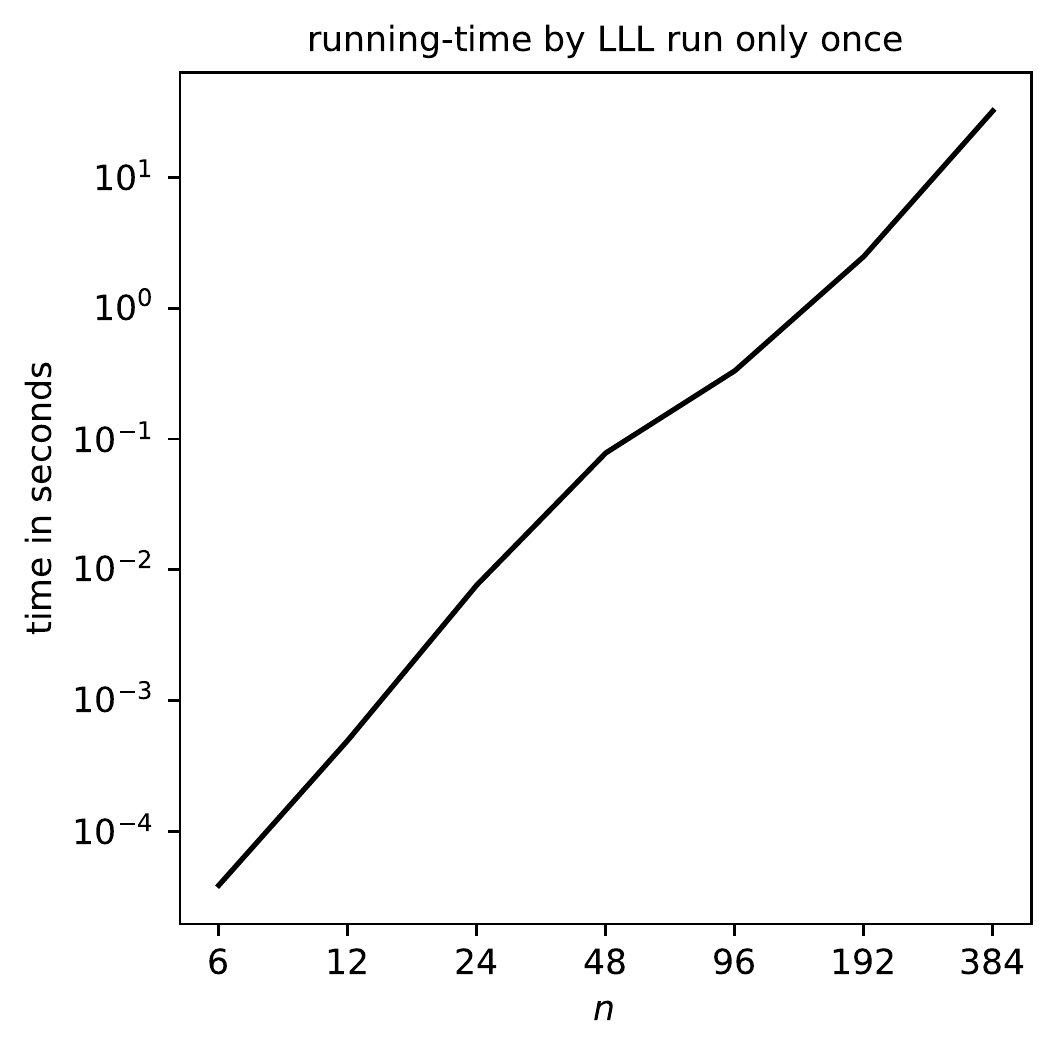}}
{\includegraphics[width=0.495\textwidth]{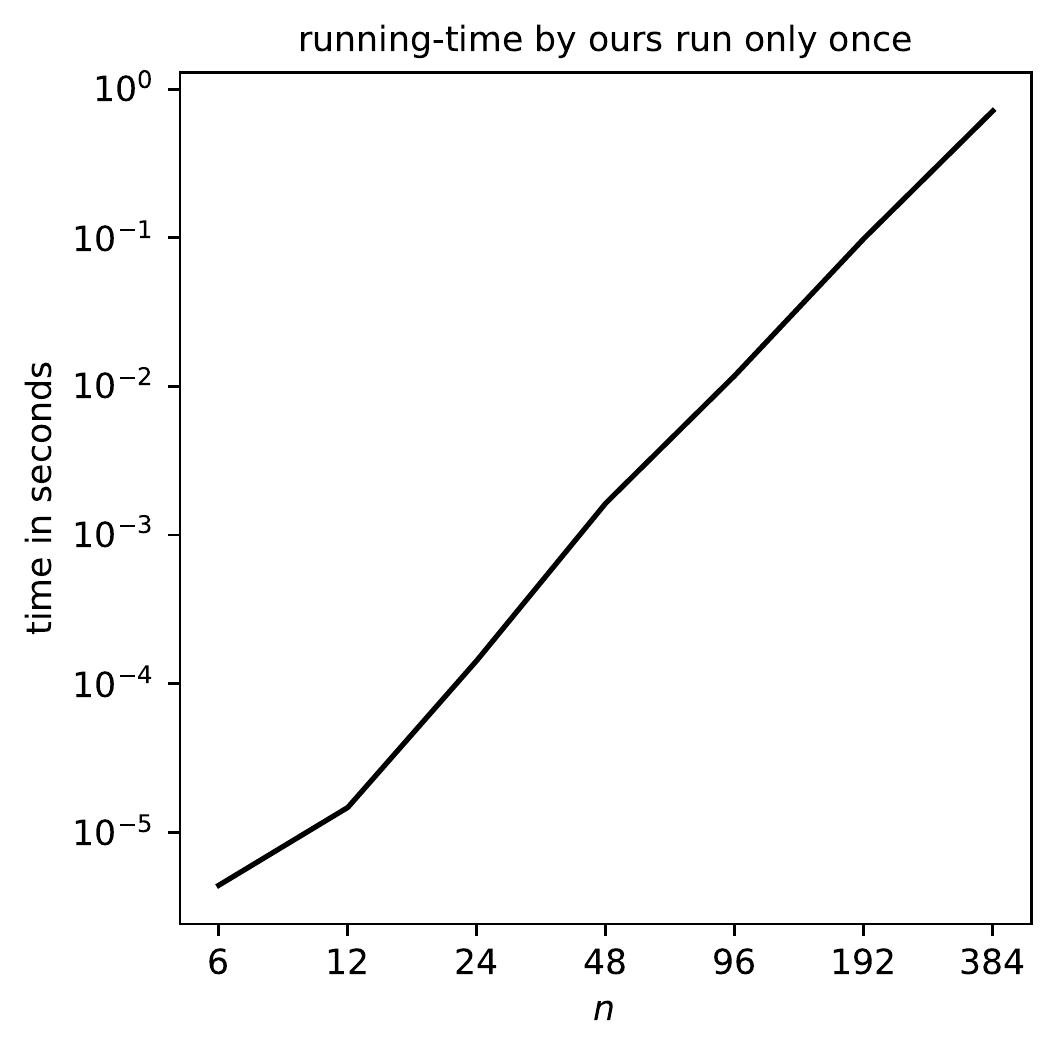}}

{\includegraphics[width=0.495\textwidth]{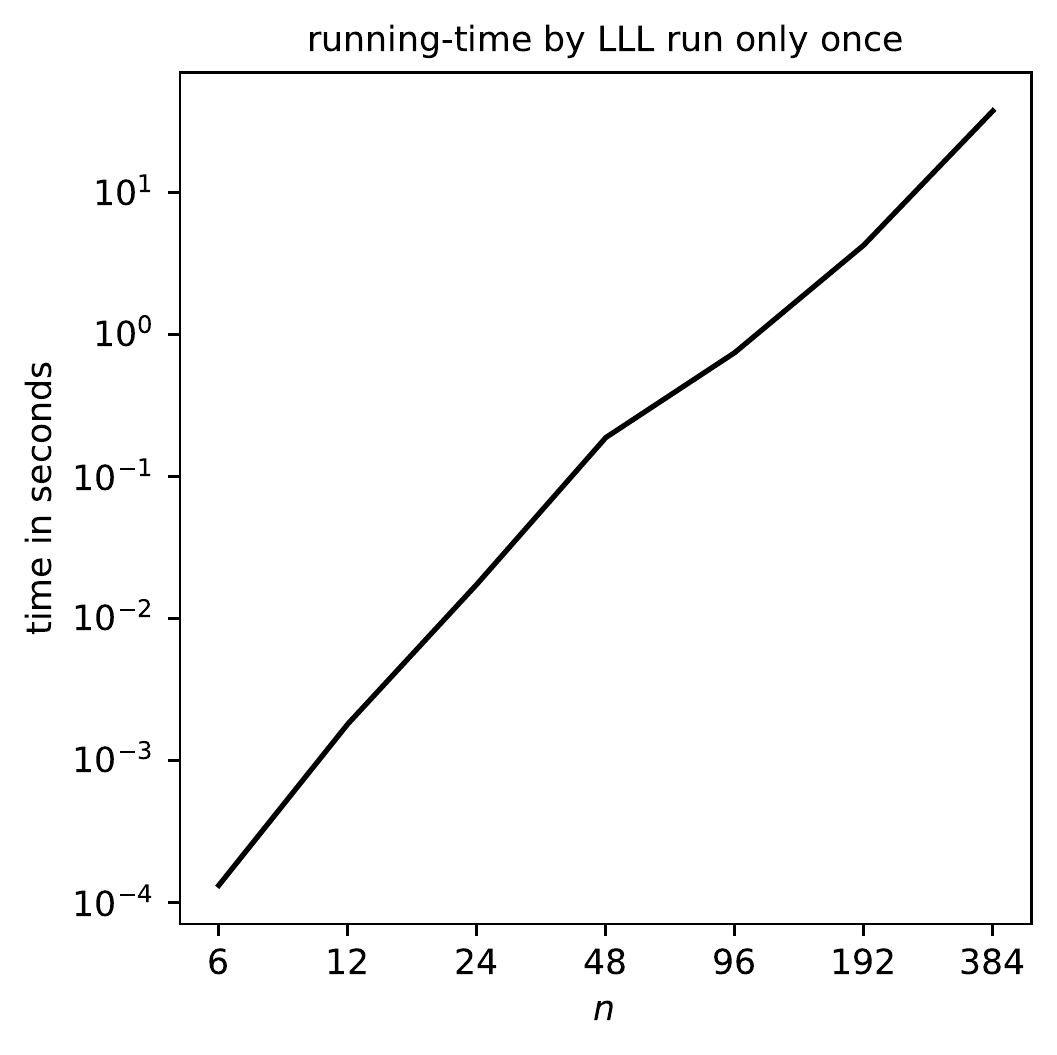}}
{\includegraphics[width=0.495\textwidth]{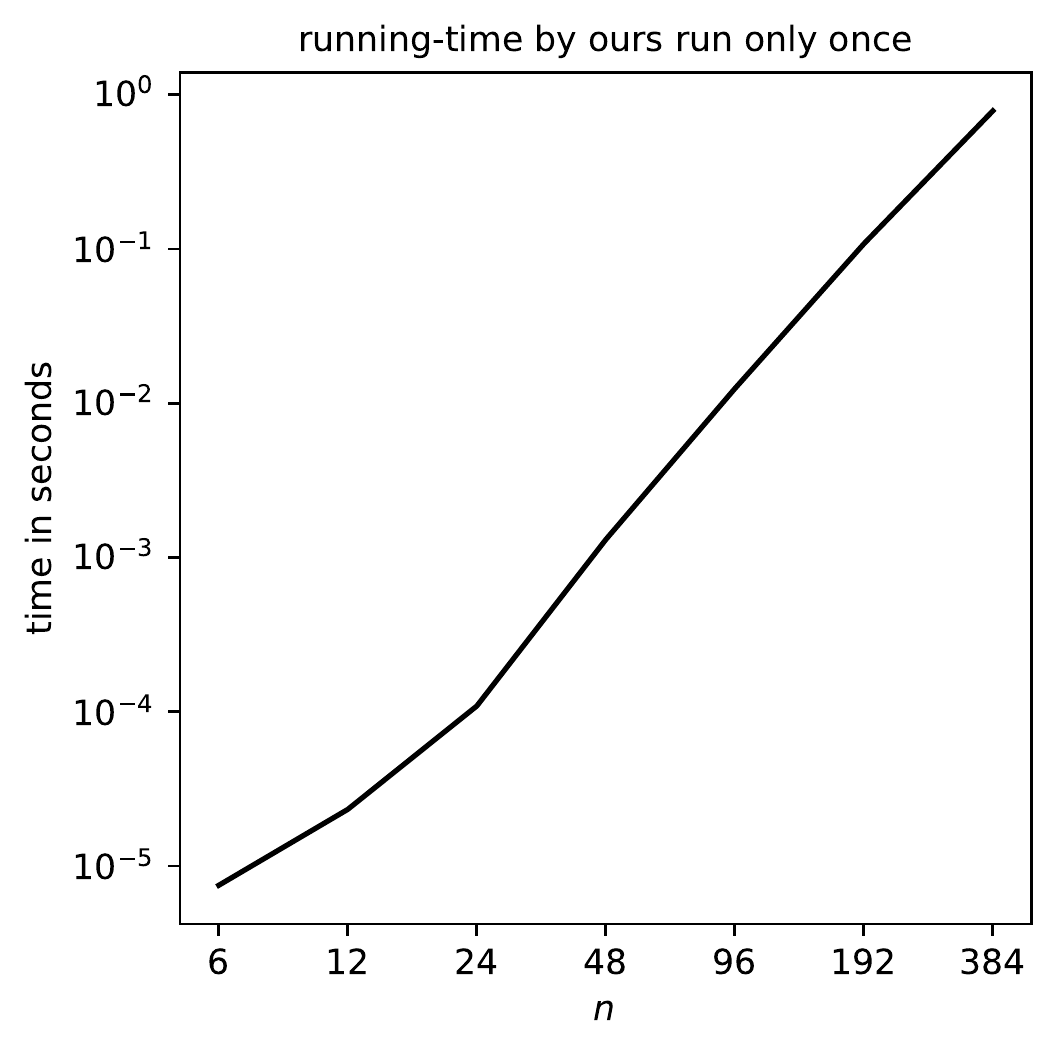}}

\end{centering}
\caption{$\delta = 1-10^{-1}$, $p = 2$;
         the upper plots are for $q = 2^{13} - 1$,
         the lower plots are for $q = 2^{31} - 1$}
\label{p2time1-1e-1}
\end{figure}

\begin{figure}
\begin{centering}
{\includegraphics[width=0.495\textwidth]{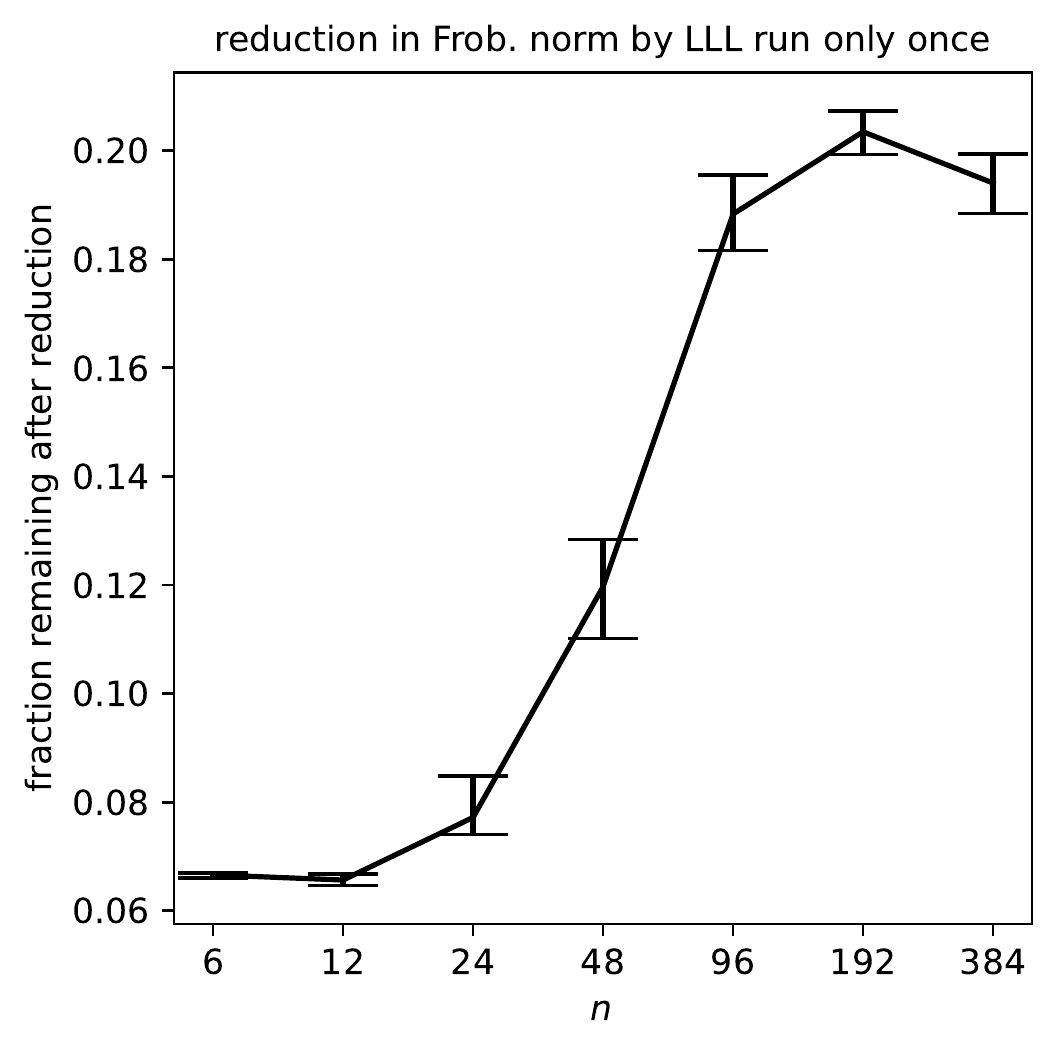}}
{\includegraphics[width=0.495\textwidth]{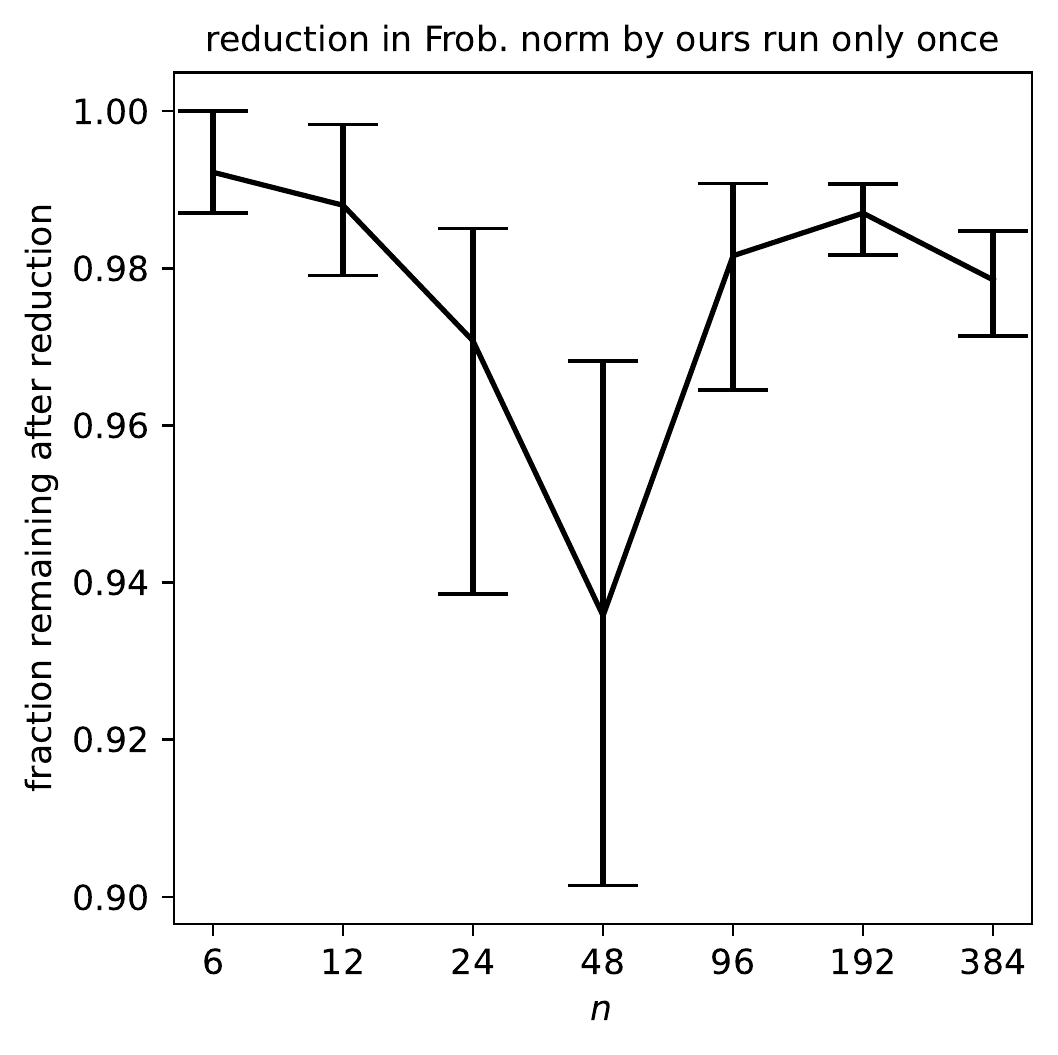}}

{\includegraphics[width=0.495\textwidth]{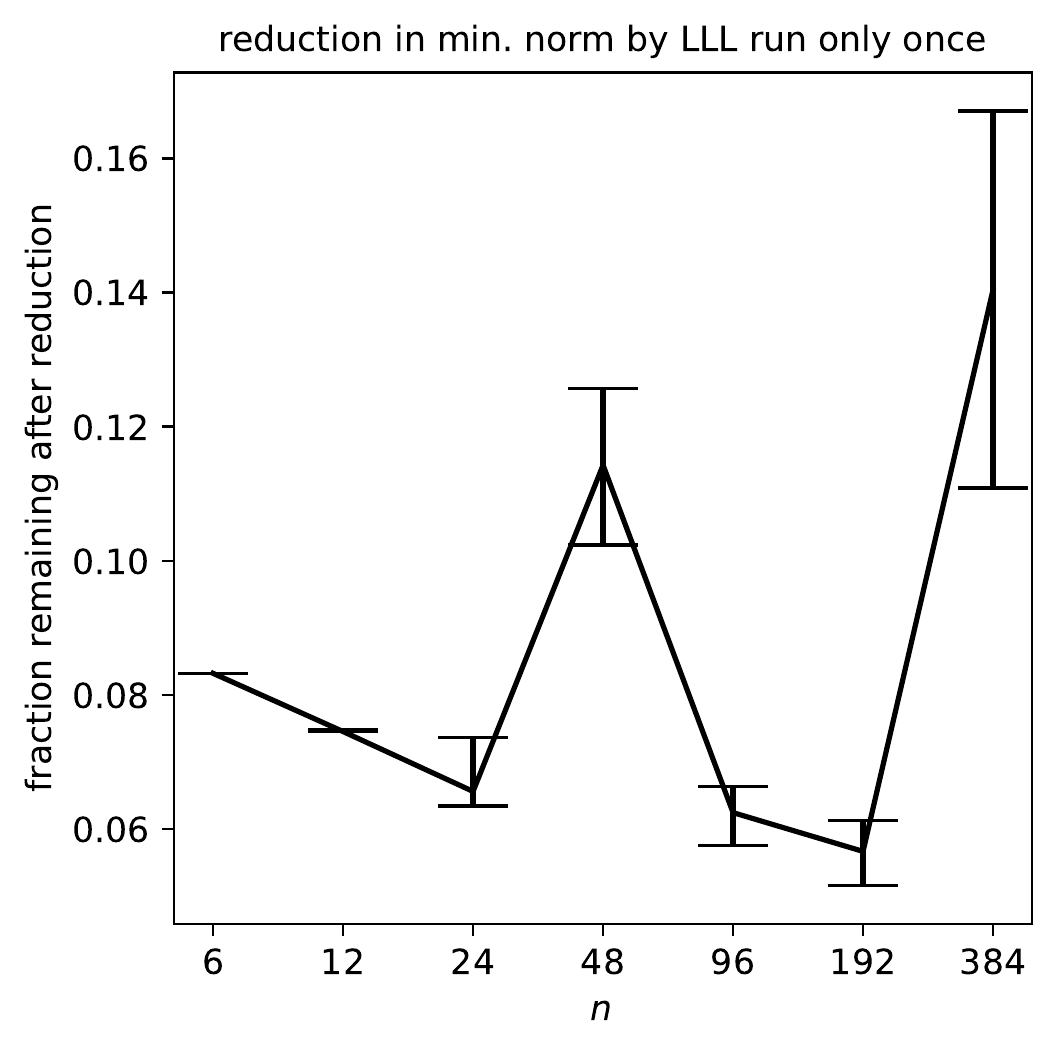}}
{\includegraphics[width=0.495\textwidth]{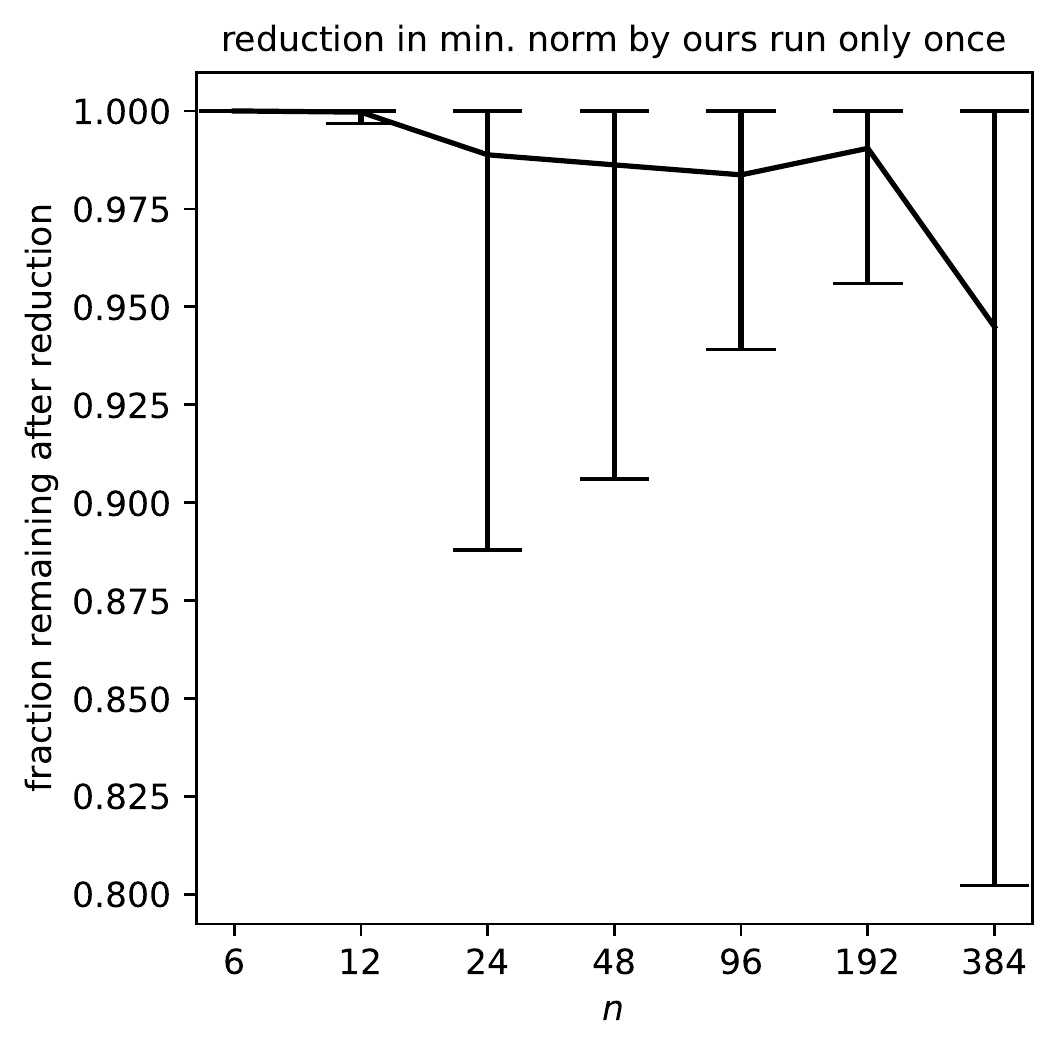}}

\end{centering}
\caption{$\delta = 1-10^{-1}$, $p = 2$, $q = 2^{13} - 1$}
\end{figure}

\begin{figure}
\begin{centering}
{\includegraphics[width=0.495\textwidth]{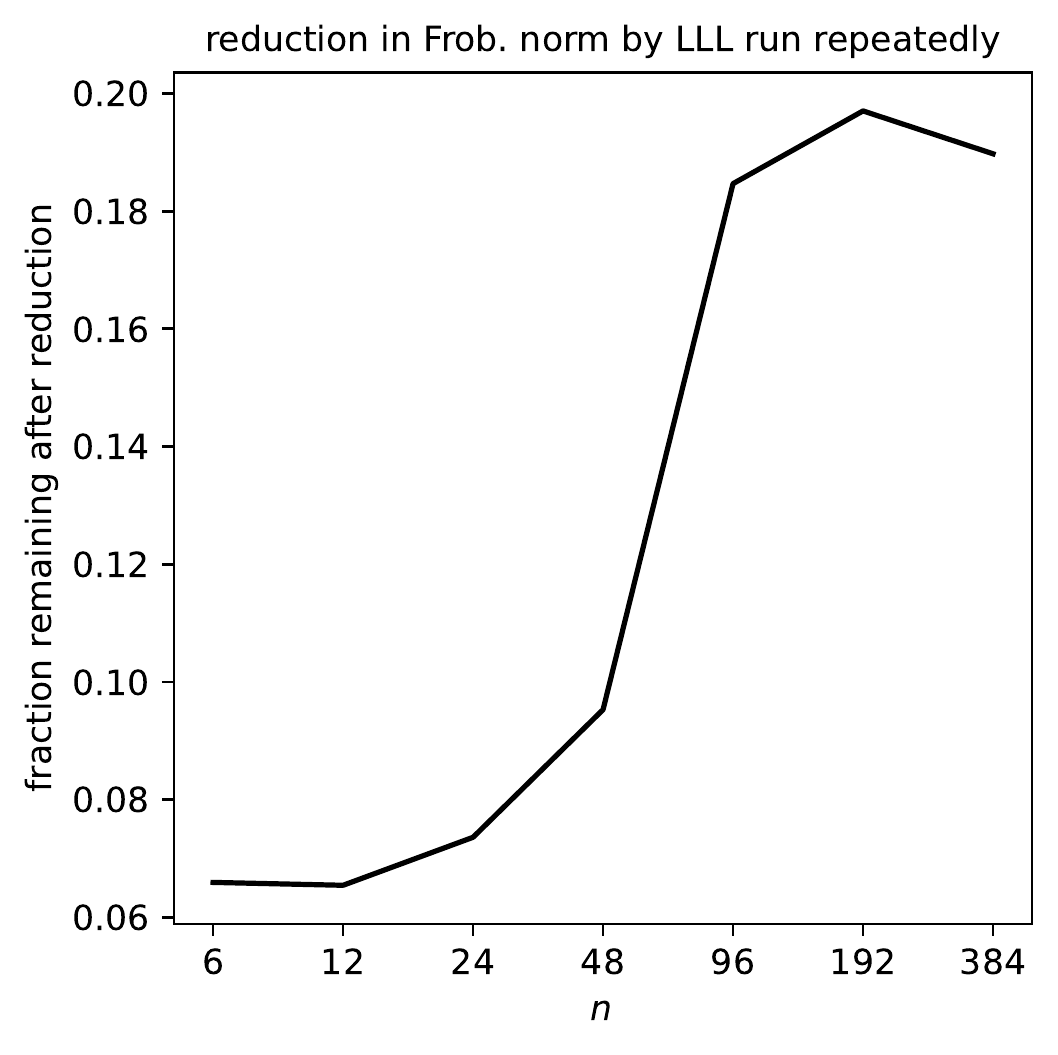}}
{\includegraphics[width=0.495\textwidth]{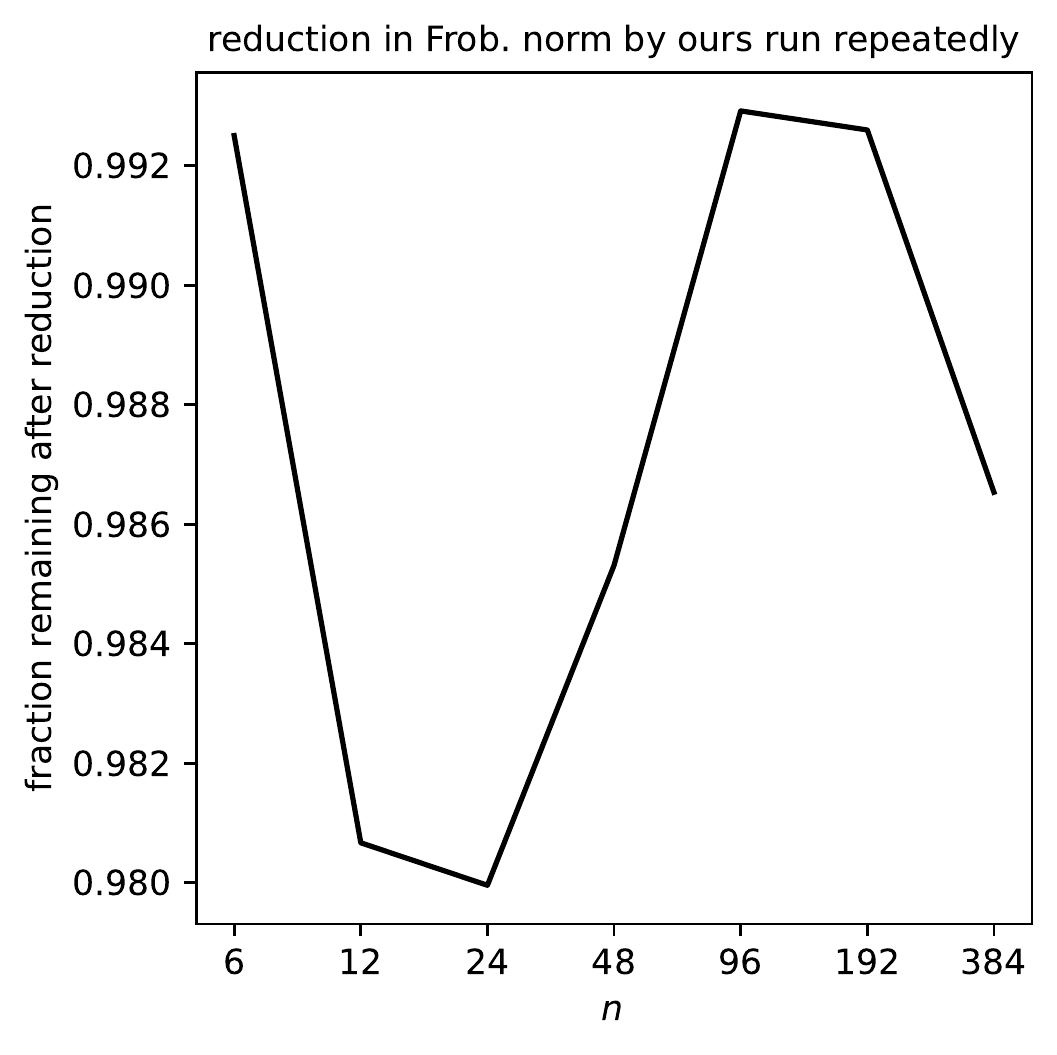}}

{\includegraphics[width=0.495\textwidth]{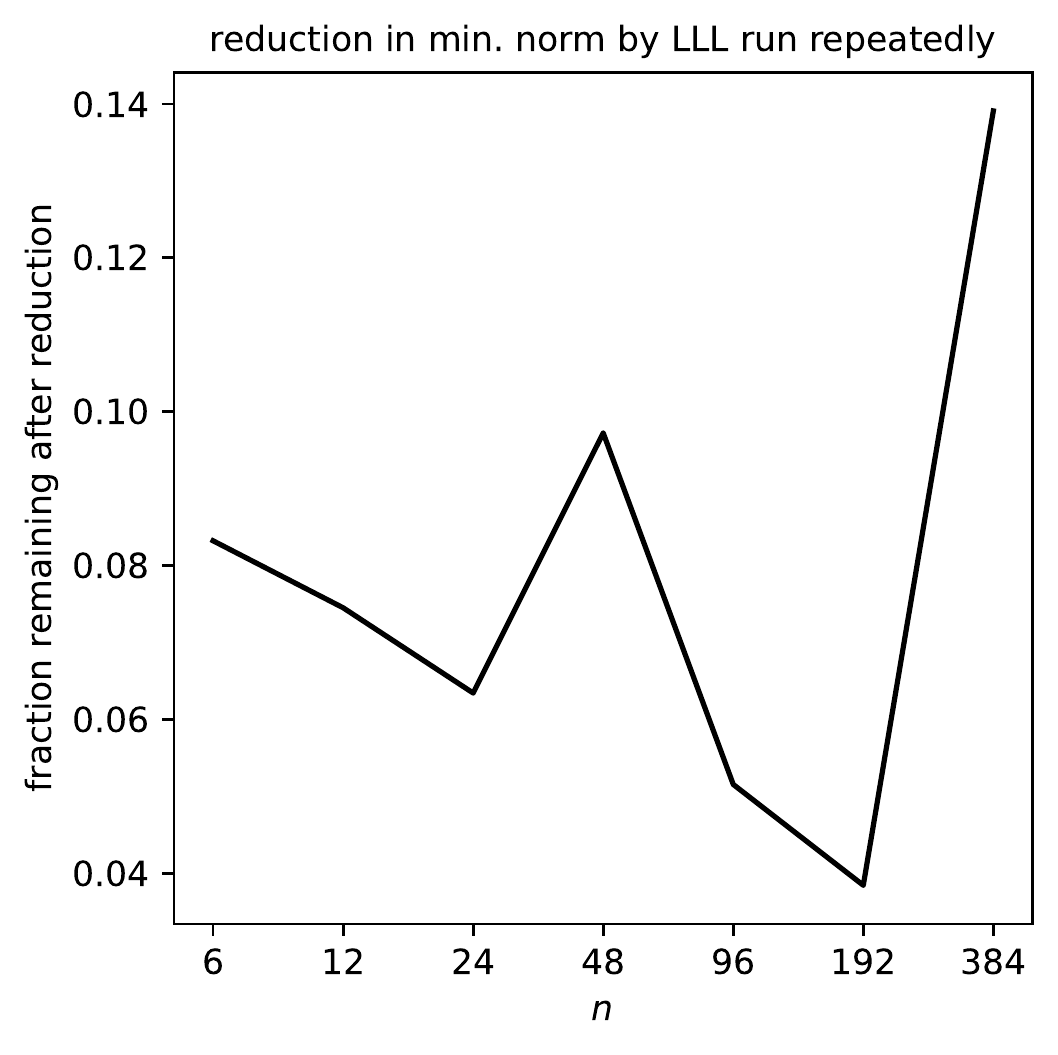}}
{\includegraphics[width=0.495\textwidth]{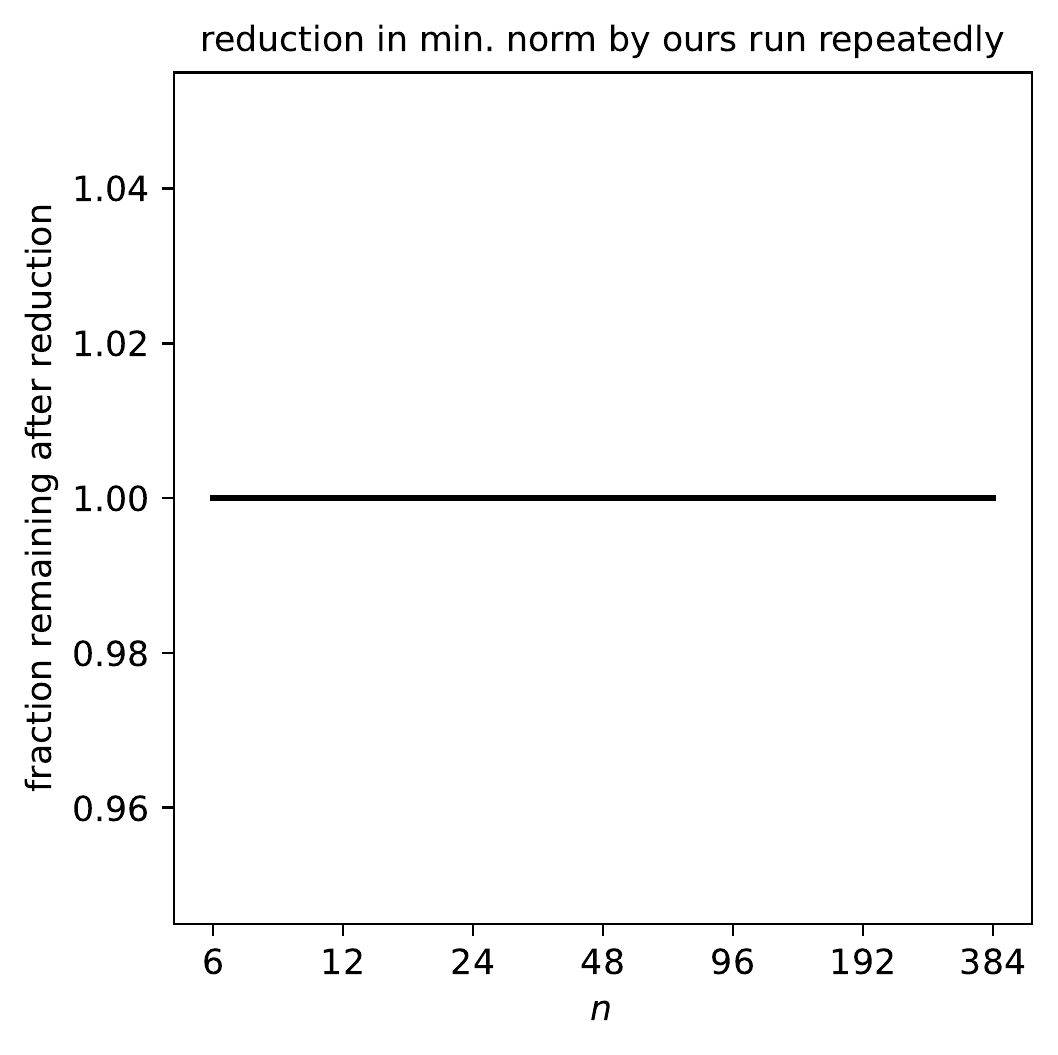}}

\end{centering}
\caption{$\delta = 1-10^{-1}$, $p = 2$, $q = 2^{13} - 1$}
\end{figure}

\begin{figure}
\begin{centering}
{\includegraphics[width=0.495\textwidth]{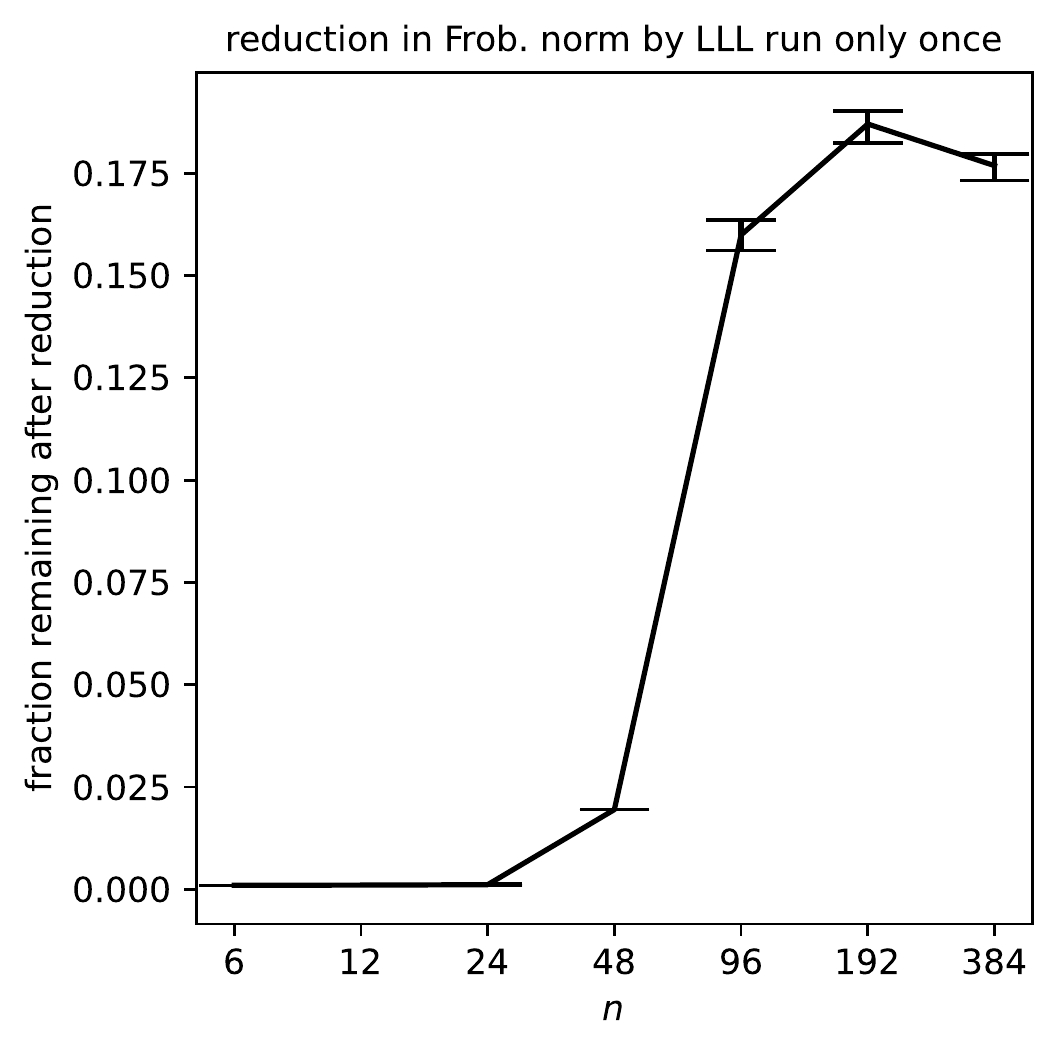}}
{\includegraphics[width=0.495\textwidth]{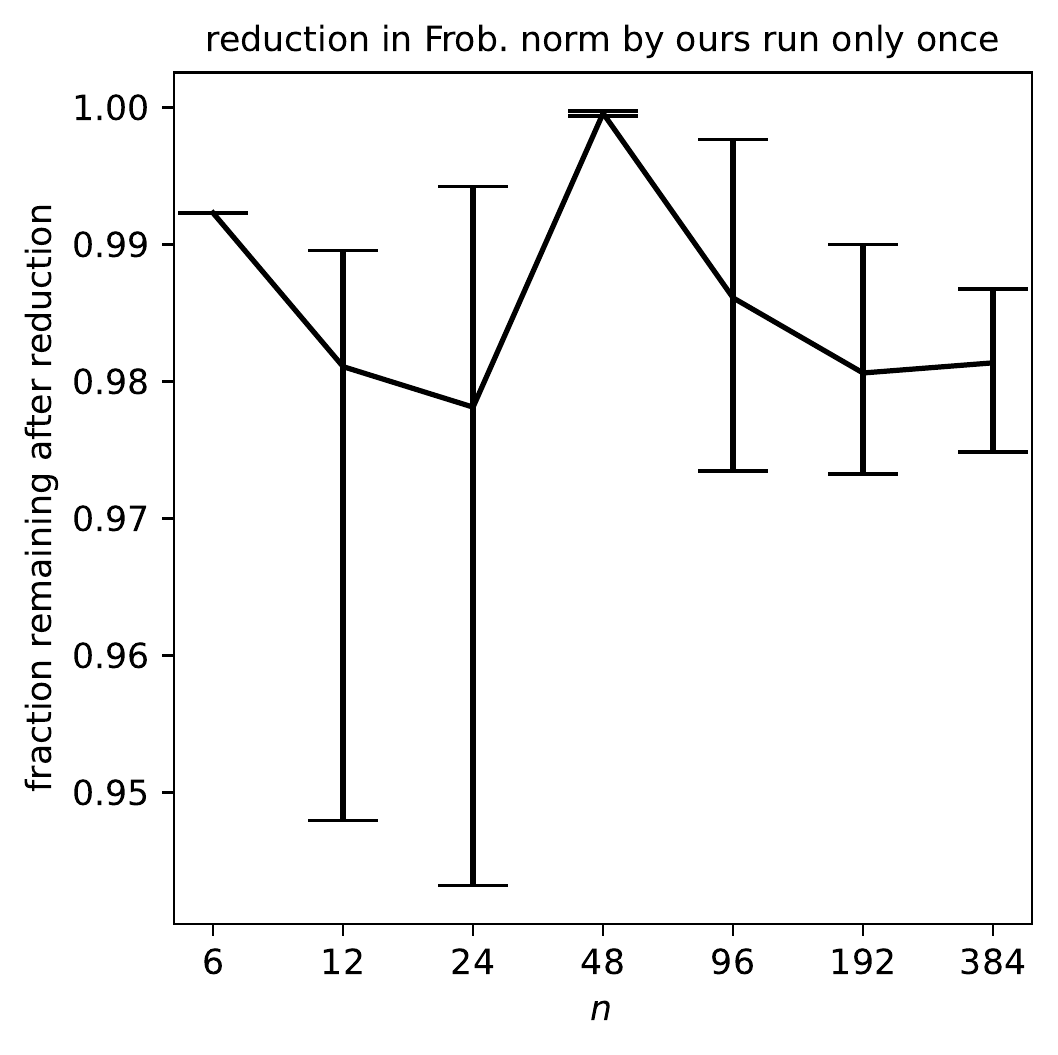}}

{\includegraphics[width=0.495\textwidth]{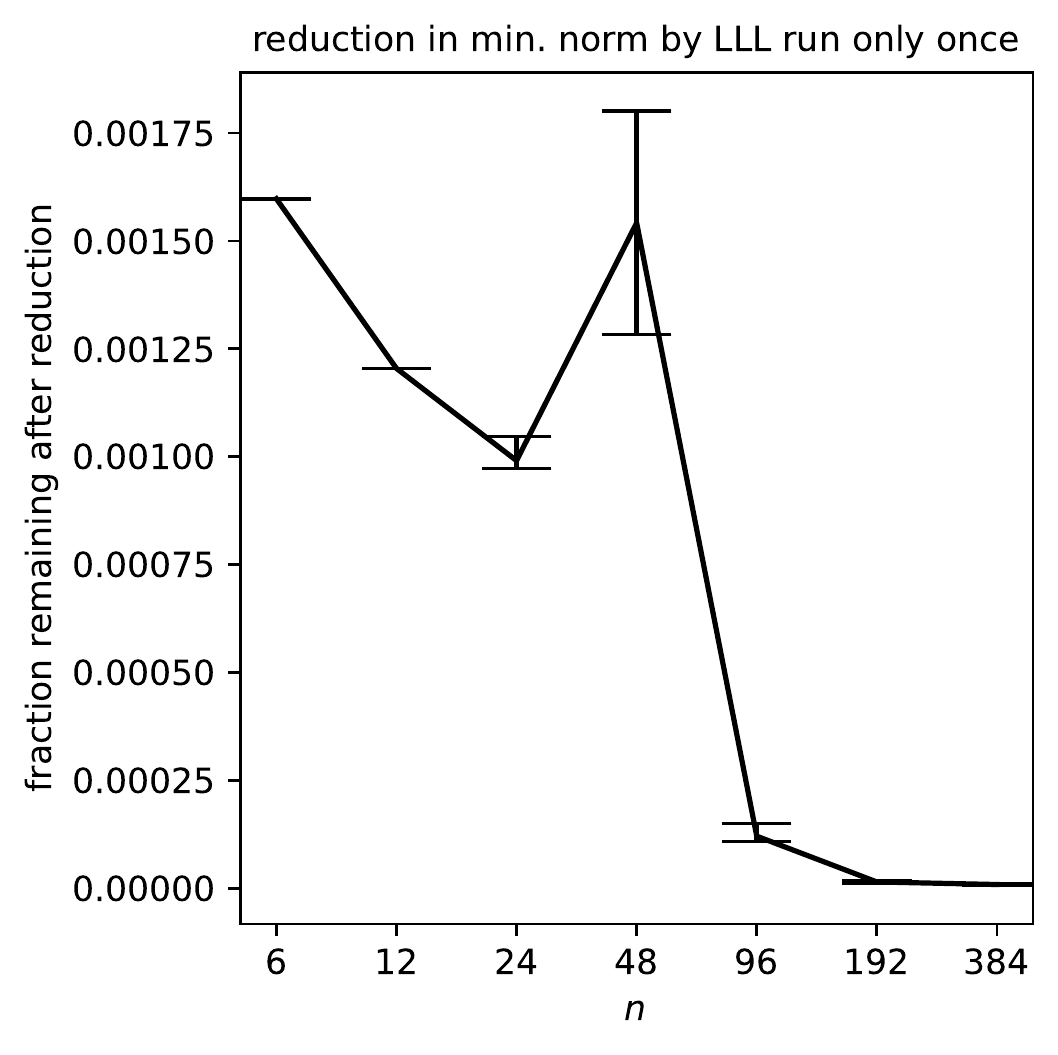}}
{\includegraphics[width=0.495\textwidth]{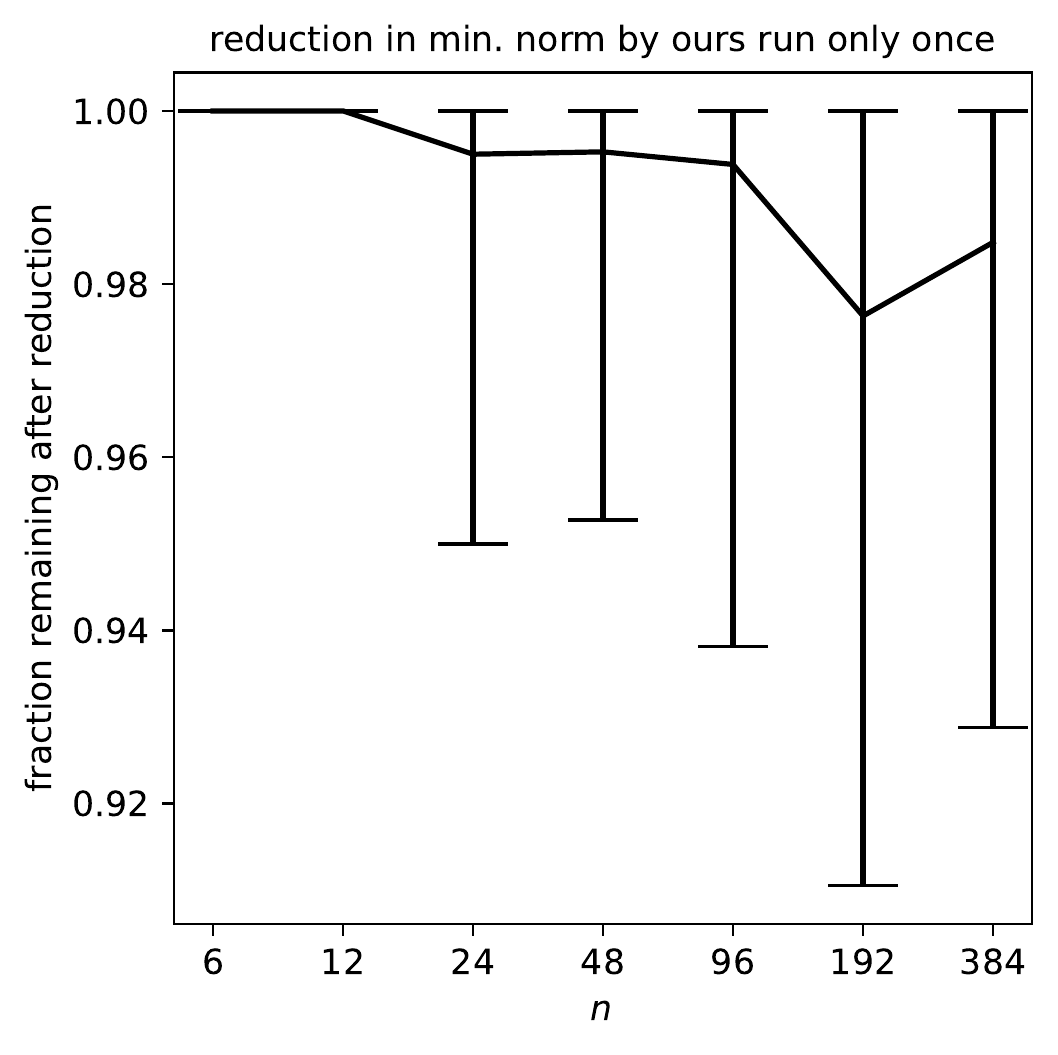}}

\end{centering}
\caption{$\delta = 1-10^{-1}$, $p = 2$, $q = 2^{31} - 1$}
\end{figure}

\begin{figure}
\begin{centering}
{\includegraphics[width=0.495\textwidth]{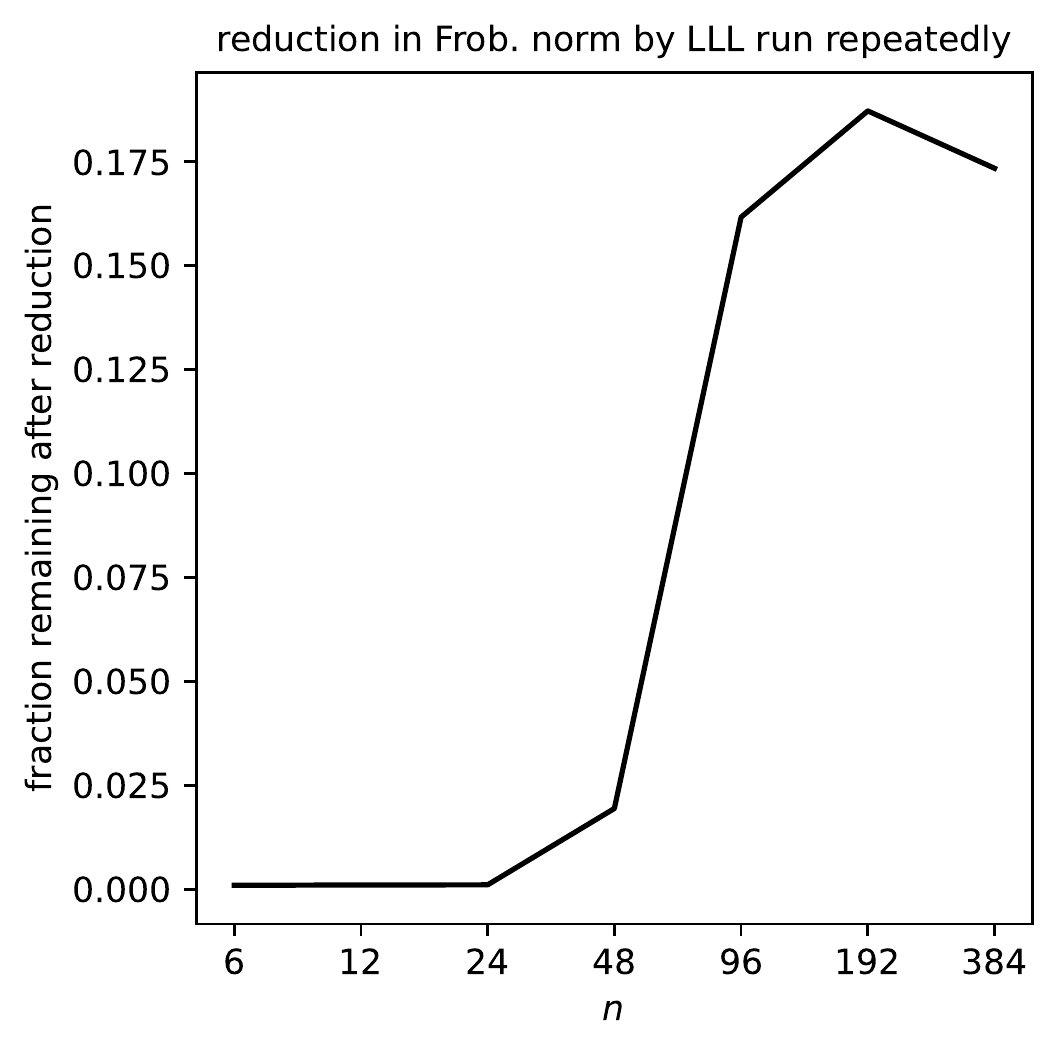}}
{\includegraphics[width=0.495\textwidth]{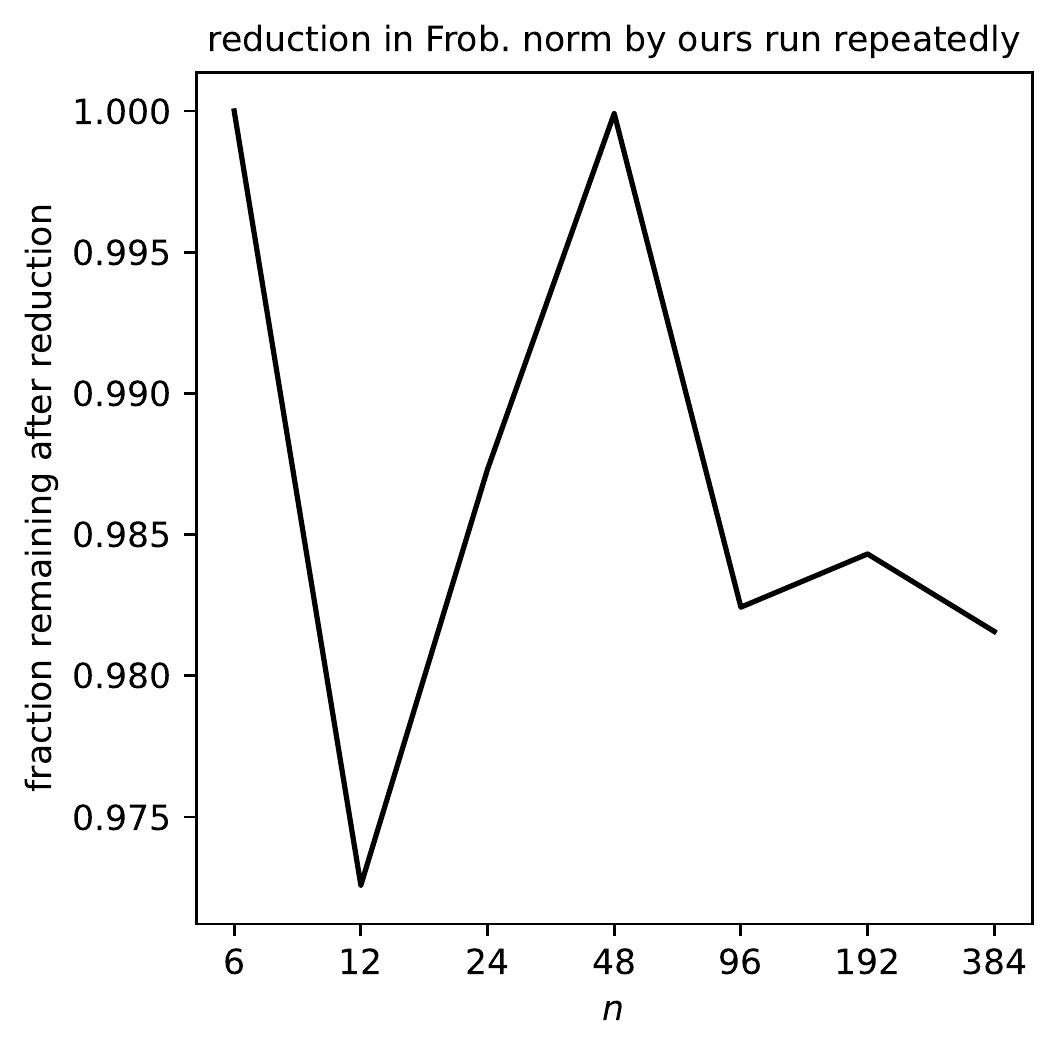}}

{\includegraphics[width=0.495\textwidth]{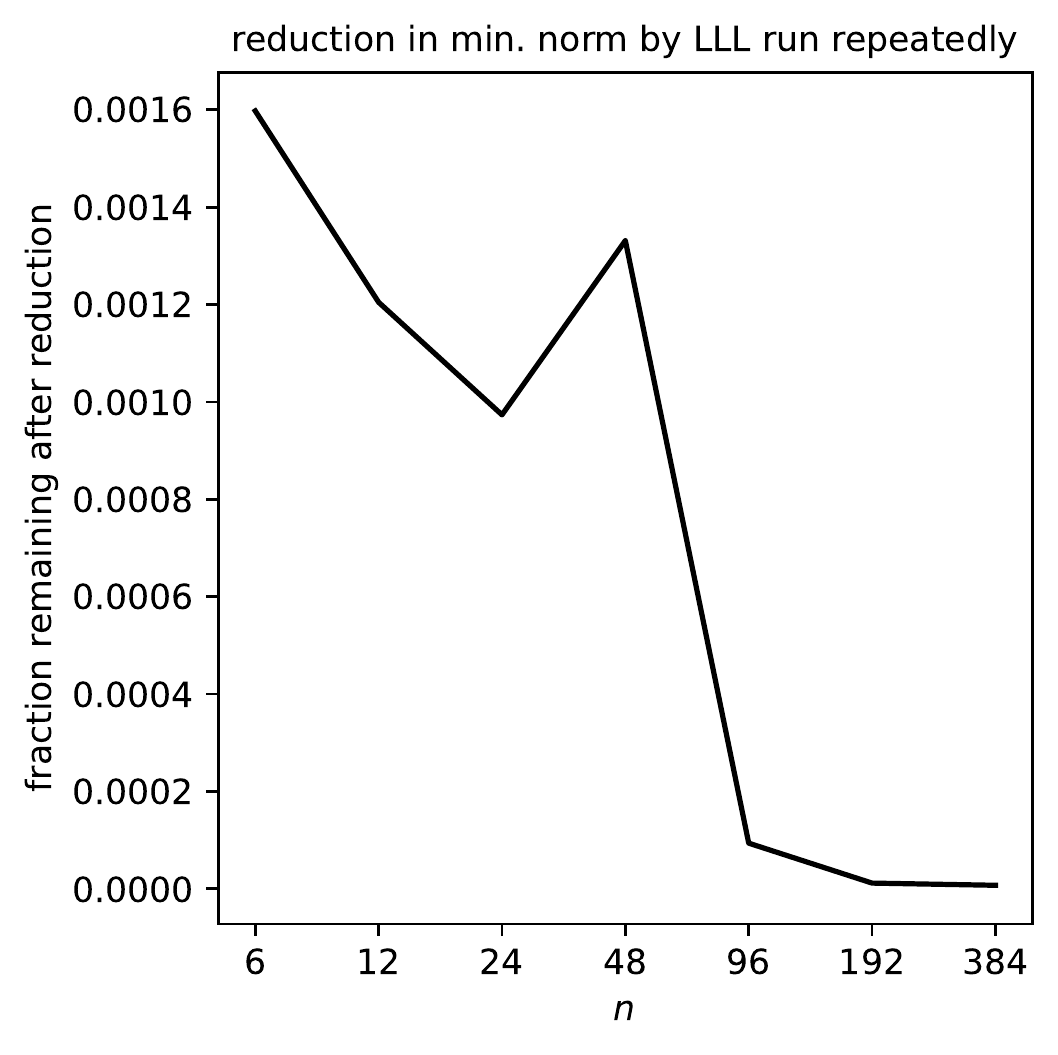}}
{\includegraphics[width=0.495\textwidth]{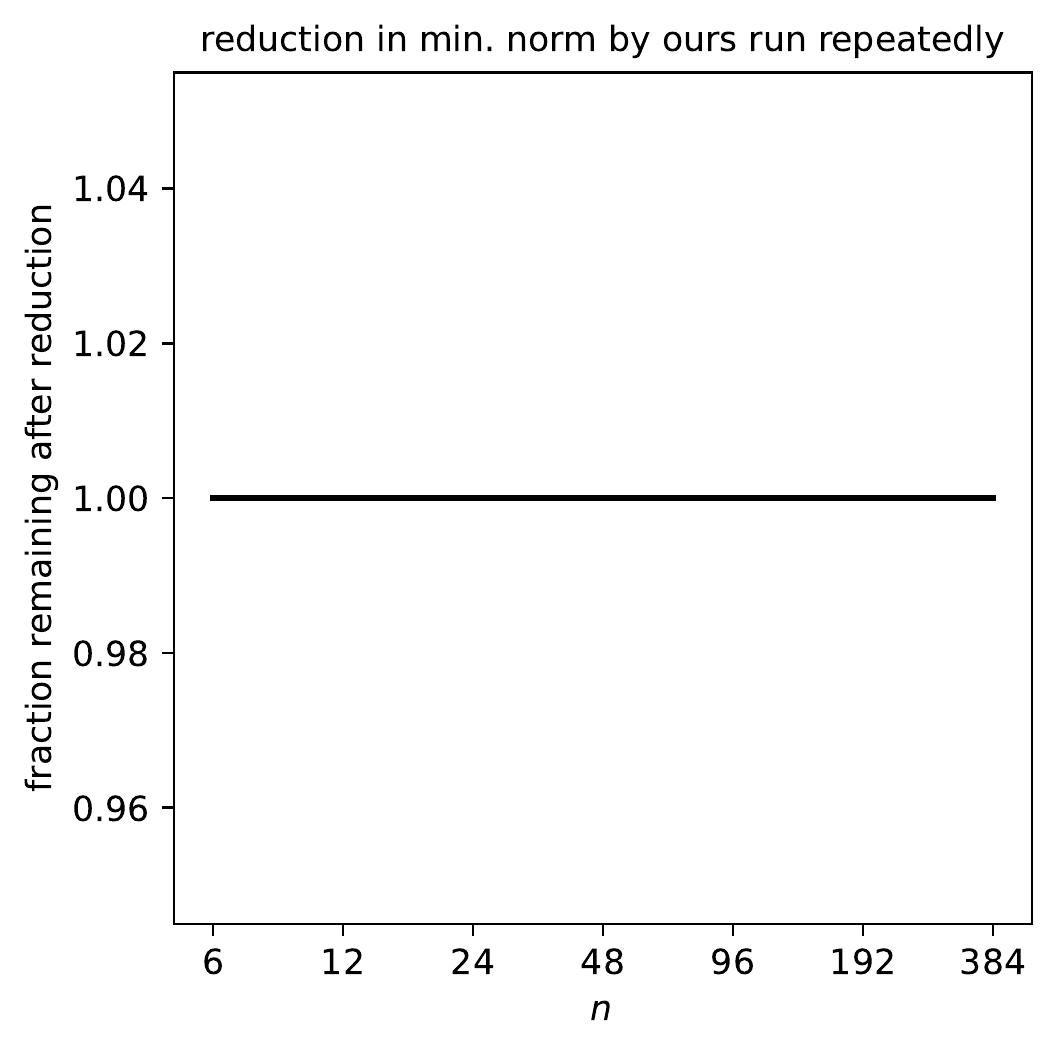}}

\end{centering}
\caption{$\delta = 1-10^{-1}$, $p = 2$, $q = 2^{31} - 1$}
\label{p2err1-1e-1-31}
\end{figure}

\begin{figure}
\begin{centering}
{\includegraphics[width=0.495\textwidth]{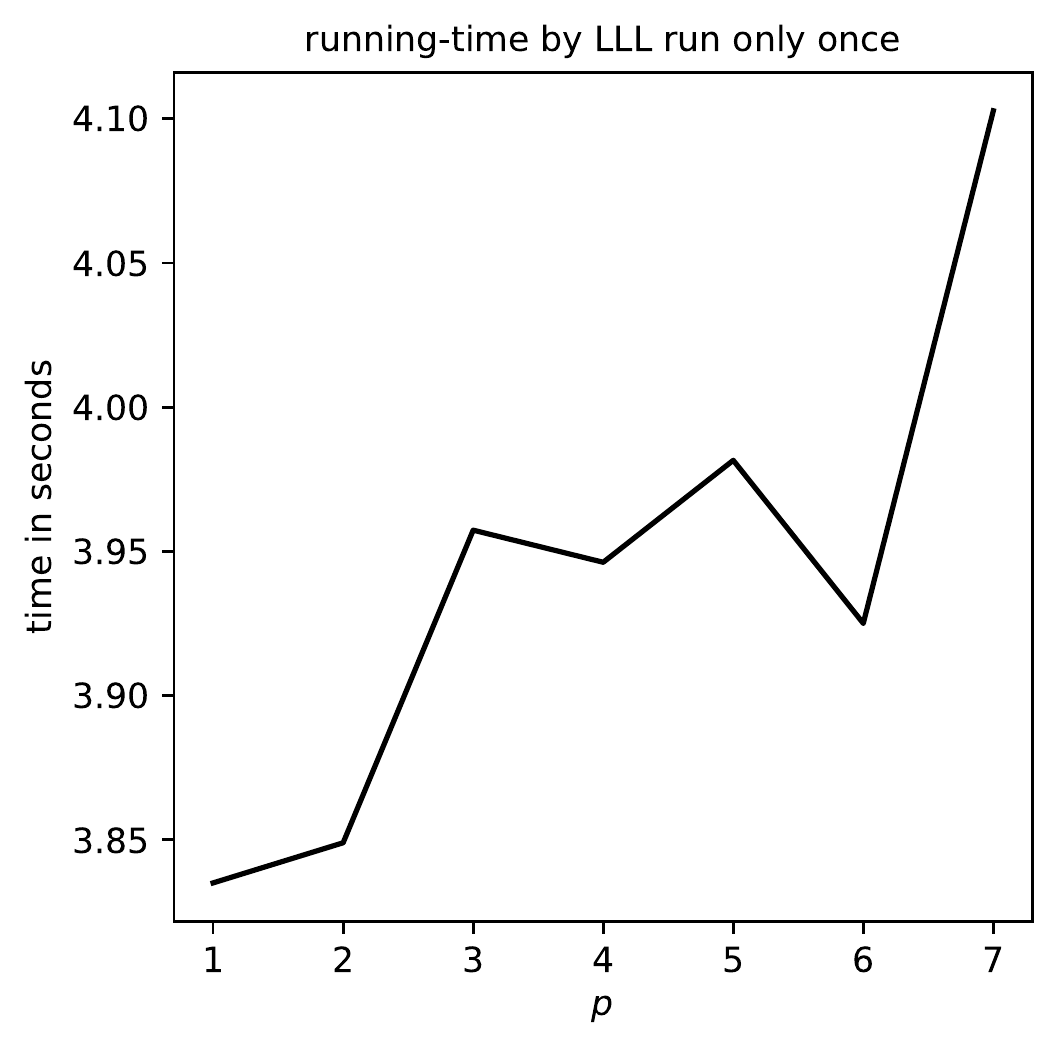}}
{\includegraphics[width=0.495\textwidth]{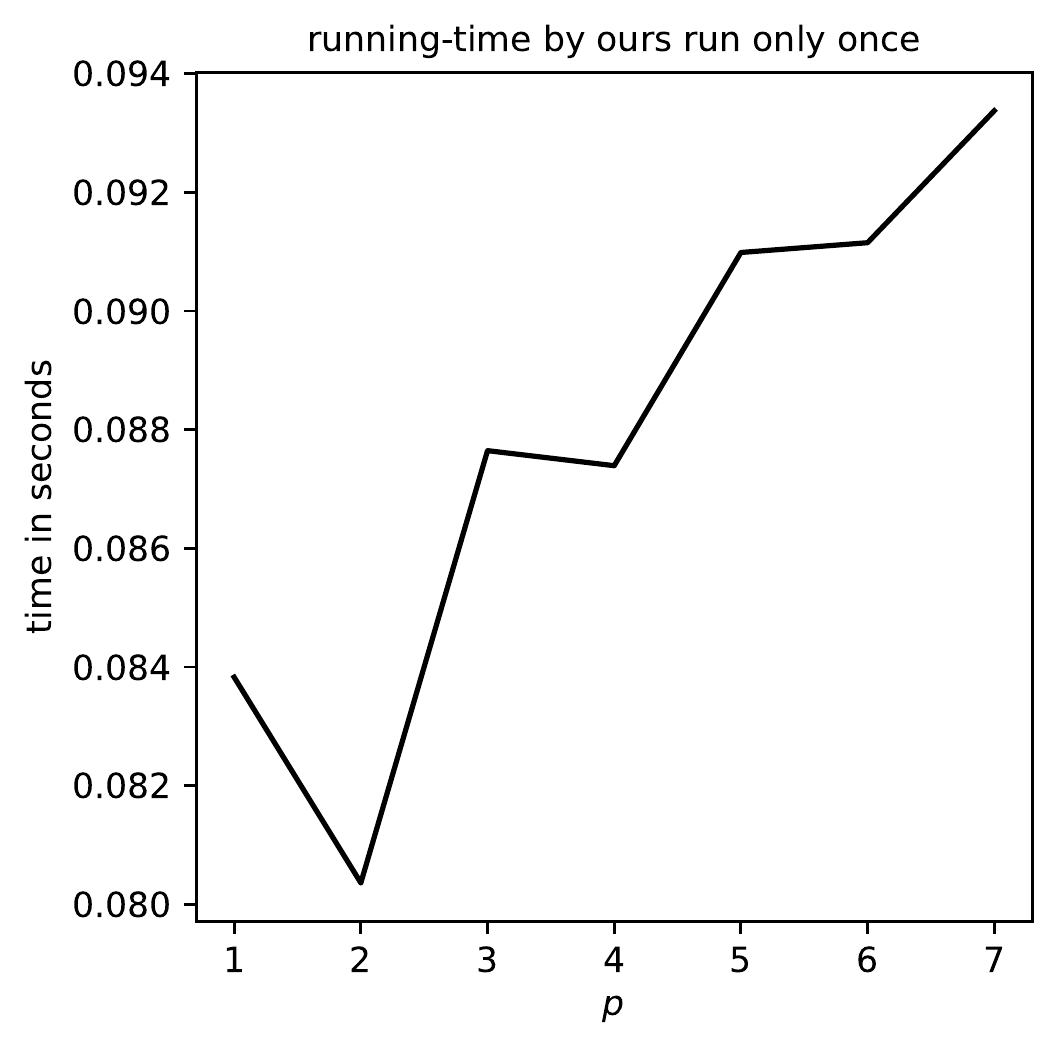}}

{\includegraphics[width=0.495\textwidth]{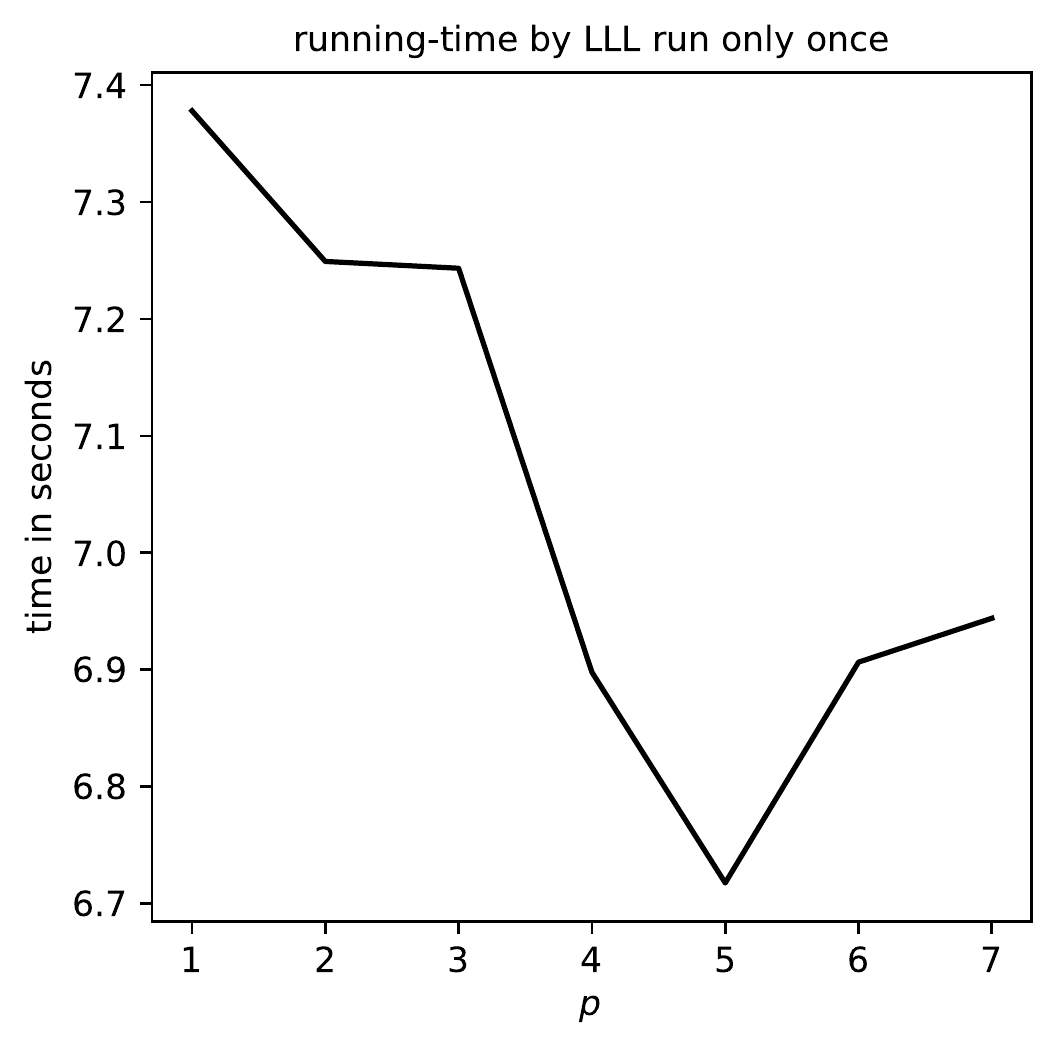}}
{\includegraphics[width=0.495\textwidth]{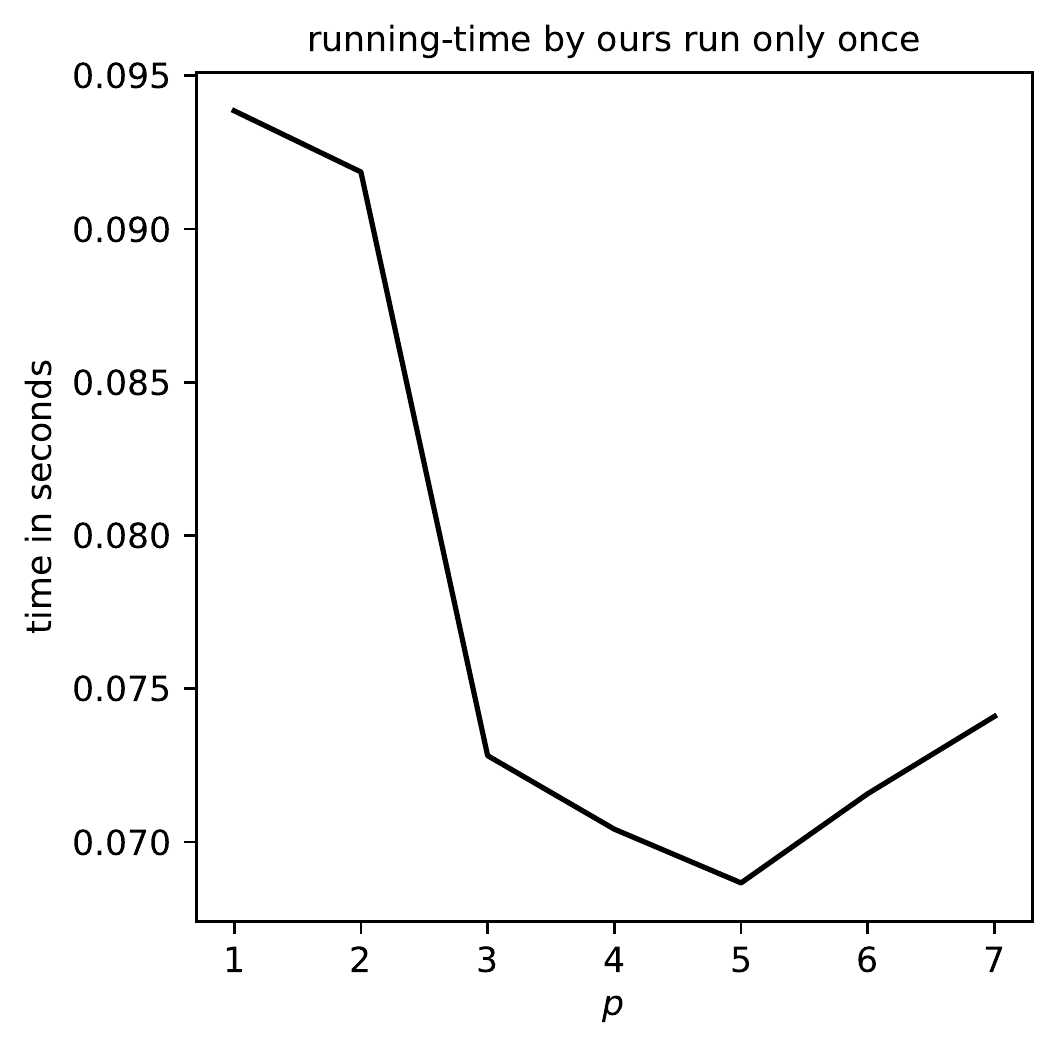}}

\end{centering}
\caption{$\delta = 1-10^{-15}$, $n = 192$;
         the upper plots are for $q = 2^{13} - 1$,
         the lower plots are for $q = 2^{31} - 1$ \dots\
         the vertical ranges of the plots on the left are very small,
         with the vertical variations displayed
         being statistically insignificant, wholly attributable to randomness
         in the computational environment.}
\label{pstime1-1e-15}
\end{figure}

\begin{figure}
\begin{centering}
{\includegraphics[width=0.495\textwidth]{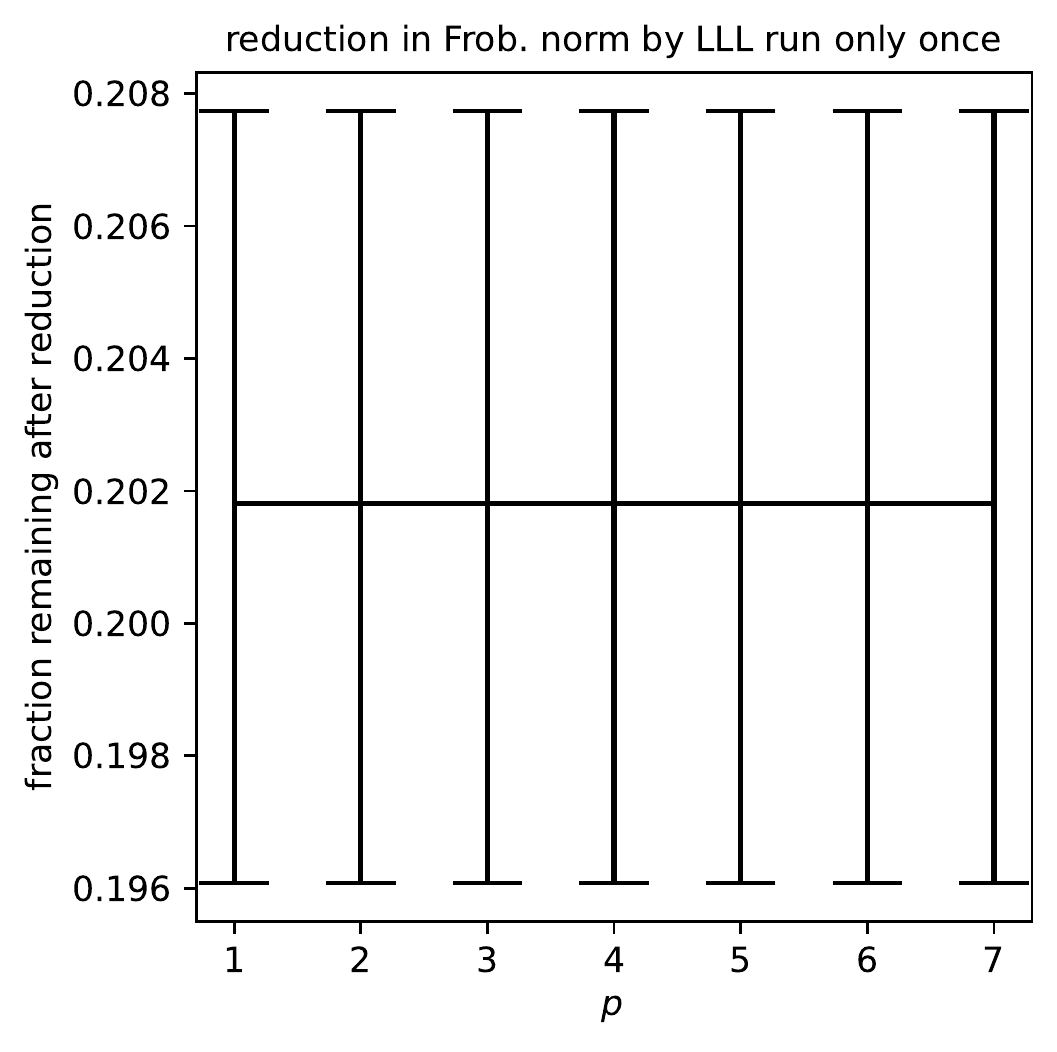}}
{\includegraphics[width=0.495\textwidth]{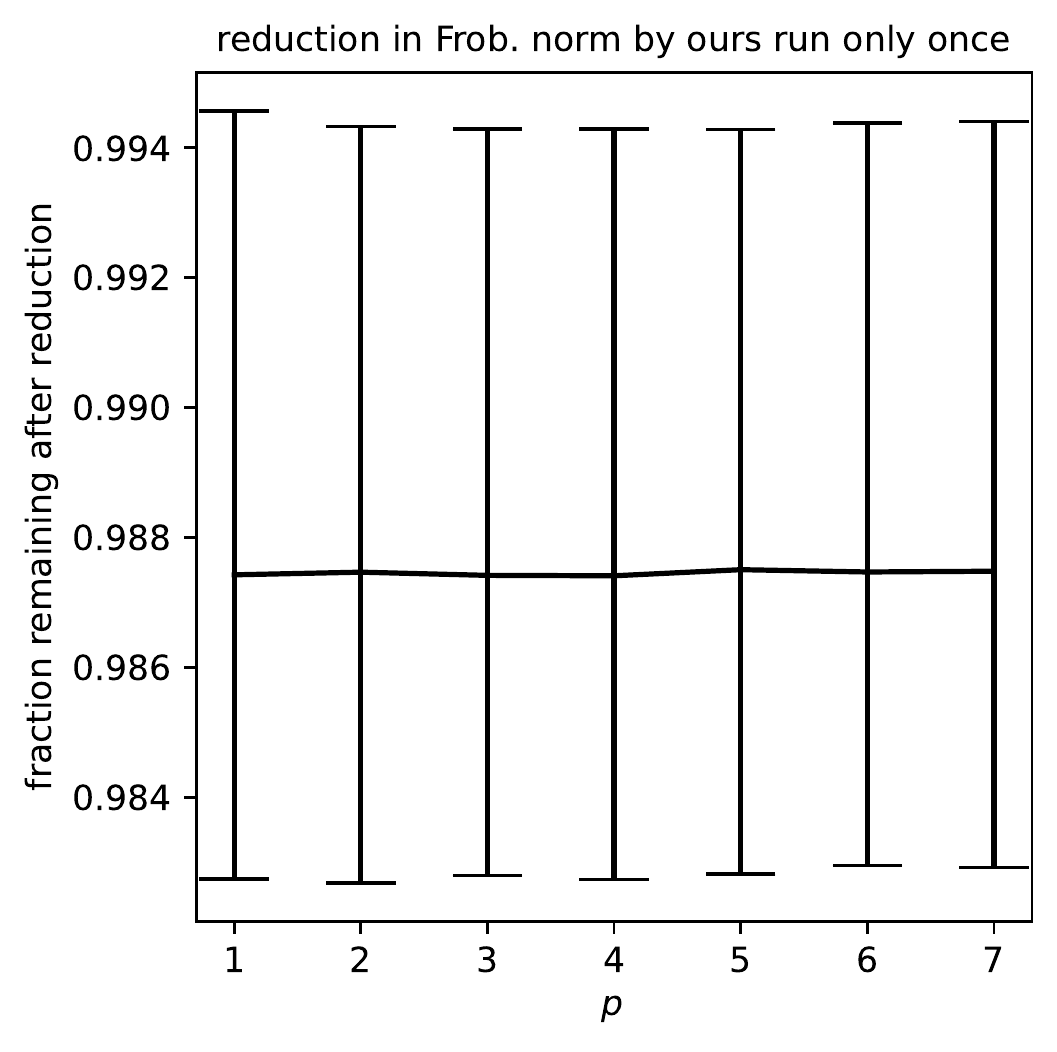}}

{\includegraphics[width=0.495\textwidth]{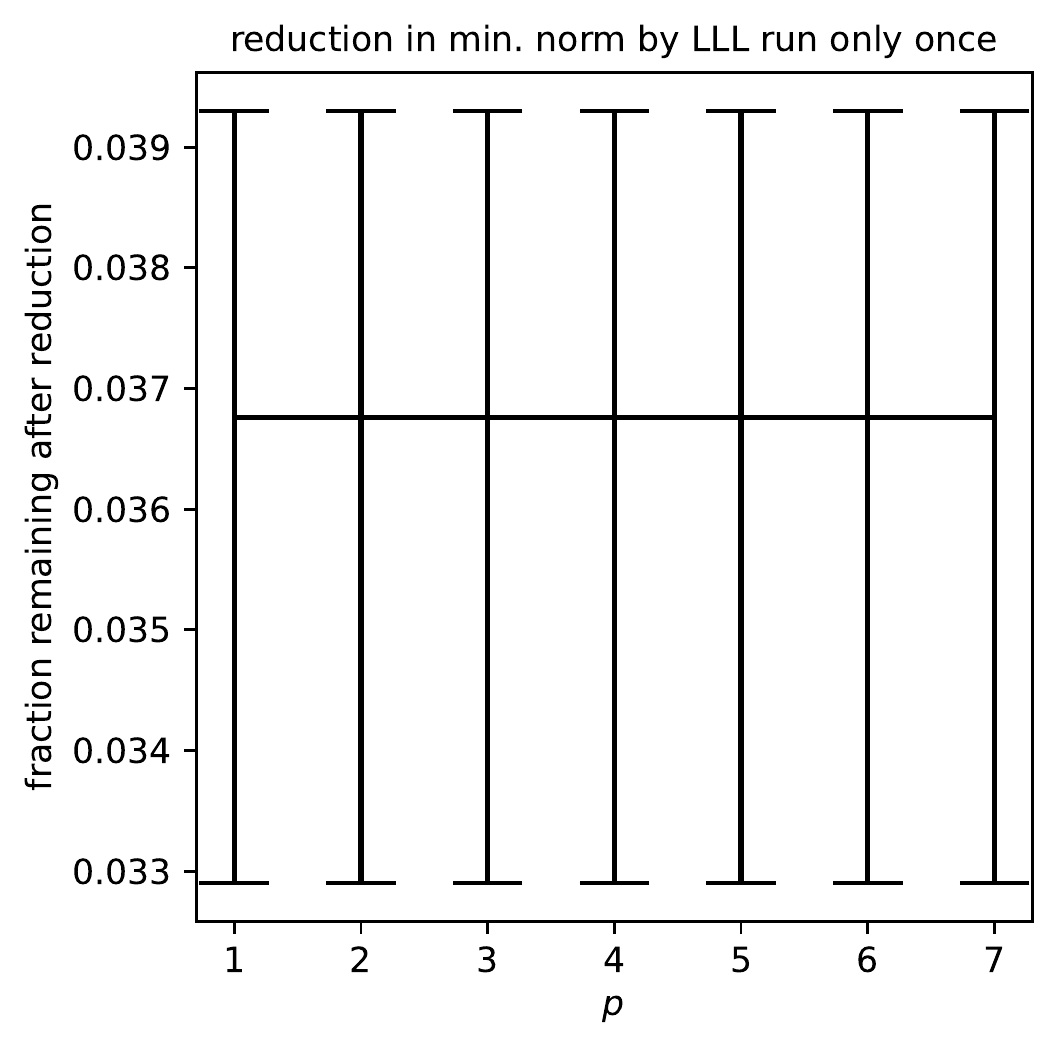}}
{\includegraphics[width=0.495\textwidth]{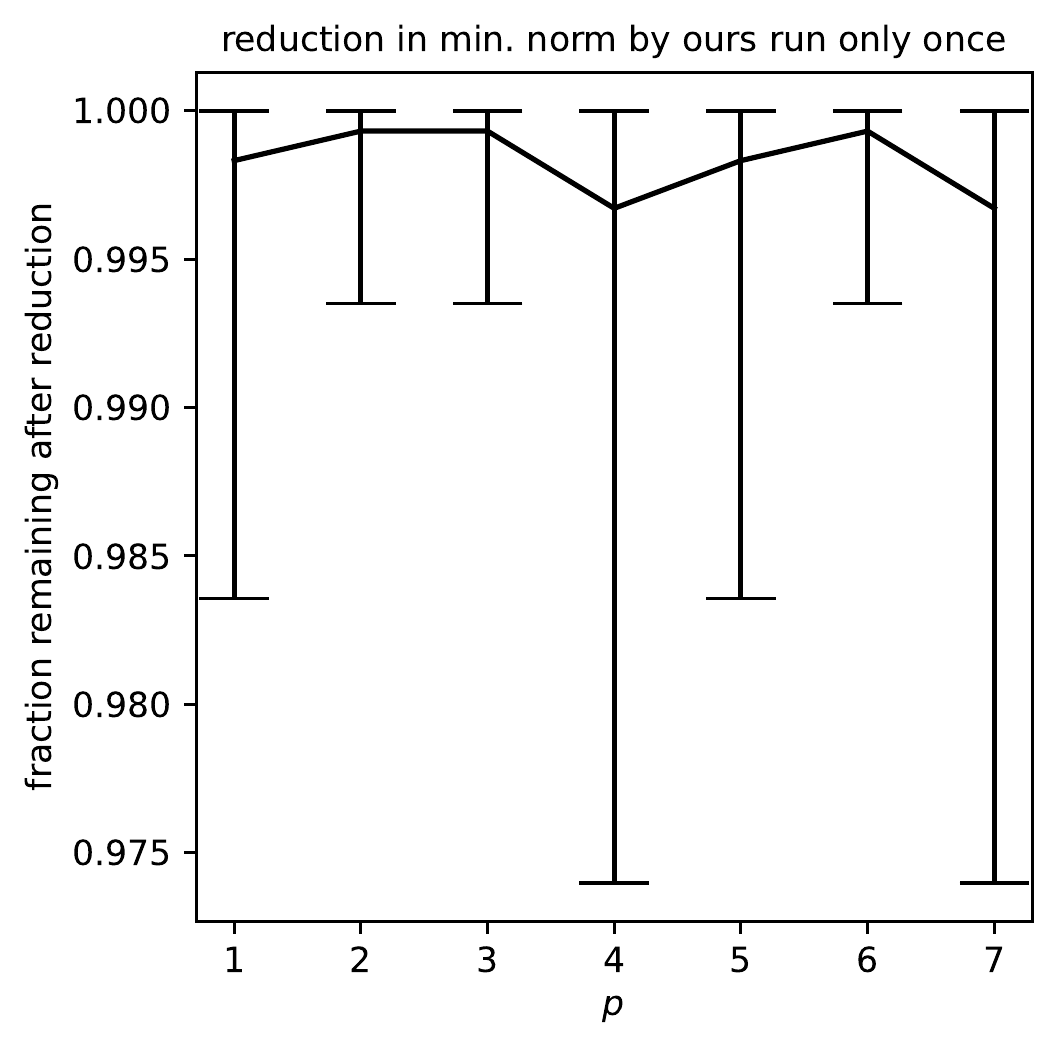}}

\end{centering}
\caption{$\delta = 1-10^{-15}$, $n = 192$, $q = 2^{13} - 1$}
\end{figure}

\begin{figure}
\begin{centering}
{\includegraphics[width=0.495\textwidth]{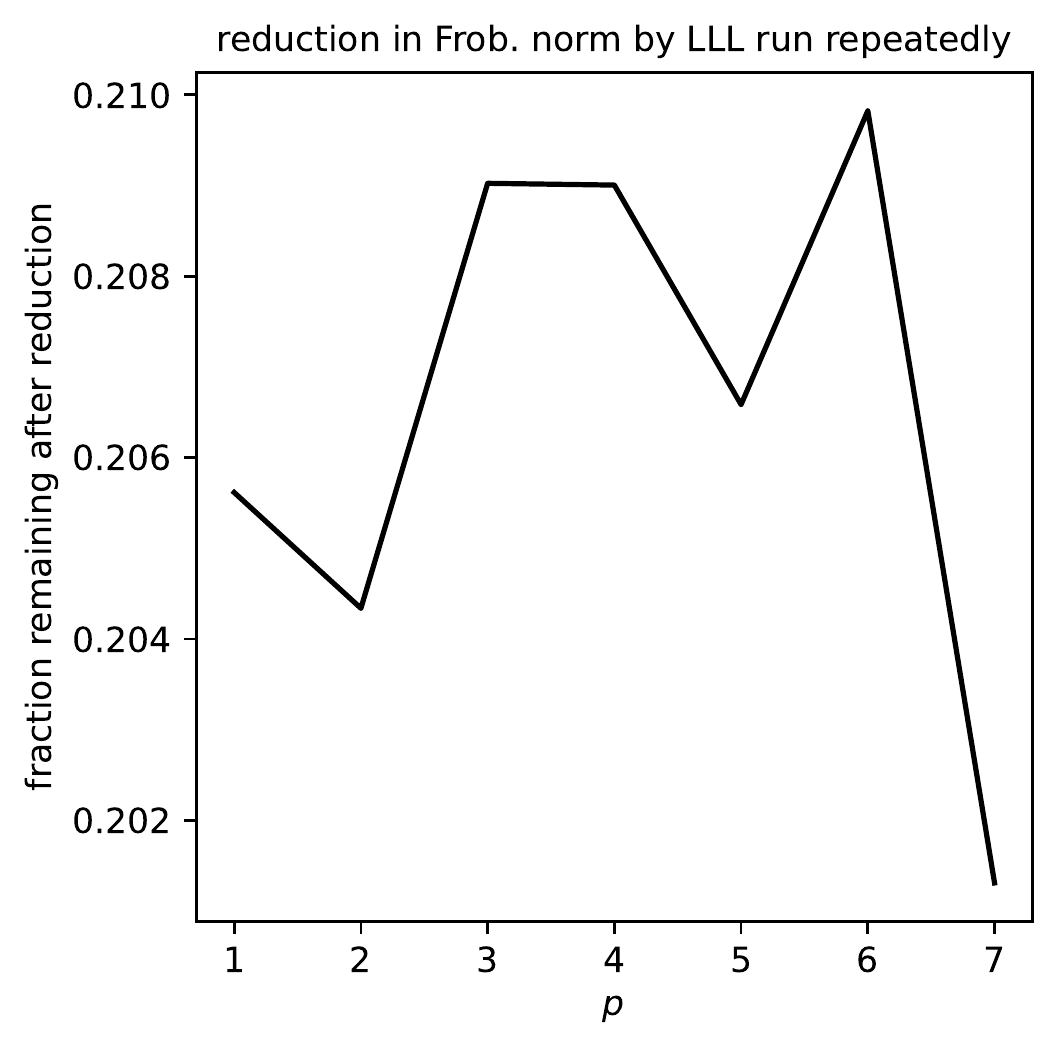}}
{\includegraphics[width=0.495\textwidth]{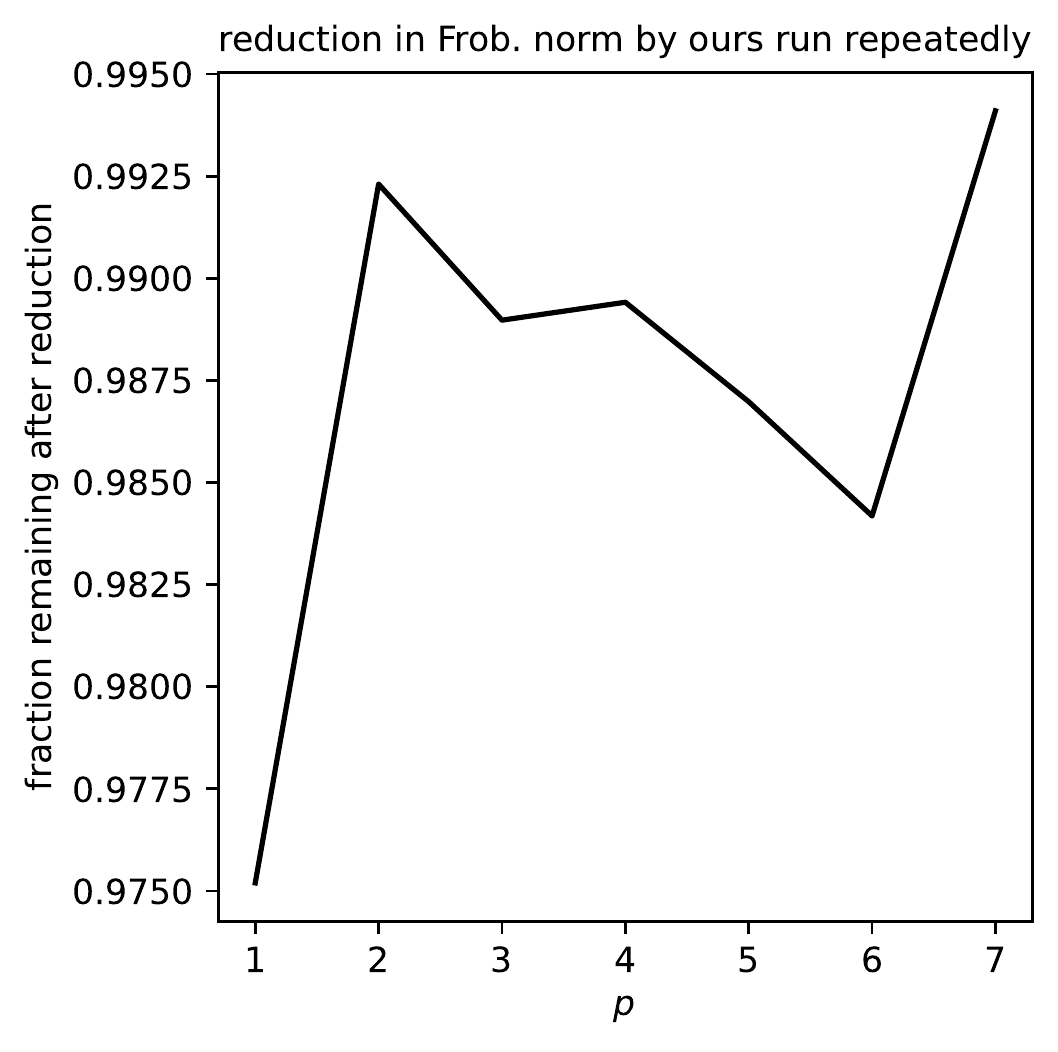}}

{\includegraphics[width=0.495\textwidth]{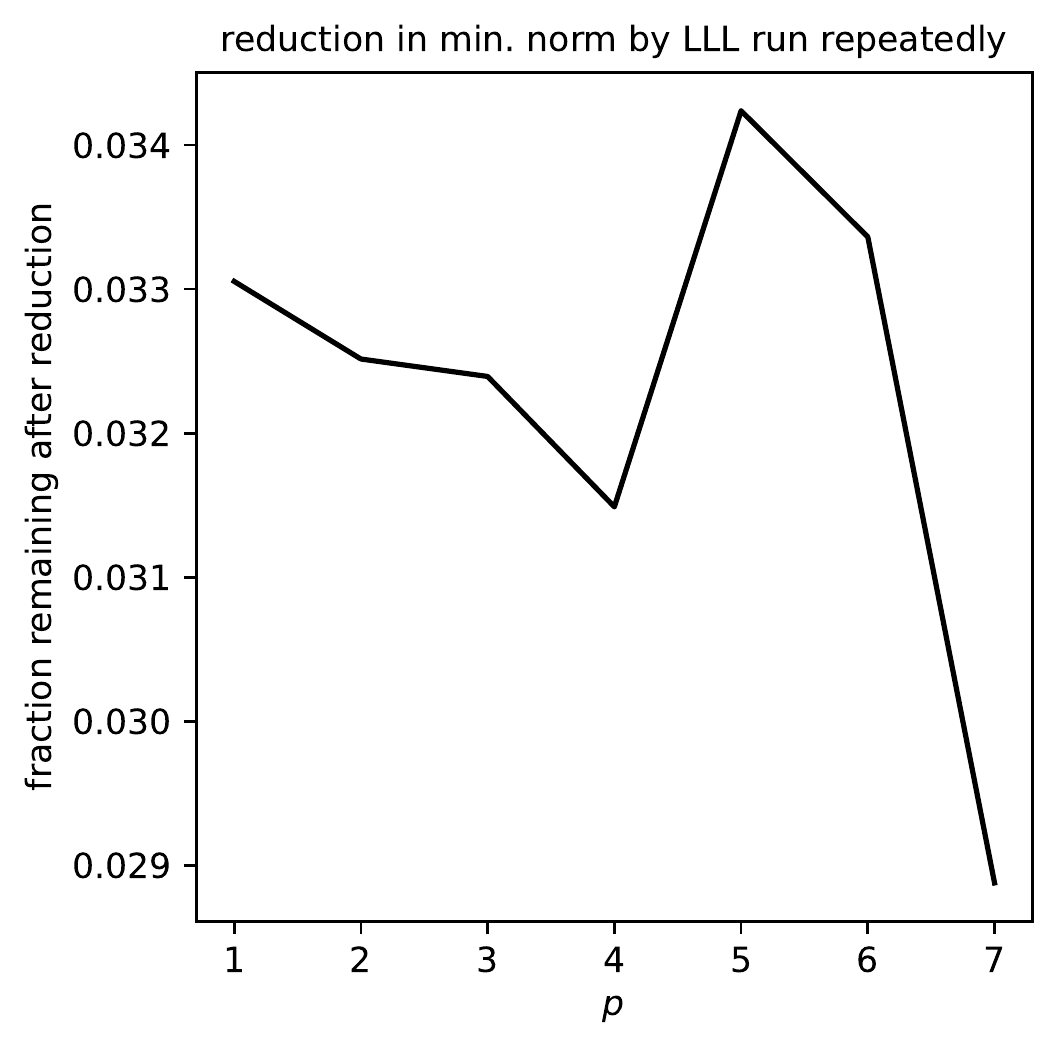}}
{\includegraphics[width=0.495\textwidth]{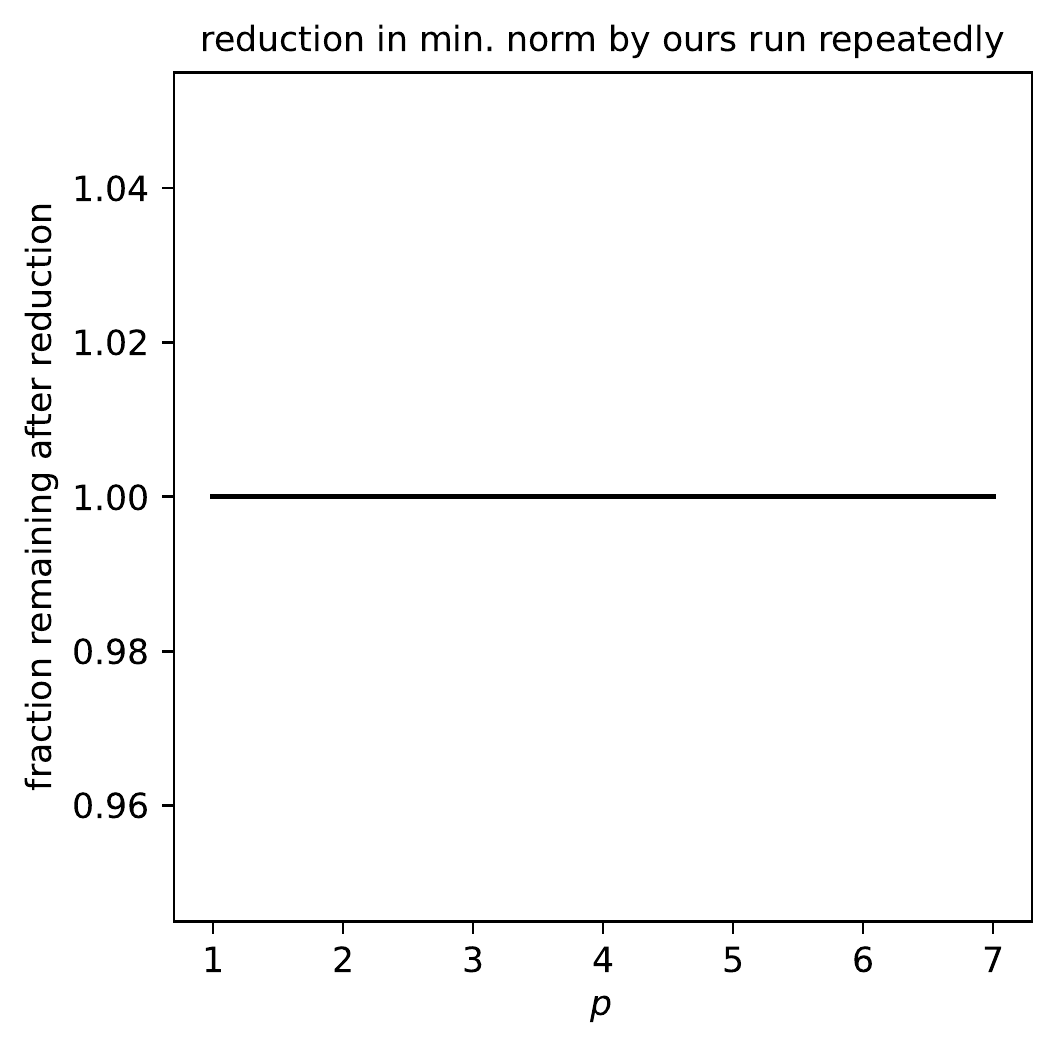}}

\end{centering}
\caption{$\delta = 1-10^{-15}$, $n = 192$, $q = 2^{13} - 1$}
\end{figure}

\begin{figure}
\begin{centering}
{\includegraphics[width=0.495\textwidth]{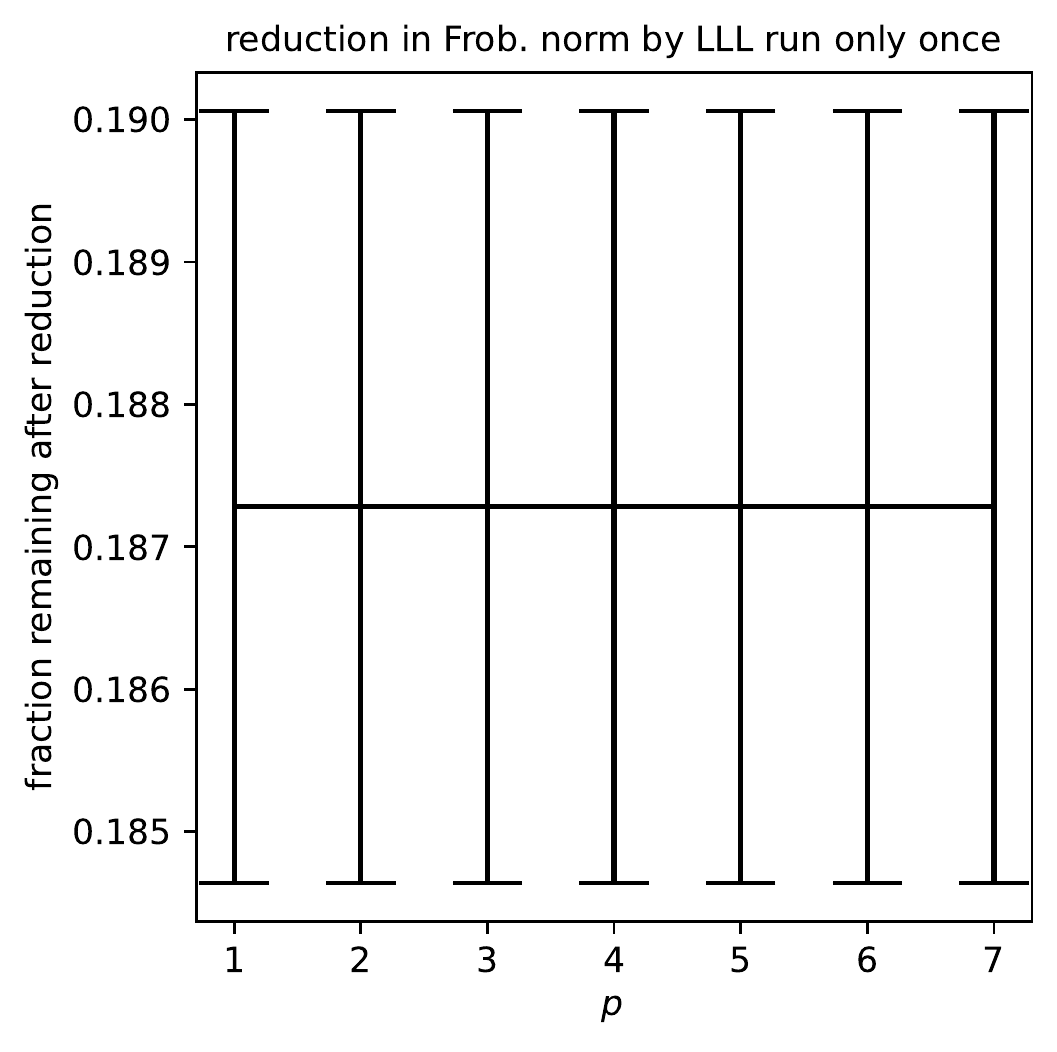}}
{\includegraphics[width=0.495\textwidth]{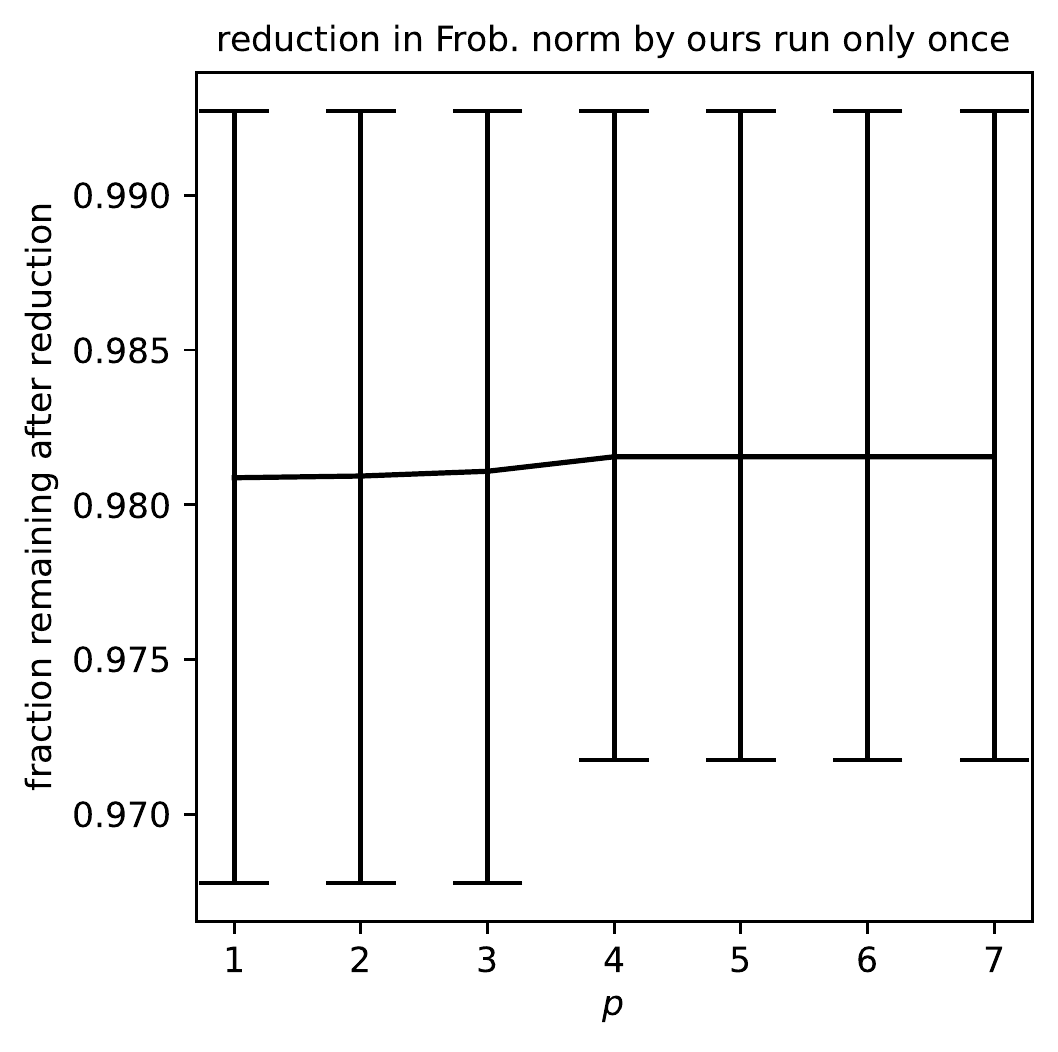}}

{\includegraphics[width=0.495\textwidth]{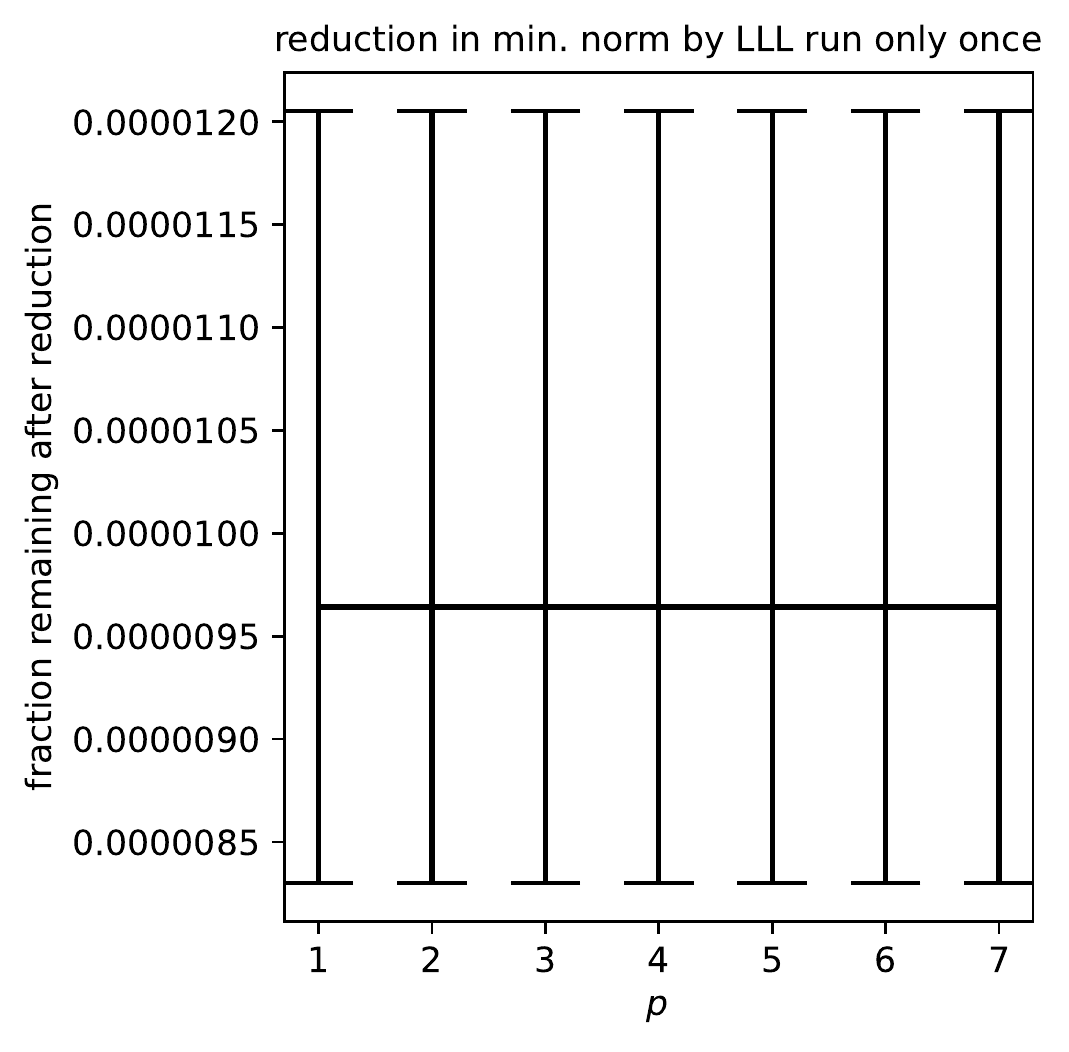}}
{\includegraphics[width=0.495\textwidth]{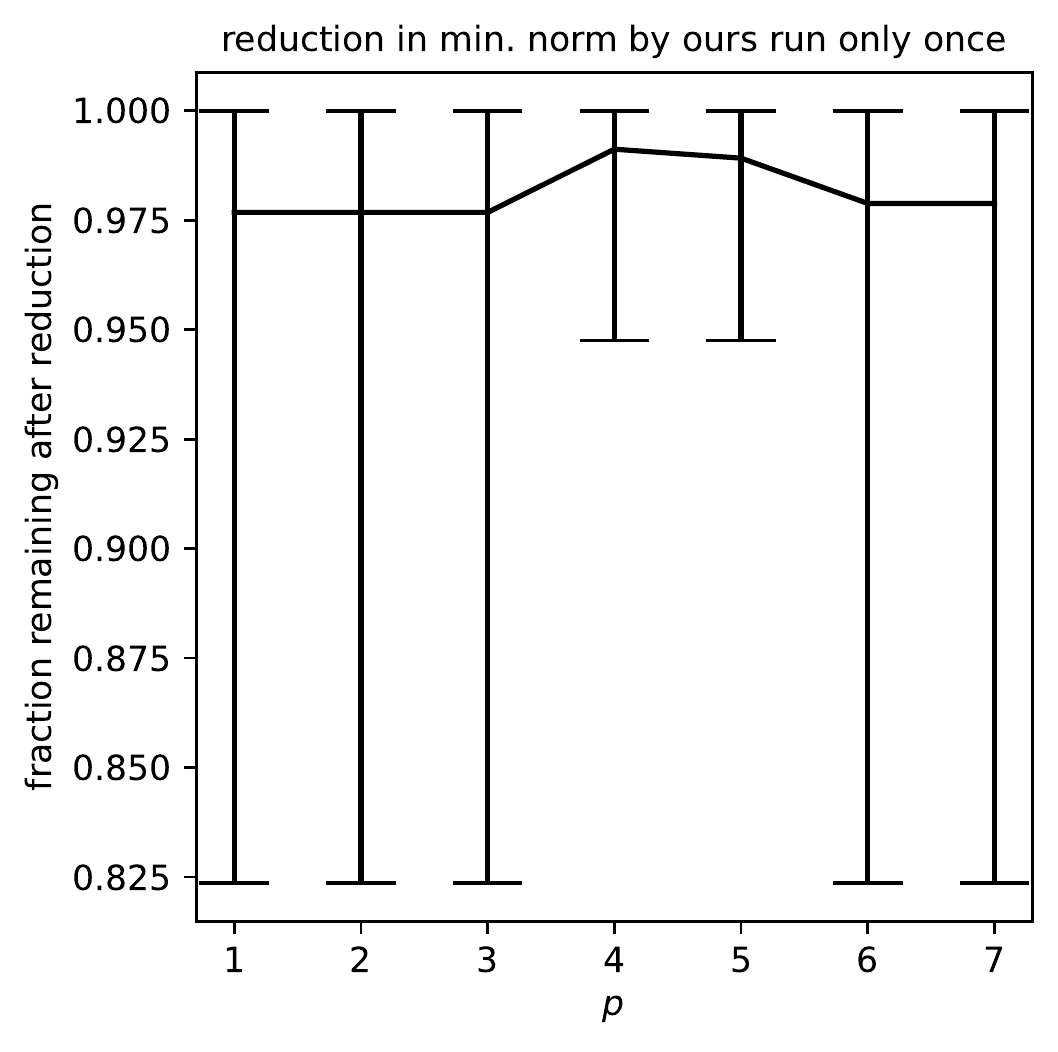}}

\end{centering}
\caption{$\delta = 1-10^{-15}$, $n = 192$, $q = 2^{31} - 1$}
\end{figure}

\begin{figure}
\begin{centering}
{\includegraphics[width=0.495\textwidth]{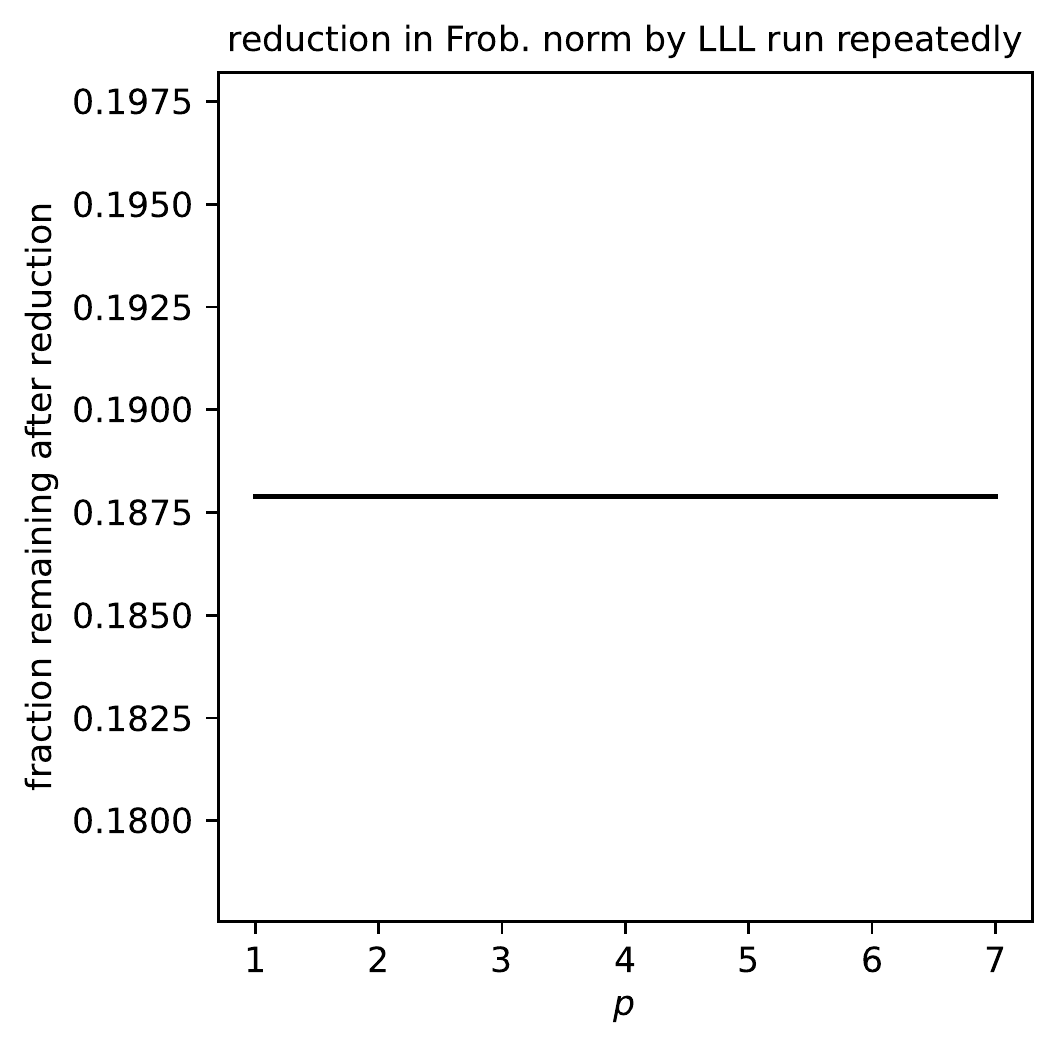}}
{\includegraphics[width=0.495\textwidth]{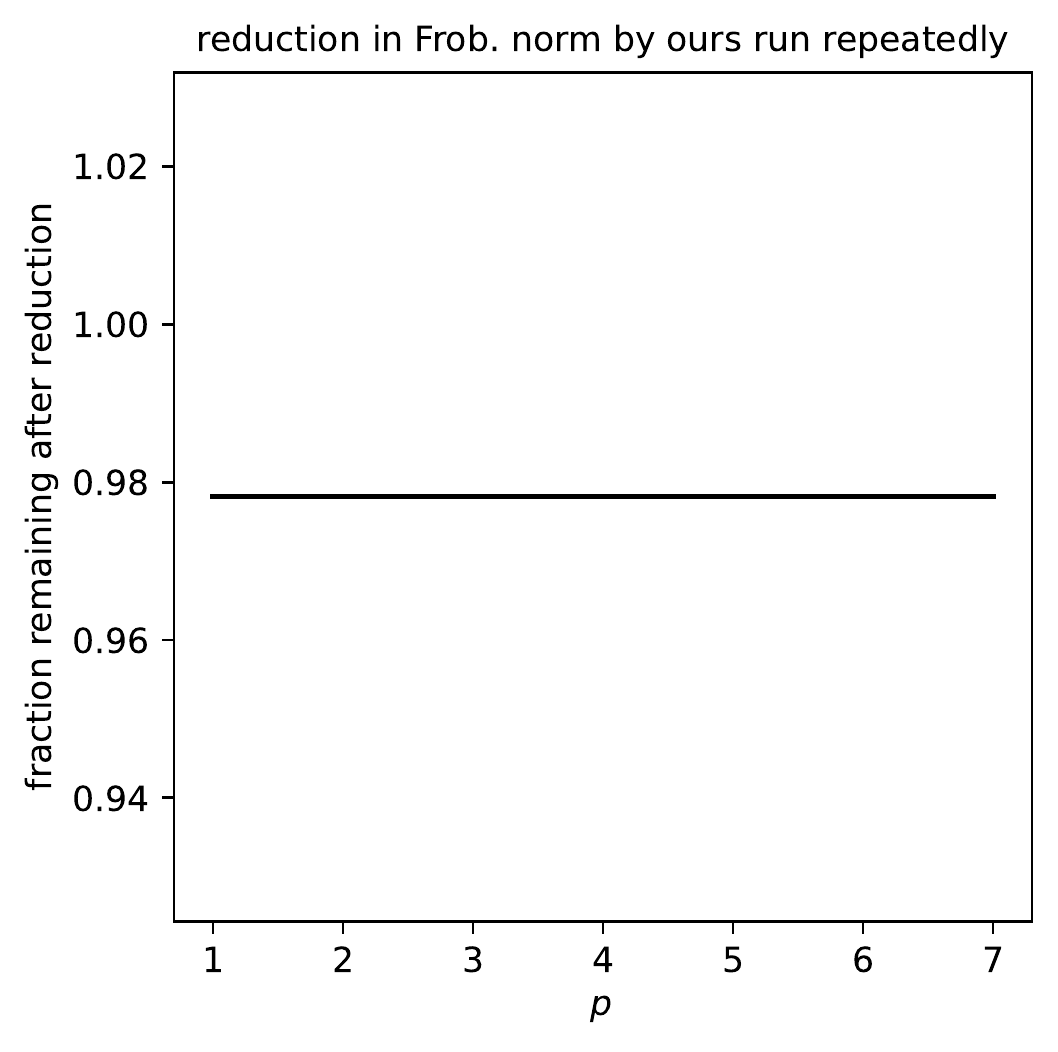}}

{\includegraphics[width=0.495\textwidth]{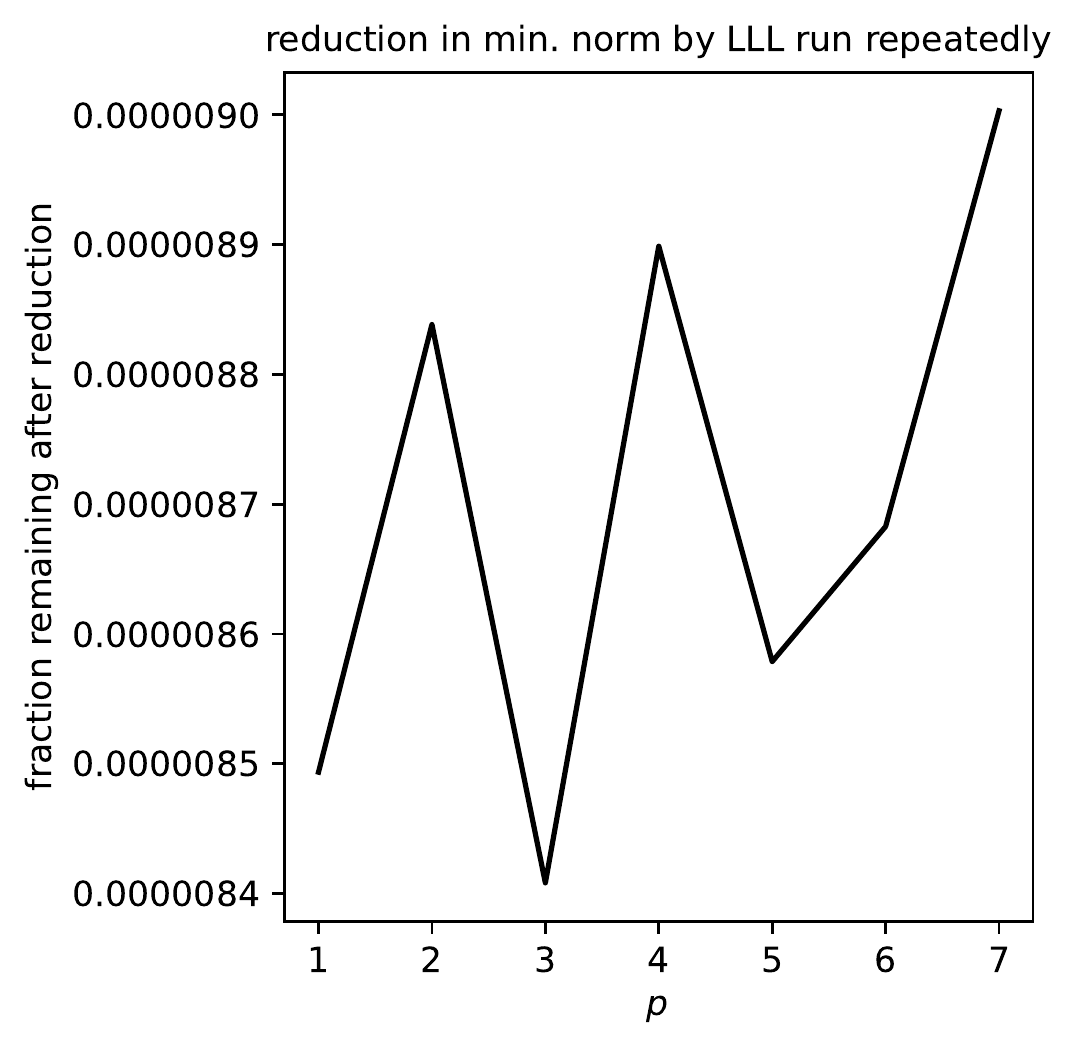}}
{\includegraphics[width=0.495\textwidth]{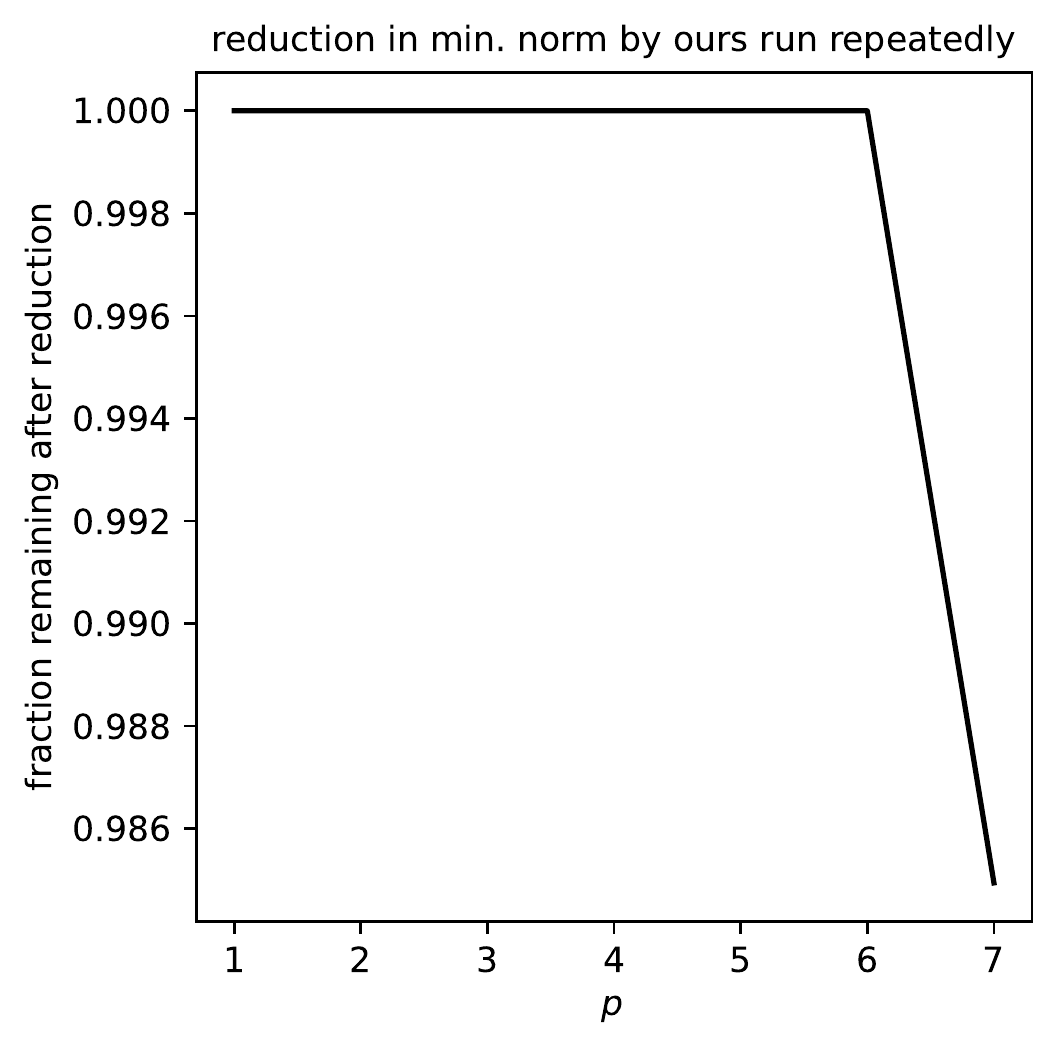}}

\end{centering}
\caption{$\delta = 1-10^{-15}$, $n = 192$, $q = 2^{31} - 1$}
\label{pserr1-1e-15-31}
\end{figure}

\begin{figure}
\begin{centering}
{\includegraphics[width=0.495\textwidth]{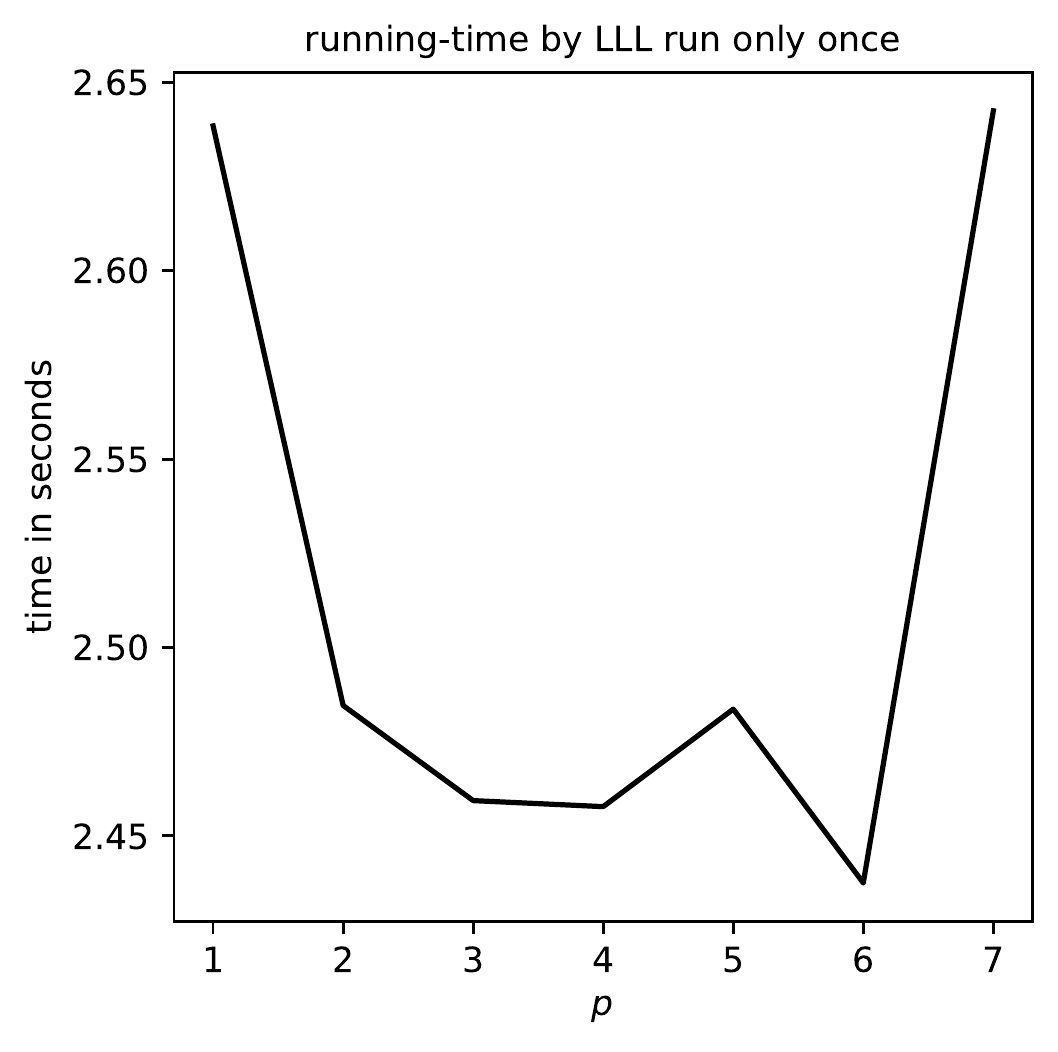}}
{\includegraphics[width=0.495\textwidth]{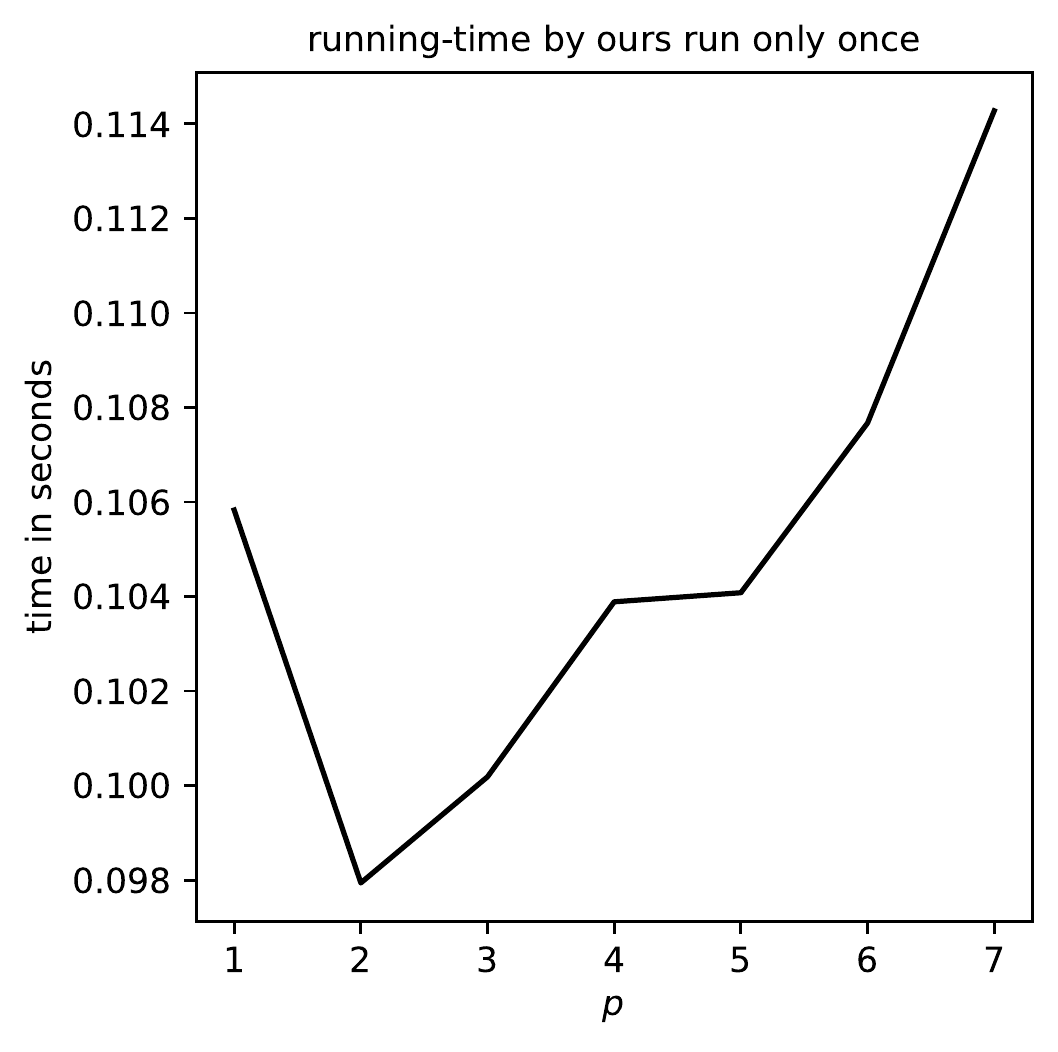}}

{\includegraphics[width=0.495\textwidth]{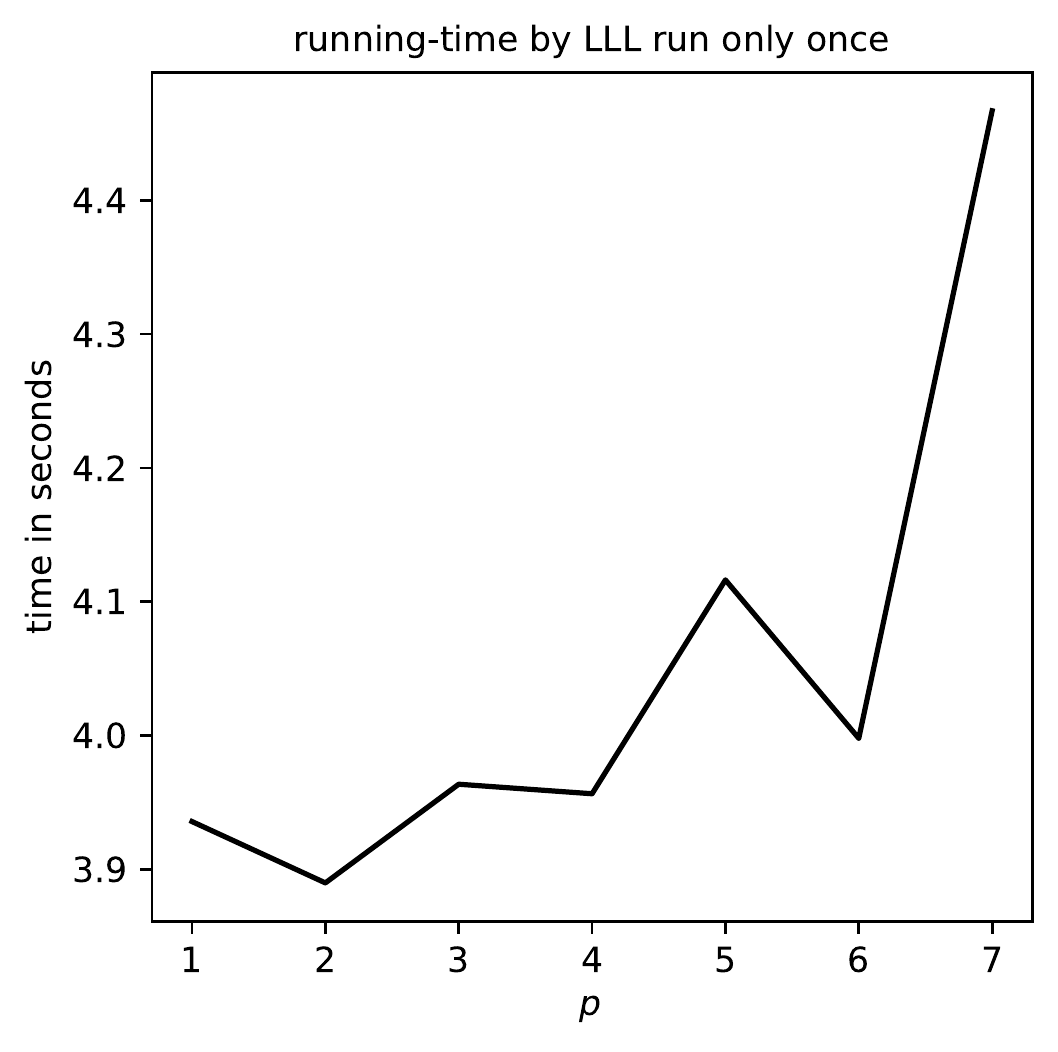}}
{\includegraphics[width=0.495\textwidth]{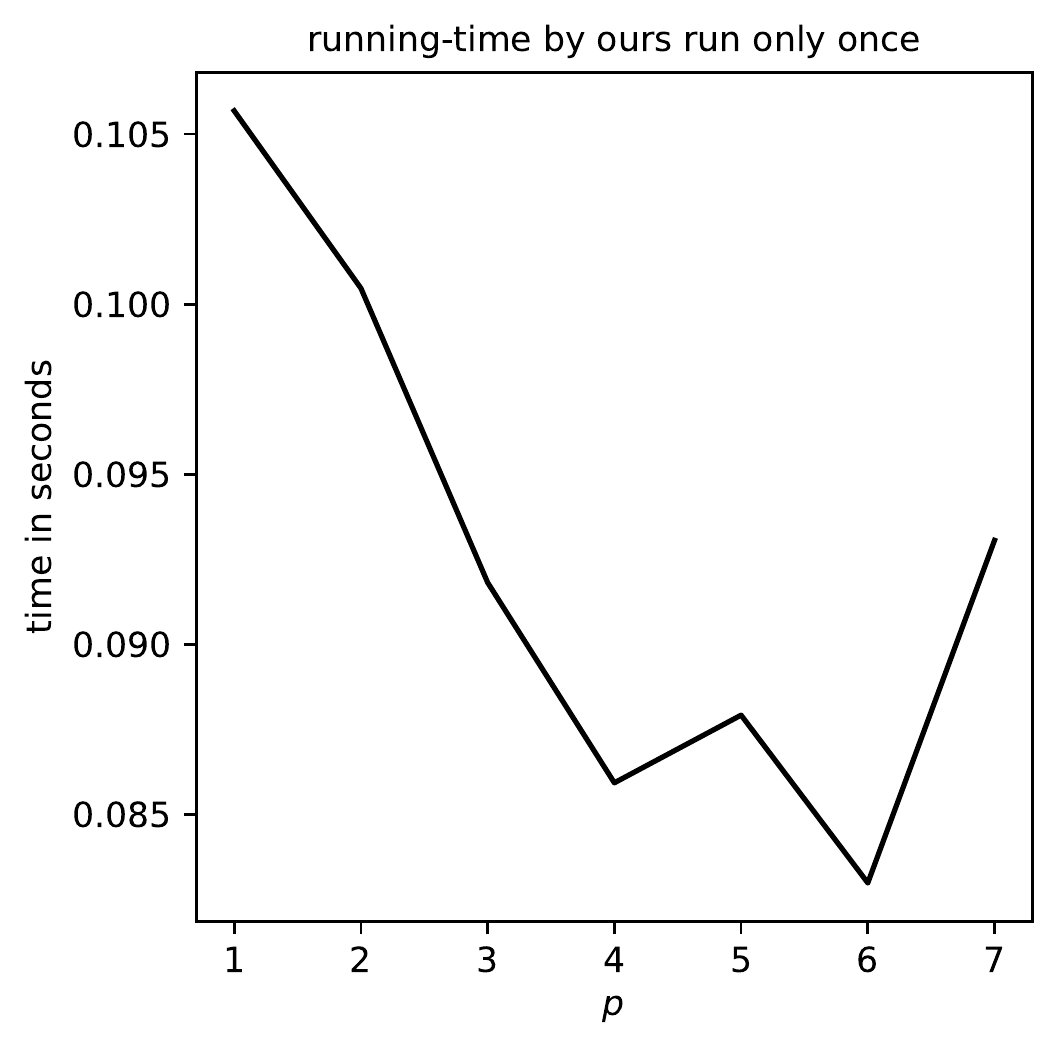}}

\end{centering}
\caption{$\delta = 1-10^{-1}$, $n = 192$;
         the upper plots are for $q = 2^{13} - 1$,
         the lower plots are for $q = 2^{31} - 1$ \dots\
         the vertical ranges of the plots on the left are very small,
         with the vertical variations displayed
         being statistically insignificant, wholly attributable to randomness
         in the computational environment.}
\label{pstime1-1e-1}
\end{figure}

\begin{figure}
\begin{centering}
{\includegraphics[width=0.495\textwidth]{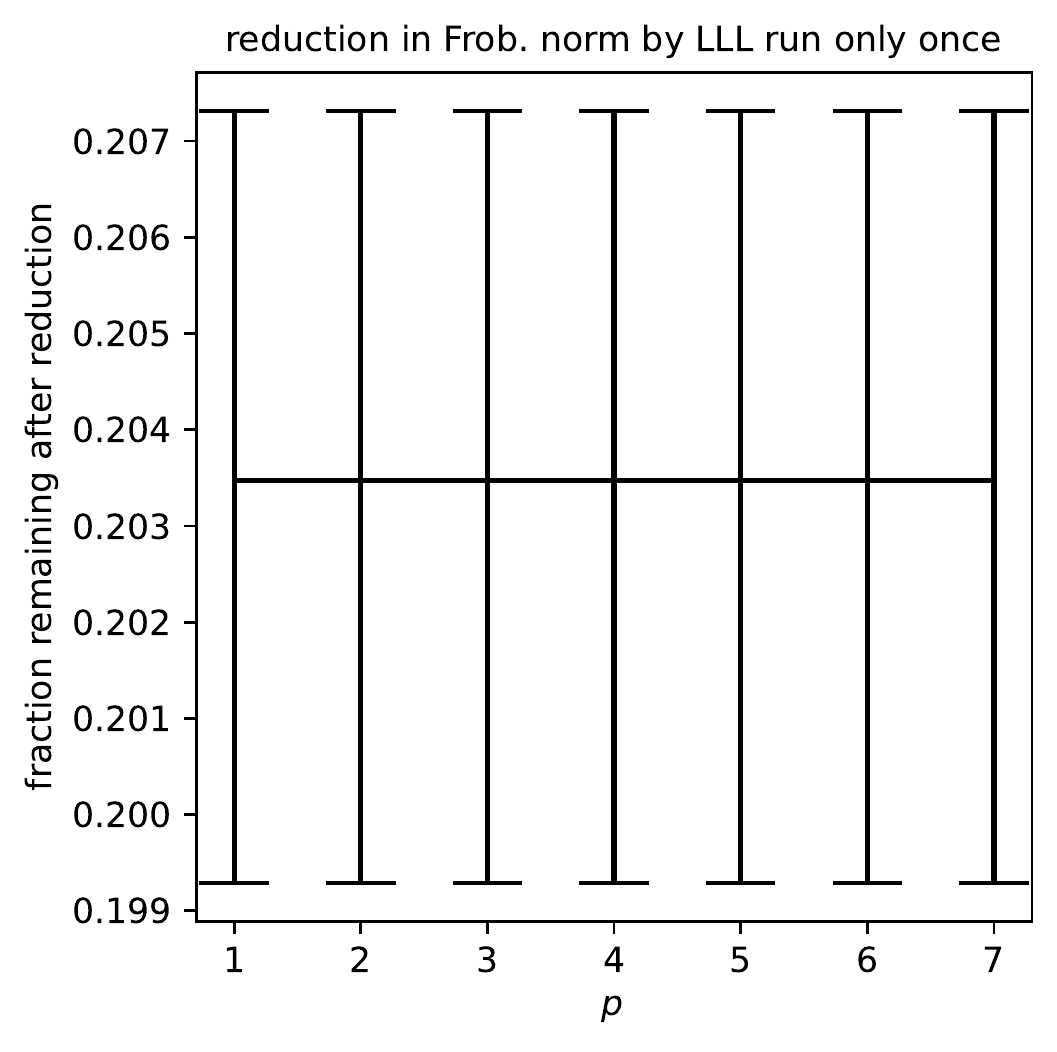}}
{\includegraphics[width=0.495\textwidth]{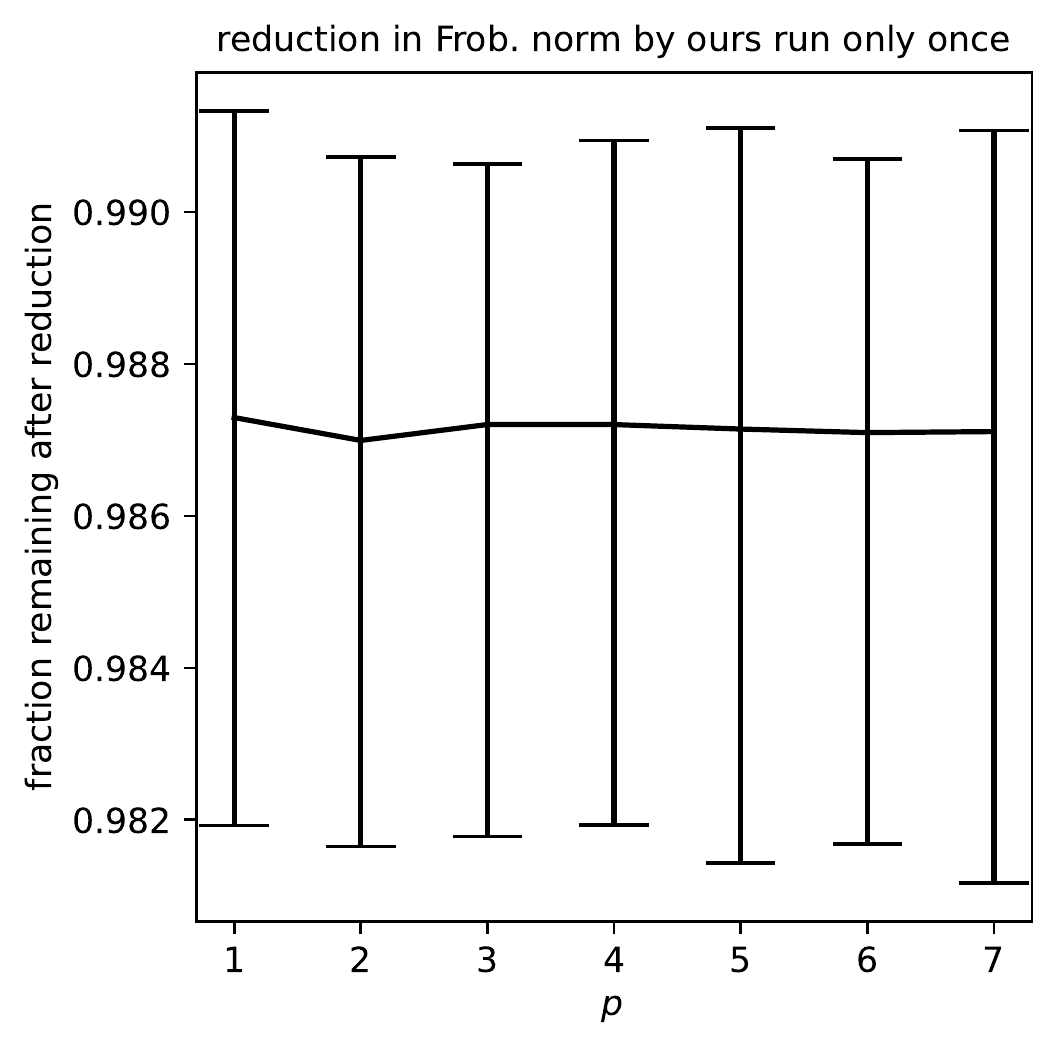}}

{\includegraphics[width=0.495\textwidth]{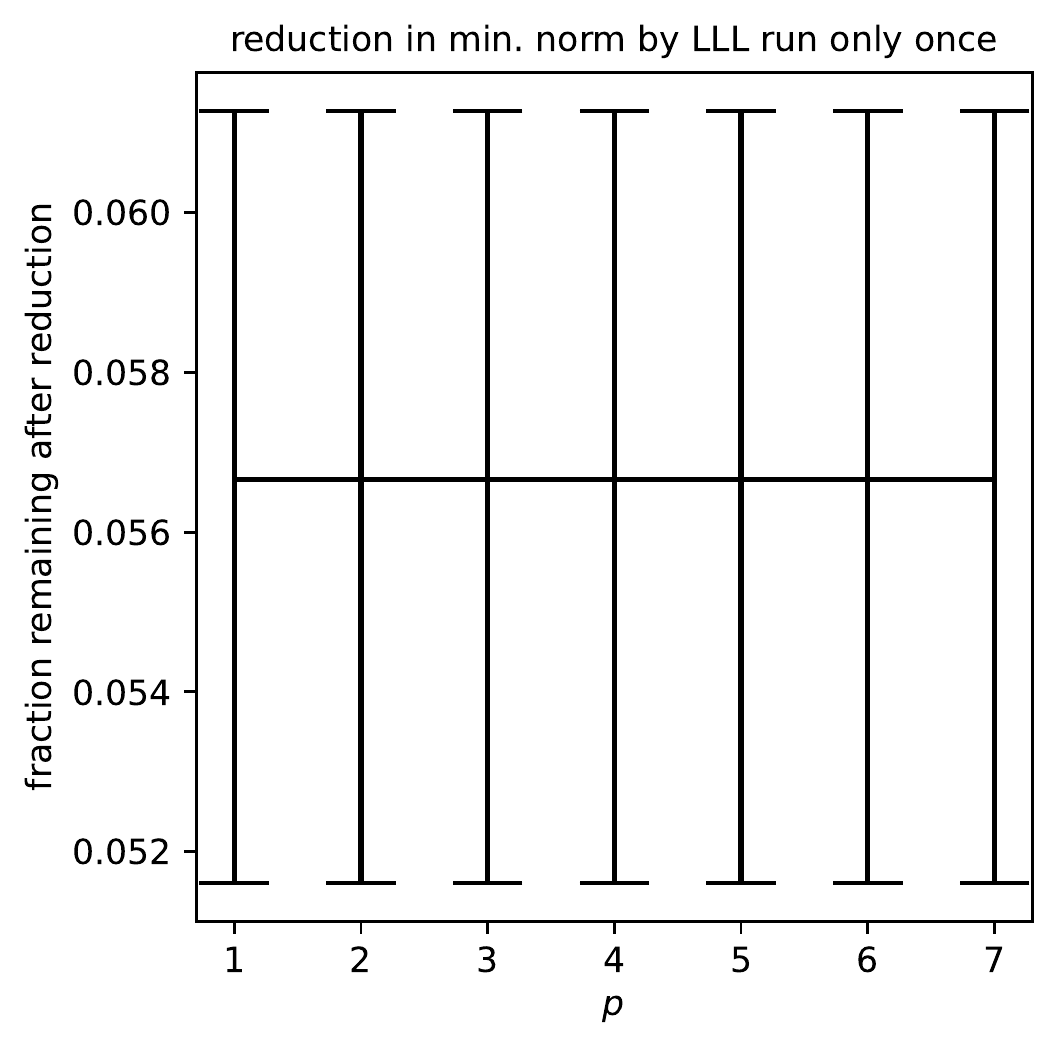}}
{\includegraphics[width=0.495\textwidth]{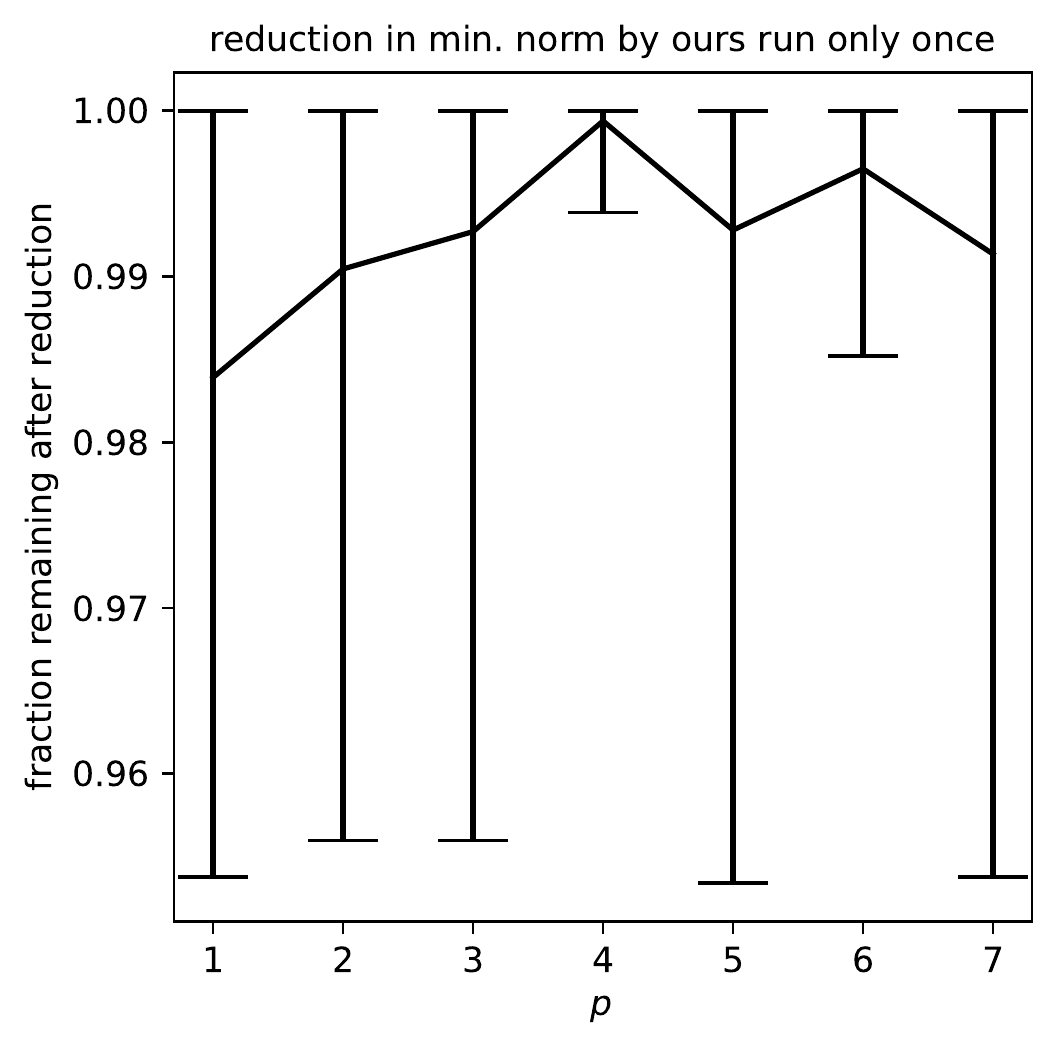}}

\end{centering}
\caption{$\delta = 1-10^{-1}$, $n = 192$, $q = 2^{13} - 1$}
\end{figure}

\begin{figure}
\begin{centering}
{\includegraphics[width=0.495\textwidth]{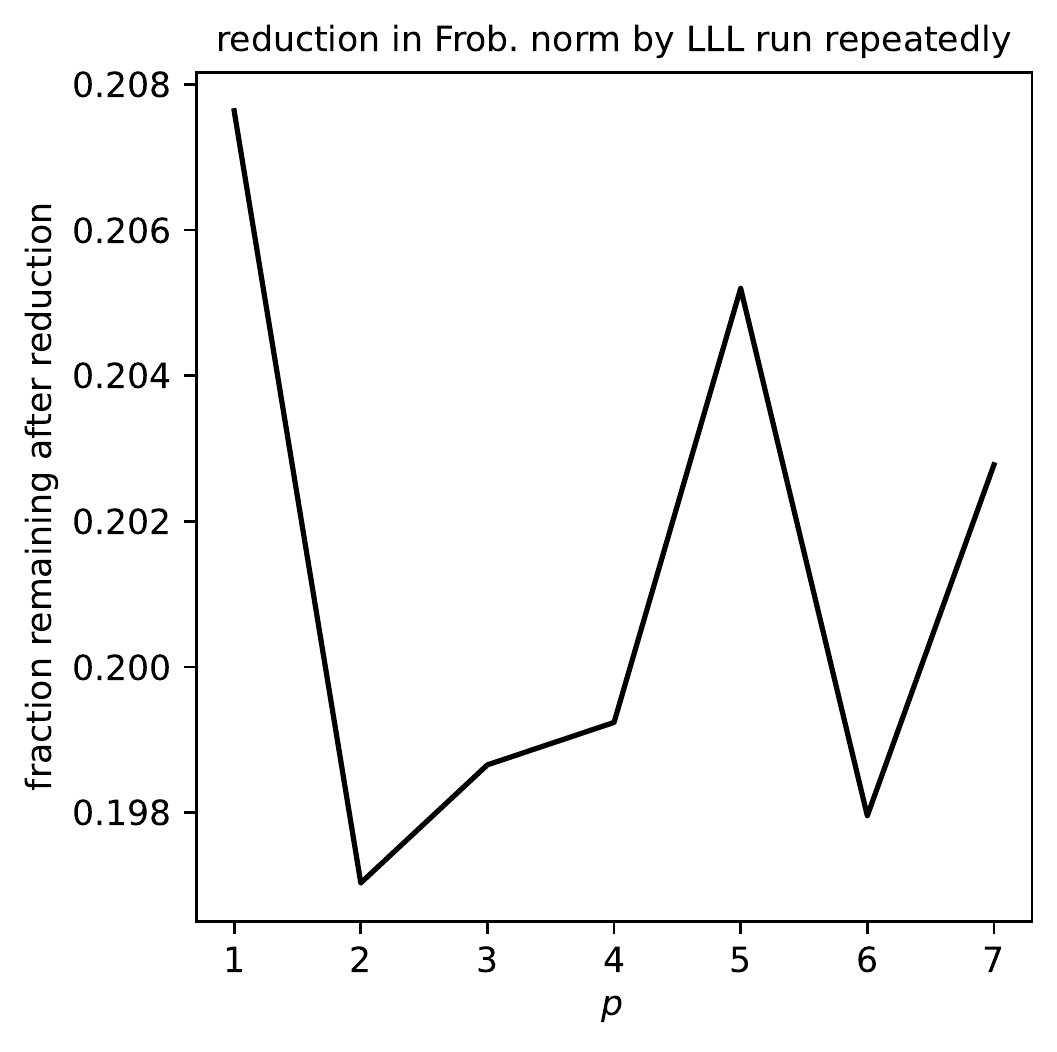}}
{\includegraphics[width=0.495\textwidth]{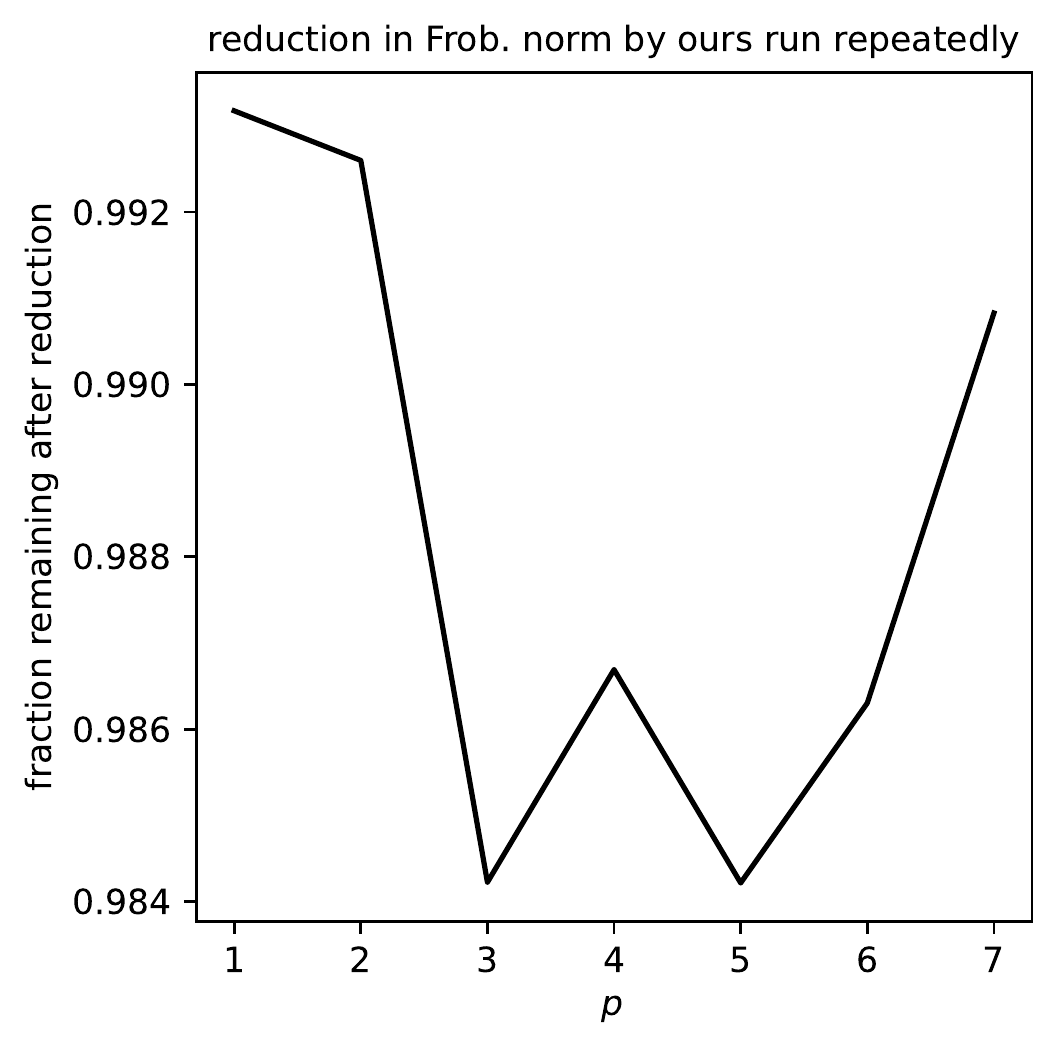}}

{\includegraphics[width=0.495\textwidth]{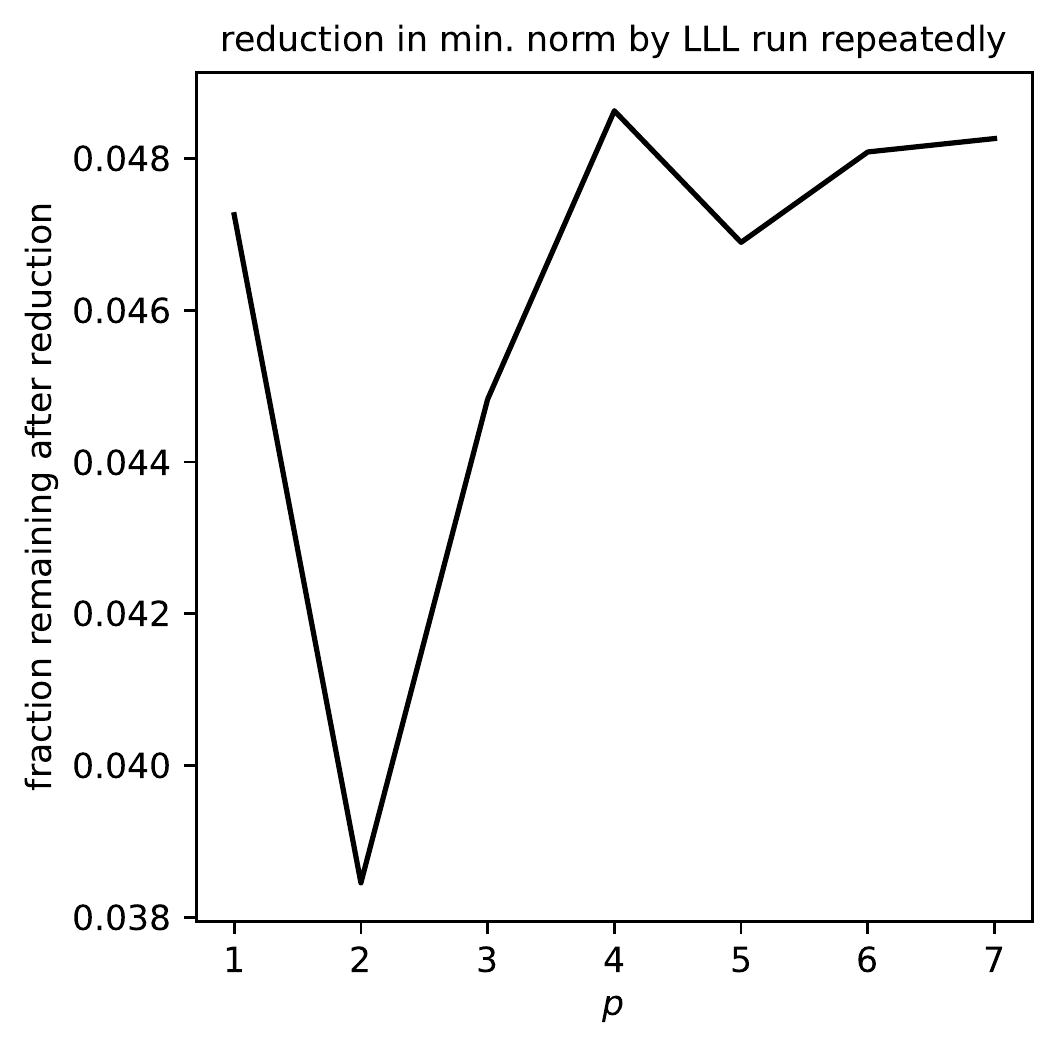}}
{\includegraphics[width=0.495\textwidth]{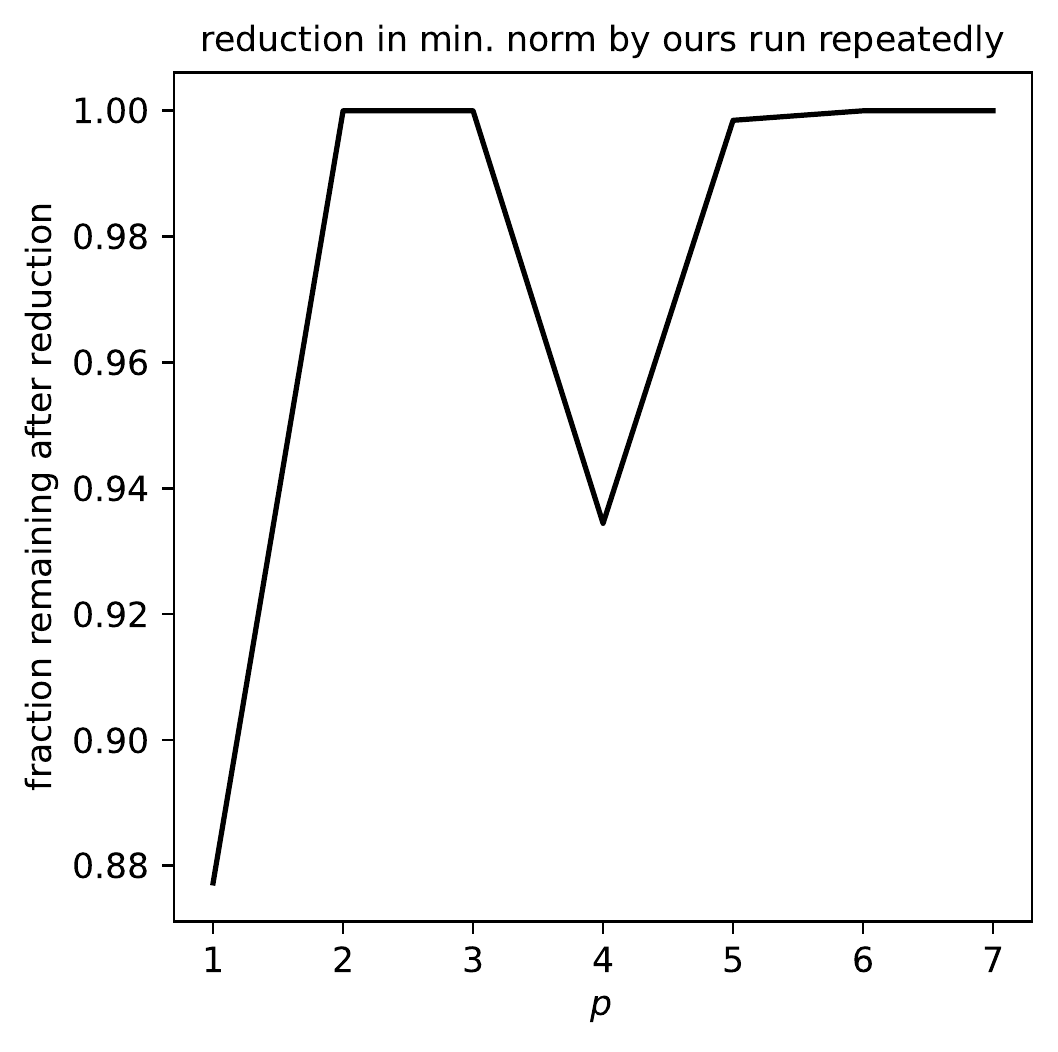}}

\end{centering}
\caption{$\delta = 1-10^{-1}$, $n = 192$, $q = 2^{13} - 1$}
\end{figure}

\begin{figure}
\begin{centering}
{\includegraphics[width=0.495\textwidth]{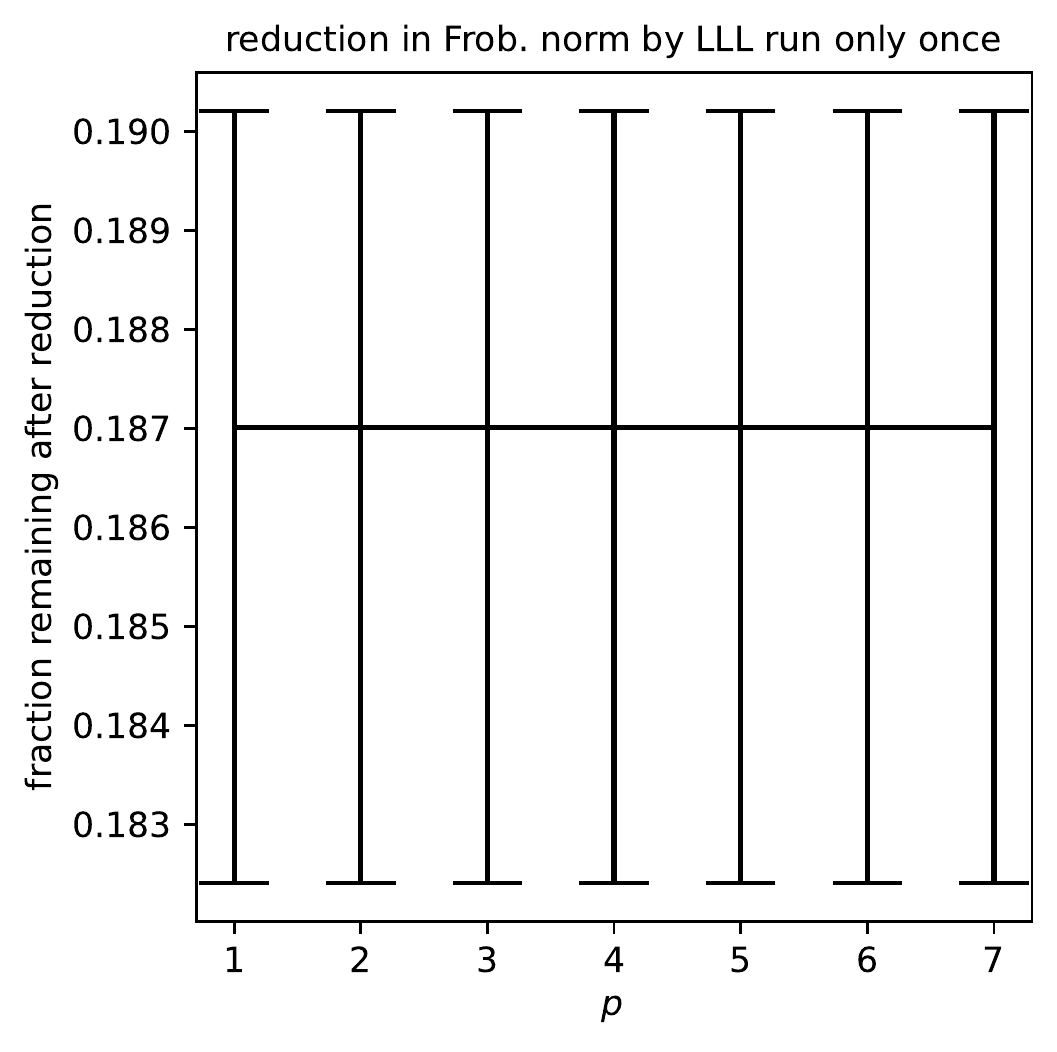}}
{\includegraphics[width=0.495\textwidth]{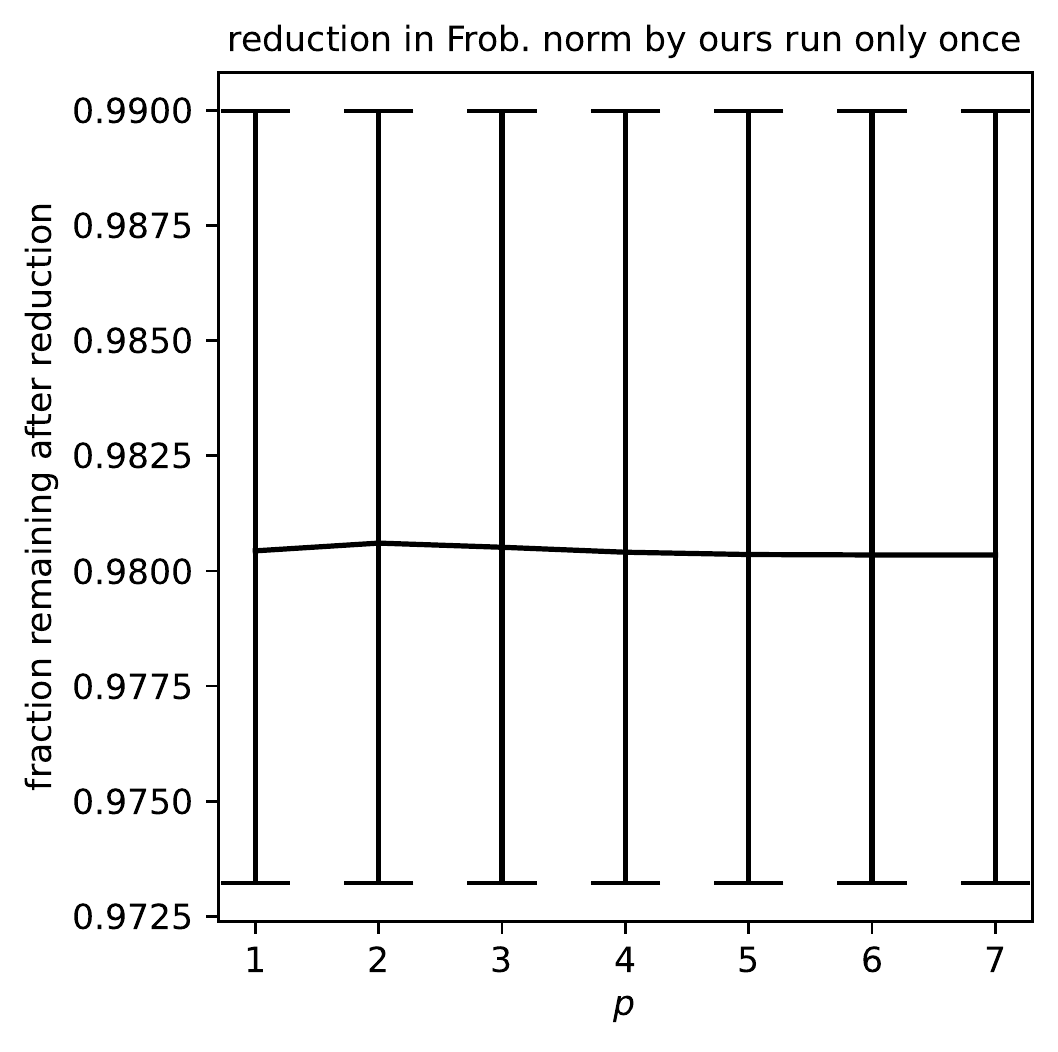}}

{\includegraphics[width=0.495\textwidth]{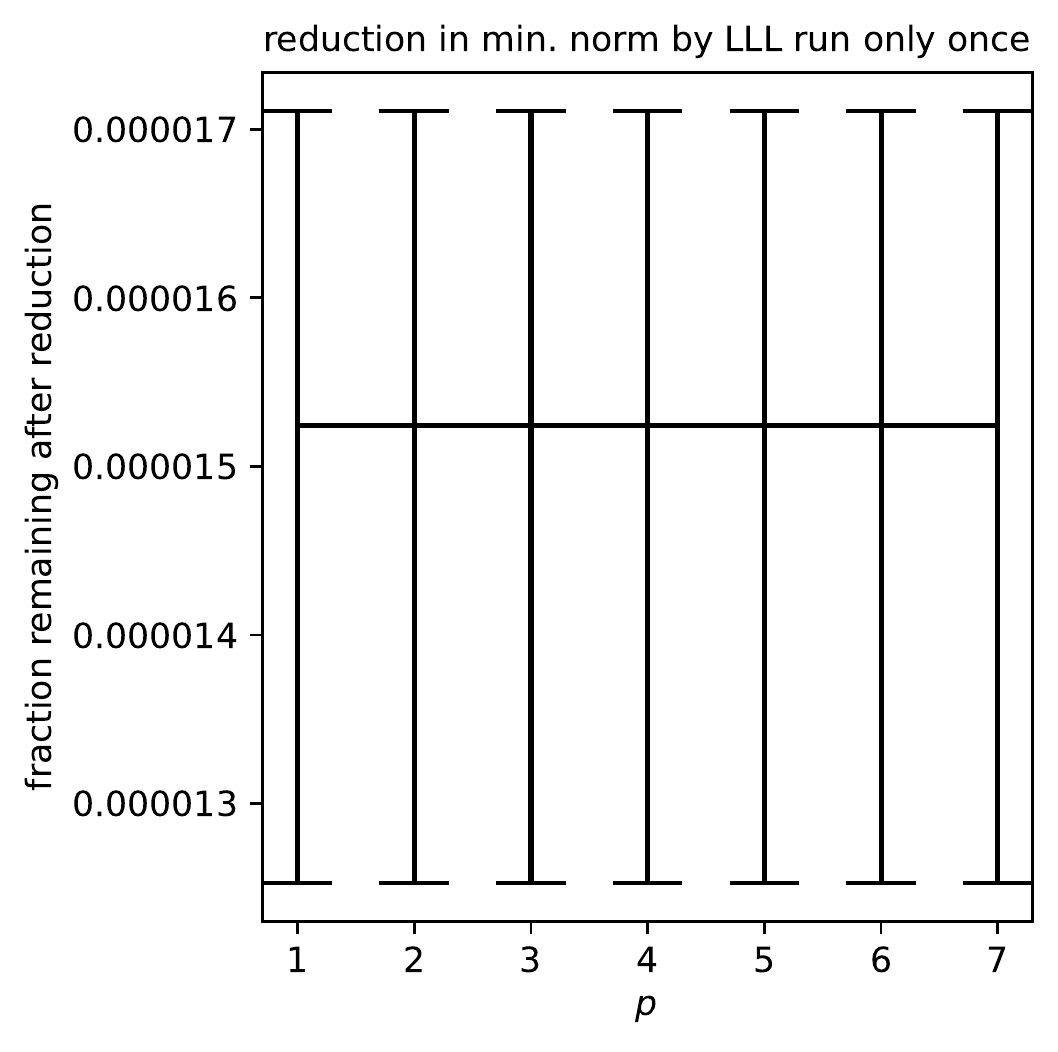}}
{\includegraphics[width=0.495\textwidth]{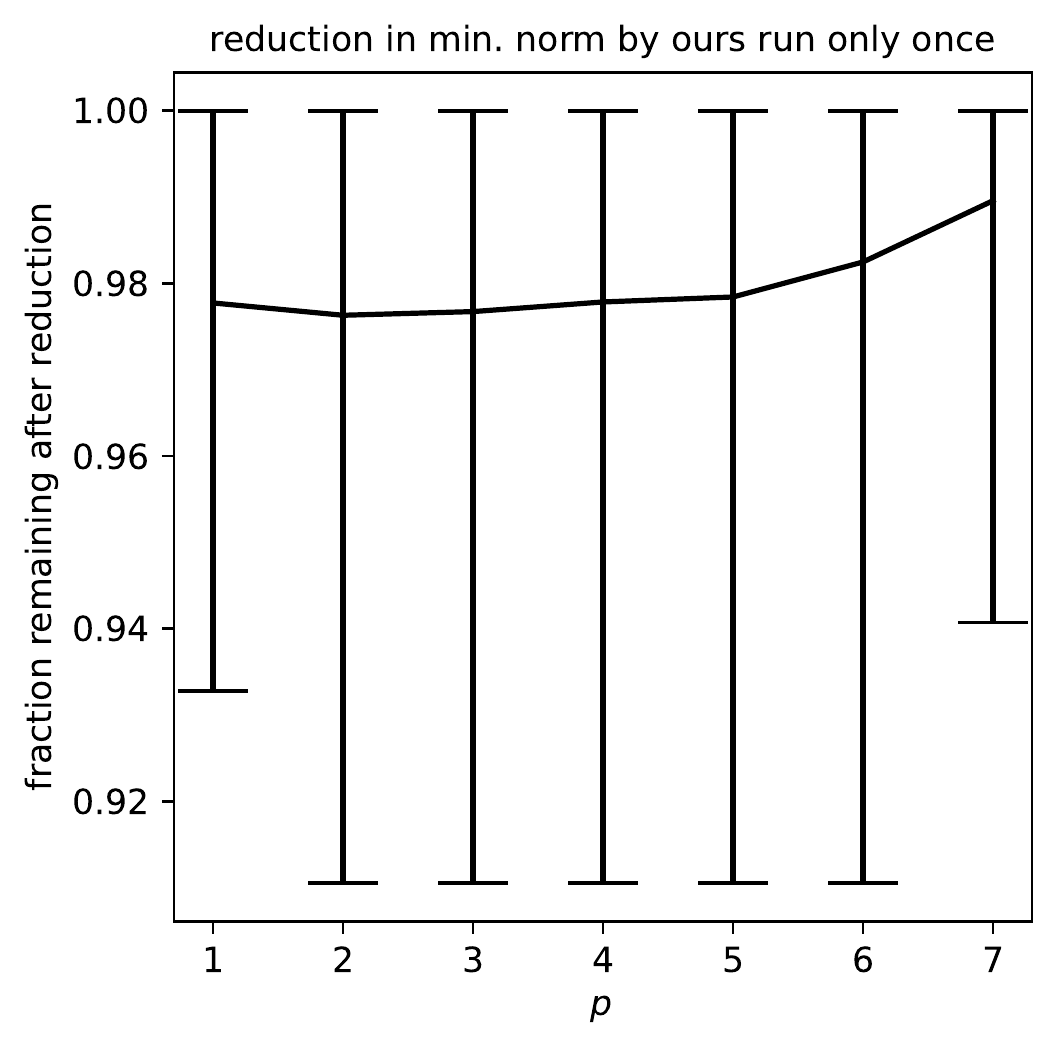}}

\end{centering}
\caption{$\delta = 1-10^{-1}$, $n = 192$, $q = 2^{31} - 1$}
\end{figure}

\begin{figure}
\begin{centering}
{\includegraphics[width=0.495\textwidth]{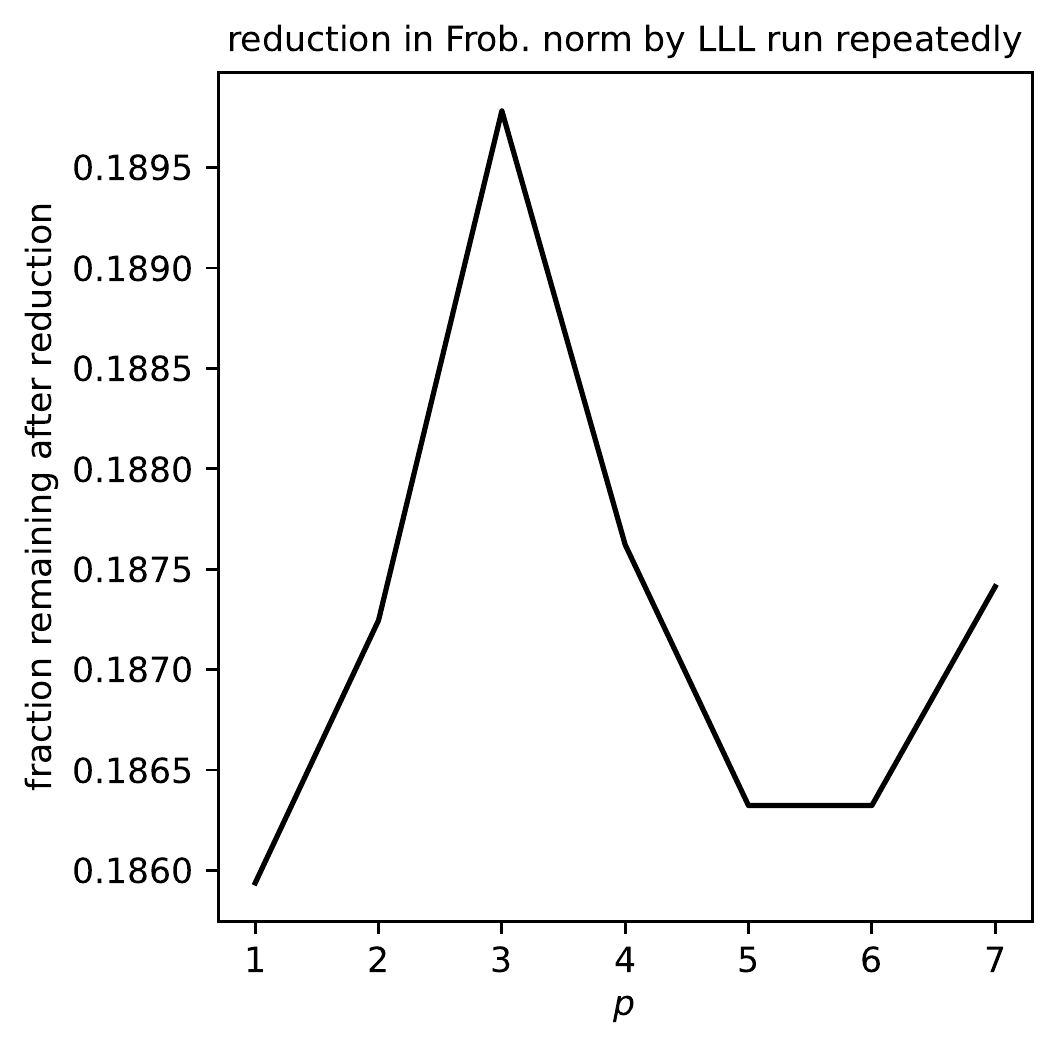}}
{\includegraphics[width=0.495\textwidth]{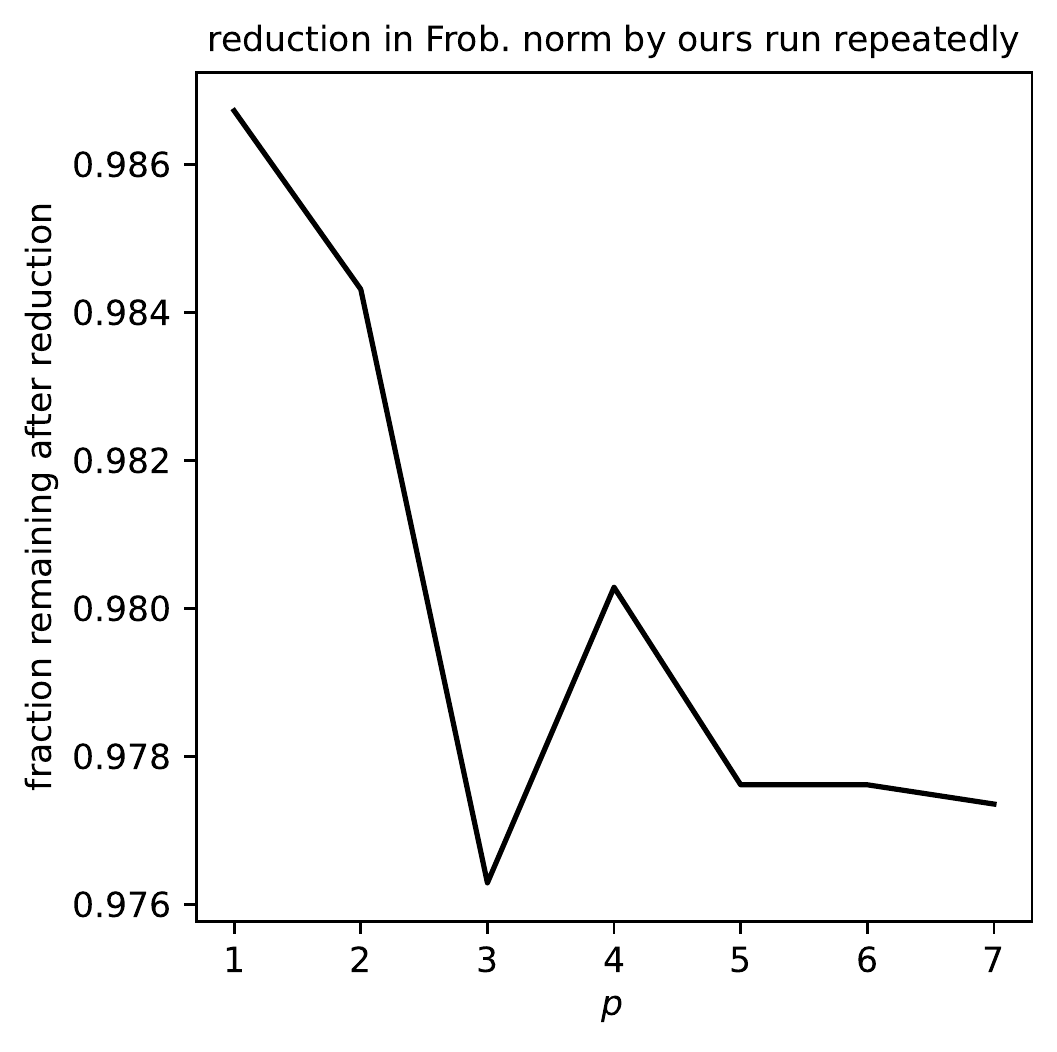}}

{\includegraphics[width=0.495\textwidth]{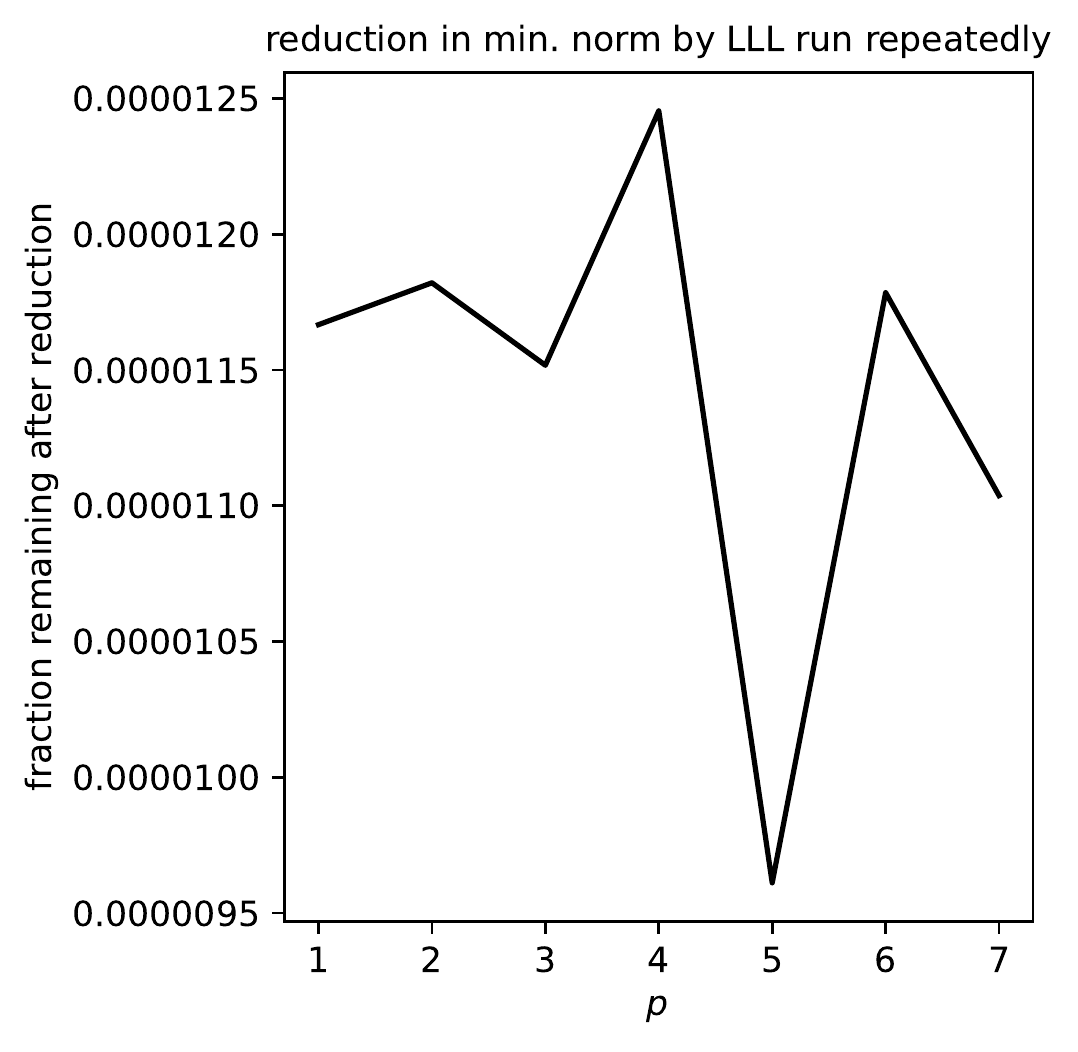}}
{\includegraphics[width=0.495\textwidth]{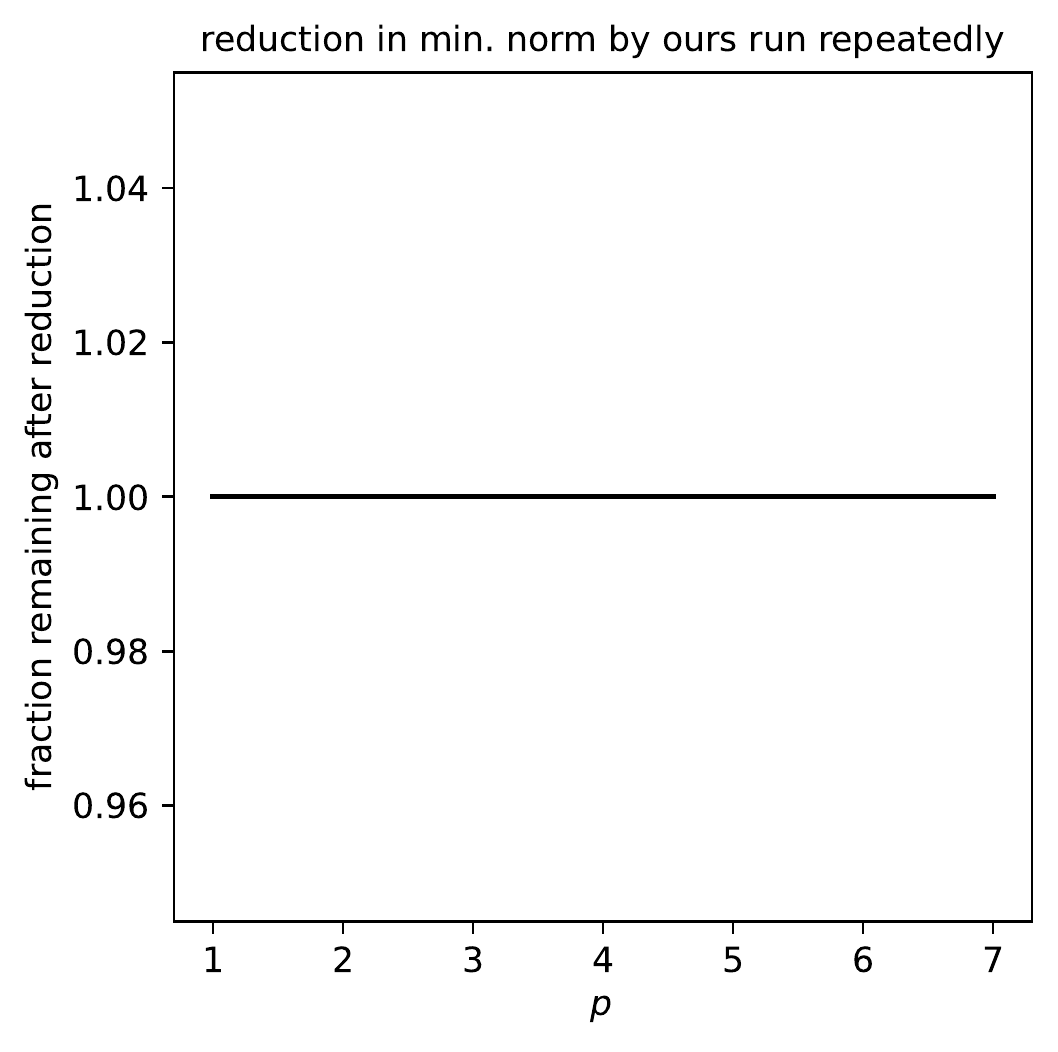}}

\end{centering}
\caption{$\delta = 1-10^{-1}$, $n = 192$, $q = 2^{31} - 1$}
\label{pserr1-1e-1-31}
\end{figure}

\clearpage

\bibliography{lattice}

\begin{thebibliography}{}

\bibitem[\protect\citename{Banaszczyk, }1993]{banaszczyk}
Banaszczyk, William. 1993.
\newblock New bounds in some transference theorems in the geometry of numbers.
\newblock {\em Math. Ann.}, {\bf 296}, 625--635.

\bibitem[\protect\citename{Blackford {\em et~al.}, }2002]{blas}
Blackford, L.~Susan, Demmel, James, Dongarra, Jack, Duff, Iain, Hammarling,
  Sven, Henry, Greg, Heroux, Michael, Kaufman, Linda, Lumsdaine, Andrew,
  Petitet, Antoine, Pozo, Roldan, Remington, Karin, \& Whaley, R.~Clint. 2002.
\newblock An updated set of basic linear algebra subprograms ({BLAS}).
\newblock {\em ACM Trans. Math. Softw.}, {\bf 28}(2), 135--151.

\bibitem[\protect\citename{Bremner, }2012]{bremner}
Bremner, Murray~R. 2012.
\newblock {\em Lattice Basis Reduction: An Introduction to the LLL Algorithm
  and Its Applications}.
\newblock Pure and Applied Mathematics.
\newblock Chapman \& Hall, CRC Press.

\bibitem[\protect\citename{Cassels, }1997]{cassels}
Cassels, John William~Scott. 1997.
\newblock {\em An Introduction to the Geometry of Numbers}.
\newblock Classics in Mathematics.
\newblock Springer.

\bibitem[\protect\citename{Fontein {\em et~al.},
  }2014]{fontein-schneider-wagner}
Fontein, Felix, Schneider, Michael, \& Wagner, Urs. 2014.
\newblock Pot{LLL}: a polynomial time version of {LLL} with deep insertions.
\newblock {\em Des. Codes Cryptogr.}, {\bf 73}, 355--368.

\bibitem[\protect\citename{Lawson {\em et~al.},
  }1979]{lawson-hanson-kincaid-krogh}
Lawson, Charles~L., Hanson, Richard~J., Kincaid, David, \& Krogh, Fred~T. 1979.
\newblock Basic linear algebra subprograms for {FORTRAN} usage.
\newblock {\em ACM Trans. Math. Softw.}, {\bf 5}(3), 308--323.

\bibitem[\protect\citename{Lenstra {\em et~al.}, }1982]{lenstra-lenstra-lovasz}
Lenstra, Arjen~K., Lenstra, Hendrik~W., \& Lov\'asz, L\'aszl\'o. 1982.
\newblock Factoring polynomials with rational coefficients.
\newblock {\em Math. Ann.}, {\bf 261}(4), 515--534.

\bibitem[\protect\citename{Leon {\em et~al.}, }2013]{leon-bjorck-gander}
Leon, Steven~J., Bj\"orck, \AA{}ke, \& Gander, Walter. 2013.
\newblock Gram-{S}chmidt orthogonalization: 100 years and more.
\newblock {\em Numer. Lin. Algebra Appl.}, {\bf 20}(3), 492--532.

\bibitem[\protect\citename{Nguyen \& Vall\'ee, }2010]{nguyen-vallee}
Nguyen, Phong~Q., \& Vall\'ee, Brigitte (eds). 2010.
\newblock {\em The LLL Algorithm: Survey and Applications}.
\newblock Information Security and Cryptography.
\newblock Springer.

\bibitem[\protect\citename{Peikart, }2016]{peikart}
Peikart, Chris. 2016.
\newblock A decade of lattice cryptography.
\newblock {\em Found. Trends Theor. Comput. Sci.}, {\bf 10}(4), 283--424.

\bibitem[\protect\citename{Regev, }2009]{regev}
Regev, Oded. 2009.
\newblock On lattices, learning with errors, random linear codes, and
  cryptography.
\newblock {\em J. ACM}, {\bf 56}(6), 1--40.

\bibitem[\protect\citename{Schnorr \& Euchner, }1994]{schnorr-euchner}
Schnorr, Claus~P., \& Euchner, Martin. 1994.
\newblock Lattice basis reduction: improved practical algorithms and solving
  subset sum problems.
\newblock {\em Math. Program.}, {\bf 66}, 181--199.

\bibitem[\protect\citename{Stehl\'e, }2010]{stehle}
Stehl\'e, Damien. 2010.
\newblock Floating-point {LLL}: theoretical and practical aspects.
\newblock {\em Pages  179--213 of:} Nguyen, Phong~Q., \& Vall\'ee, Brigitte
  (eds), {\em The LLL Algorithm: Survey and Applications}.
\newblock Springer.

\bibitem[\protect\citename{Yasuda \& Yamaguchi, }2019]{yasuda-yamaguchi}
Yasuda, Masaya, \& Yamaguchi, Junpei. 2019.
\newblock A new polynomial-time variant of {LLL} with deep insertions for
  decreasing the squared-sum of Gram-Schmidt lengths.
\newblock {\em Des. Codes Cryptogr.}, {\bf 87}, 2489--2505.

\end{thebibliography}
\bibliographystyle{authordate1}

\end{document}